\pdfoutput=1

\documentclass{article}
\usepackage[utf8]{inputenc}
\usepackage{amsmath,amssymb,amsthm}
\usepackage{caption, subcaption}
\usepackage{geometry}
\usepackage{graphicx}
\usepackage{thmtools, thm-restate}
\usepackage{hyperref}

\usepackage[sort]{natbib} 
\usepackage{float}
\usepackage{lscape}
\usepackage{booktabs}
\usepackage{enumitem}
\usepackage{listings}

\hypersetup{colorlinks,citecolor=blue,urlcolor=blue,linkcolor=blue,bookmarks=false}

\allowdisplaybreaks

\newtheorem{lemma}{Lemma}
\newtheorem{corollary}{Corollary}

\newtheorem{algorithm}{Algorithm}
\theoremstyle{definition}
\newtheorem{definition}{Definition}

\theoremstyle{remark}
\newtheorem{assumption}{Assumption}

\newtheorem{remark}{Remark}

\def\inprob{\stackrel{p}{\rightarrow}}
\def\indist{\rightsquigarrow}
\newcommand\ind{\protect\mathpalette{\protect\independenT}{\perp}}
\def\independenT#1#2{\mathrel{\rlap{$#1#2$}\mkern4mu{#1#2}}}

\def\logit{\text{logit}}
\def\expit{\text{expit}}

\DeclareSymbolFont{bbold}{U}{bbold}{m}{n}
\DeclareSymbolFontAlphabet{\mathbbold}{bbold}
\newcommand{\one}{\mathbbold{1}}
\def\bbE{\mathbb{E}}

\def\bbP{\mathbb{P}}

\def\bbR{\mathbb{R}}

\def\bbV{\mathbb{V}}

\def\cov{\hspace{0.2em}\text{cov}\hspace{0.1em}}

\title{Nonparametric Estimation of Conditional Incremental Effects}
\author{\\ Alec McClean, Zach Branson, Edward H. Kennedy \\ \\
    Department of Statistics \& Data Science \\
    Carnegie Mellon University \\ \\ 
    \texttt{\{alec, zach, edward\} @ stat.cmu.edu} \\
\date{}
    }
\begin{document}
\maketitle

\begin{abstract}
Conditional effect estimation has great scientific and policy importance because interventions may impact subjects differently depending on their characteristics. Most research has focused on estimating the conditional average treatment effect (CATE).  However, identification of the CATE requires all subjects have a non-zero probability of receiving treatment, or positivity, which may be unrealistic in practice. Instead, we propose conditional effects based on incremental propensity score interventions, which are stochastic interventions where the odds of treatment are multiplied by some factor. These effects do not require positivity for identification and can be better suited for modeling scenarios in which people cannot be forced into treatment. We develop a projection estimator and a flexible nonparametric estimator that can each estimate all the conditional effects we propose and derive model-agnostic error guarantees showing both estimators satisfy a form of double robustness. Further, we propose a summary of treatment effect heterogeneity and a test for any effect heterogeneity based on the variance of a conditional derivative effect and derive a nonparametric estimator that also satisfies a form of double robustness. Finally, we demonstrate our estimators by analyzing the effect of intensive care unit admission on mortality using a dataset from the (SPOT)light study.   
\end{abstract}


\section{Introduction} \label{sec:intro}

Estimating causal effects has great scientific and policy importance, and often there is interest in understanding if the effectiveness of a treatment depends on subjects' characteristics.  Conditional, or `heterogeneous', effects describe how a treatment effect varies with subjects' characteristics, and can illustrate qualitatively important phenomena that would be disguised by average effects.  Previous work has focused on estimating the conditional average treatment effect (CATE), which considers the difference between counterfactual mean outcomes when all subjects at some covariate level receive treatment and all subjects receive control (e.g., \cite{kennedy2020towards, kunzel2019metalearners, semenova2020debiased, athey2016recursive, foster2019orthogonal, shalit2017estimating, nie2017quasi}, among others). However, in many contexts researchers cannot force subjects to receive treatment or prevent them from receiving treatment, thereby making the counterfactual interventions behind the CATE unrealistic in practice.  As a concrete example, we will consider the effect of intensive care unit (ICU) admission on mortality for emergency room entrants \cite{keele2019does}.  Typically, the counterfactual interventions where everyone is admitted to the ICU and no one is admitted to the ICU are both practically infeasible because there are a finite number of ICU beds and because hospitals have a duty of care towards sick patients.  Instead, we may be interested in assessing the causal effect of an intervention that could more realistically be implemented in practice, such as an intervention that moderately increases or decreases the probability of admission to the ICU. For example, increasing or decreasing the number of ICU beds would likely increase or decrease the probability of admission for all patients.  Generally, these interventions can best be described with stochastic interventions, which characterize counterfactual outcomes under a shift in the treatment distribution \cite{diaz2012population, haneuse2013estimation, kennedy2019incremental, moore2012causal, young2014identification, zhou2022marginal, diaz2020causal}.  With a binary treatment, this shift can be characterized by an incremental propensity score intervention (``incremental intervention''), which multiplies the odds of treatment by a user-specified factor $\delta$ \cite{bonvini2021incremental, kennedy2019incremental}.  

\medskip

Recent research on stochastic interventions generally and incremental interventions specifically has focused on average effects \cite{kennedy2019incremental, diaz2020causal, wen2021intervention}. In this paper, we consider estimating conditional incremental effects, where we assess to what extent an incremental effect depends on subjects' characteristics, which can uncover treatment effect heterogeneity that is obscured by average effects. Furthermore, as well as corresponding to more realistic interventions, there are two additional advantages to considering conditional incremental effects instead of the CATE. First, when some subjects are estimated to be very likely or unlikely to receive treatment, then, without strong parametric modeling assumptions, it can be difficult to estimate the average treatment effect or the CATE, in the sense that variance estimates are large and confidence intervals are wide \cite{westreich2010positivity}. However, in this situation, incremental effects can still be estimated with narrow confidence intervals that provide precise results.  The reason is that identification and estimation of incremental effects does not rely on positivity, because the magnitude of the counterfactual intervention is allowed to vary for subjects with different probabilities of receiving treatment.

\medskip

A second advantage of using conditional incremental effects instead of the CATE is the ability to describe a continuum of policies between treating all subjects and treating none, where the interventions behind the CATE are special cases at each end of the continuum. A researcher might presume that stochastic effects follow a roughly linear relationship from one end of the continuum to the other, with the slope of the line matching the sign of the CATE.  As discussed in Remark \ref{rem:vary} in Section \ref{sec:setup}, this assumption is reasonable when conditioning on all the covariates, as the conditional incremental effect curve must be monotonic in the incremental parameter $\delta$ and so its slope will match the sign of the CATE.  However, most analyses, including our ICU data analysis in Section \ref{sec:data_analysis}, condition on only a few covariates of interest, allowing for the possibility of other incremental effect curves.  For example, consider Figure \ref{fig:motivation} -- a preview of the real data analysis in Section \ref{sec:data_analysis} -- which shows conditional incremental effect curves for several ICNARC scores (a measure of mortality risk).  The x-axis represents the incremental intervention parameter, where $\delta = 1$ corresponds to no intervention while $\delta > 1$ and $\delta < 1$ correspond to increasing and decreasing the likelihood, respectively, that patients are admitted to the ICU, while the y-axis shows estimated mortality rate.  The curves illustrate that the estimated counterfactual mortality rate is higher at $\delta = 5$ than at $\delta = 0.2$, in agreement with prior research indicating that admitting everyone to the ICU is harmful compared to admitting no one \cite{keele2019does}.  However, the full curves suggest a different practical implication, since $\delta = 1$ corresponds to the lowest estimated mortality rate, suggesting that maintaining the status quo is preferable to sending no one to the ICU.  Without considering stochastic effects that can be evaluated over a continuum of interventions, we would be unaware of these nuances.

\begin{figure}[h]
	\centering
	\includegraphics[width=.7\linewidth]{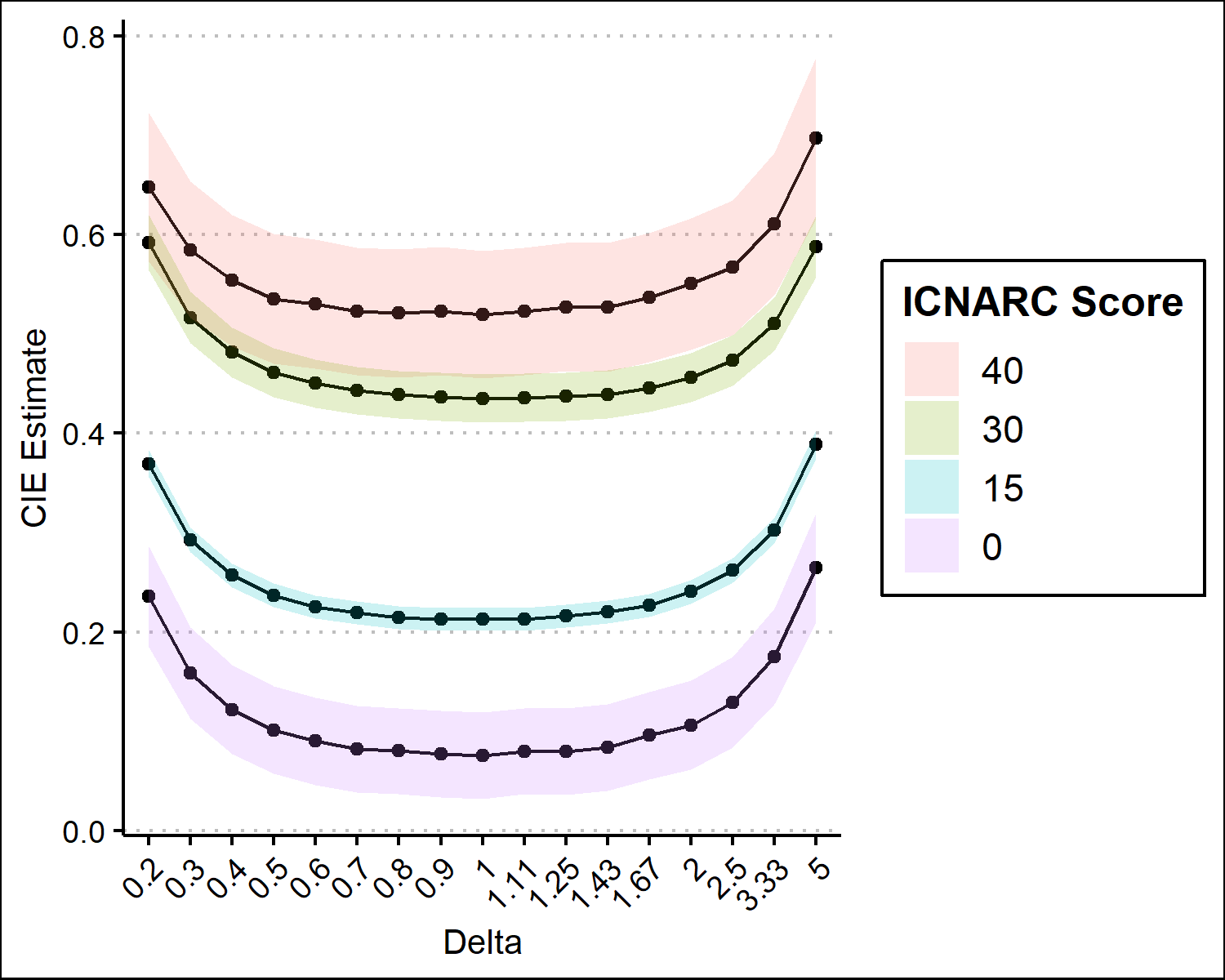}
	\caption{Conditional incremental effect curves for select ICNARC scores. The x-axis represents the incremental intervention parameter $\delta$, where $\delta = 1$ corresponds to no intervention, and $\delta > 1$ and $\delta < 1$ correspond to increasing and decreasing the likelihood of admission to the ICU, respectively. The y-axis shows estimated mortality rate. The curves depict the estimated conditional incremental effect for different ICNARC scores, which measure mortality risk. Our analysis shows that sending very few people to the ICU ($\delta = 0.2$) is preferable to sending many people to the ICU ($\delta = 5$) for each ICNARC score. However, we also estimate that the status quo ($\delta =1$) has the lowest mortality rate for each ICNARC score.}
	\label{fig:motivation}
\end{figure}


\subsection{Contribution and Structure}

Motivated by these observations, in this paper we describe how to estimate conditional causal effects for incremental interventions and illustrate how these effects facilitate a more nuanced understanding of treatment effect heterogeneity than the usual CATE.  We focus on incremental interventions for two reasons.  First, the incremental intervention has an intuitive parameterization for binary treatment since it corresponds to multiplying the odds of treatment by some factor.  Second, the intervention demonstrates favorable properties because it is anchored at the observed treatment distribution and considers a smooth shift from the observed distribution.  For example, identifying effects with this intervention does not require the positivity assumption that the probability of treatment is bounded between zero and one for all subjects, which is required for identifying the CATE.  This allows estimation of conditional incremental effects to still be precise even in the face of positivity violations, unlike estimation of the CATE.

\medskip

We consider three conditional effects in this paper.  First, we describe the conditional incremental effect (CIE), which is the conditional analog to the average incremental effect. As shown in Figure \ref{fig:motivation}, the CIE is described by a curve for each covariate value; this makes quantifying treatment effect heterogeneity challenging, because we have to consider how much these curves vary across covariate values.  As a preliminary extension of the CIE, we describe the conditional incremental \emph{contrast} effect (CICE), which considers a contrast between two incremental interventions, and is the incremental analog to the CATE.  The CICE can enable better understanding of treatment effect heterogeneity than the CIE, but it requires specifying two incremental $\delta$ parameters, and it may not immediately be clear which parameter values would be of most interest in a particular application.  Therefore, we propose the conditional incremental \emph{derivative} effect (CIDE), which corresponds to the change in the CIE under an infinitesimal shift of the treatment distribution.  We find that the CIDE is particularly useful for quantifying treatment effect heterogeneity for incremental interventions because it allows the researcher to examine the spectrum of interventions like in Figure \ref{fig:motivation} and also construct estimators and tests to quantify treatment effect heterogeneity, as we discuss further below and in Section \ref{sec:heterogeneity}.

\medskip

For the three conditional effects, we propose two estimators.  Our first estimator, the Projection-Learner, estimates the projection of the true conditional effect onto a finite dimensional model.  This added structure allows us to re-frame the estimator as the solution to a moment condition, and derive an efficient influence function.  Utilizing the properties of efficient influence functions, we provide double robust style error guarantees for the Projection-Learner, and show that its bias scales as a product of errors of the nuisance function estimators (in this paper, the nuisance functions are the propensity score and the outcome regression, and are defined in Section \ref{sec:setup}).  As a result, the Projection-Learner can achieve parametric efficiency even when the nuisance functions are estimated nonparametrically.  Our second conditional effect estimator, the I-DR-Learner, is a two-stage meta-learner that extends the DR-Learner from \cite{kennedy2020towards} to incremental effects.  For the I-DR-Learner, the first stage estimates the efficient influence function values for the relevant average effect and the second stage regresses those values against the conditioning covariates. We establish when the I-DR-Learner exhibits double robust style guarantees; in particular, the conditional effect must lie in a certain infinite dimensional function class, and the second stage regression must satisfy a form of stability.  In this case, we demonstrate that the I-DR-Learner can attain oracle efficiency when the nuisance functions are estimated nonparametrically.  Therefore, the I-DR-Learner cannot obtain parametric efficiency like the Projection-Learner, but it can estimate a larger class of true conditional effect curves with oracle efficiency.

\medskip

Both the Projection-Learner and the I-DR-Learner can be used to estimate conditional effect curves across variables of interest.  A natural question is whether there is any treatment effect heterogeneity across the curve. Thus, researchers may also be interested in a one-dimensional summary of effect heterogeneity and a corresponding test for any effect heterogeneity.  Therefore, we also propose a fourth effect, the variance of the conditional incremental derivative effect (V-CIDE), which can be used to estimate the degree of effect heterogeneity and test for any effect heterogeneity.  For the V-CIDE, we derive a novel double robust style estimator based on its efficient influence function, illustrate that our estimator attains parametric efficiency under weak conditions on the nuisance function estimators, and derive a corresponding test for any effect heterogeneity.

\medskip

The structure of the paper is as follows.  In Section \ref{sec:notation} we define relevant notation.  In Section \ref{sec:setup}, we define the data setup and different estimands of interest, state the causal assumptions required for identification, and establish identification results for our conditional effects.  In Section \ref{sec:conditional_effects}, we outline the Projection-Learner and I-DR-Learner and demonstrate their convergence properties in Sections \ref{sec:proj-learner} and \ref{sec:i-dr-learner} respectively.  In Section \ref{sec:heterogeneity}, we outline a nonparametric estimator for the V-CIDE, demonstrate its convergence properties, and describe methods for inference. In Section \ref{sec:data_analysis}, we analyze data on ICU admission from the (SPOT)light prospective cohort study. We estimate that increasing or decreasing subjects' odds of attending the ICU would adversely affect mortality rates, suggesting that the status quo is preferable. Importantly, this differs from what would be concluded for CATE estimation, which for this application would not be reliable, because there are positivity violations.  Using our test, we do not find evidence that there is treatment effect heterogeneity.  Finally, in Section \ref{sec:discussion} we conclude and discuss future extensions of this research.

\subsection{Notation} \label{sec:notation}

We use $\bbE$ for expectation and $\bbV$ for variance.  We use $\bbP_n(f) = \bbP_n \{ f(Z) \} = \frac1n \sum_{i=1}^{n} f (Z_i)$ as a shorthand for sample averages and $\bbV_n \{ f(Z) \} = \frac{1}{n-1} \sum_{i=1}^{n} \left\{ f(Z_i) - \frac1n \sum_{j=1}^{n} f(Z_j) \right\}^2$ as shorthand for the sample variance.  When $x \in \bbR^d$ we let $\lVert x \rVert^2 = \sum_{j=1}^{d} x_j^2$ denote the squared Euclidean norm, and for generic possibly random functions $f$ we let $\lVert f \rVert^2 = \int_{\mathcal{Z}} f(z)^2 d\bbP(z)$ denote the squared $\ell_2 (\bbP)$ norm.  We use the notation $a \lesssim b$ to mean $a \leq Cb$ for some constant $C$, and $a \asymp b$ to mean $cb \leq a \leq Cb$ for some constants $c$ and $C$, so that $a \lesssim b$ and $b \lesssim a$.  We use $\indist$ to denote convergence in distribution and $\inprob$ for convergence in probability.  We use $\widehat \bbE_n$ to denote the predicted regression function estimate from $n$ samples (e.g., if we considered a regression of $Y$ against $X$, then $\widehat \bbE_n (Y \mid X = x)$ is the estimated regression function of $Y$ against $X$ at $X = x$ using $n$ data points $\{ (X_i, Y_i) \}_{i=1}^{n}$).   We use the set notation $A \setminus B$ to indicate ``$A$ and not $B$''.


\section{Estimands and Identification Results for Conditional Incremental Effects} \label{sec:setup}

In this section we describe estimands for incremental effects, and establish assumptions for identifying these effects.  Assume we observe $\{Z_1, ..., Z_n \}$ with $Z_i \stackrel{iid}{\sim} \mathcal{P}$ where $Z = (X, A, Y)$, $X \in \bbR^d$ are covariates, $A \in \{0, 1\}$ is treatment status, and $Y \in \bbR$ is an outcome.  We define potential outcomes $Y^a$ as the outcome that would be observed when treatment $A = a$. 

\medskip

Much of the causal literature has focused on estimating the average treatment effect (ATE) and conditional ATE (CATE), defined as 
\begin{align}
	\text{ATE: }&\psi_{ate} = \bbE (Y^1 - Y^0),\text{ and} \\
	\text{CATE: } &\tau_{cate}(x) = \bbE (Y^1 - Y^0 \mid X = x).
\end{align}
\noindent  To identify the ATE and the CATE, three causal assumptions are required:
\begin{assumption}\label{asmp:cons} 
	\textbf{(Consistency)}, $Y = Y^{a}$ if $A = a$
\end{assumption} 

\begin{assumption}\label{asmp:exch} 
	\textbf{ (Exchangeability)}, $A \ind Y^{a} \mid X$
\end{assumption} 

\begin{assumption} \label{asmp:positivity} 
	\textbf{(Positivity)}, There exists $\varepsilon > 0$ such that $\bbP \left\{ \varepsilon \leq \bbP \left( A = a \mid X \right) \leq 1 - \varepsilon \right\} = 1$ for $a \in \{0, 1\}$ and all $X$
\end{assumption}

Consistency says that if an individual takes treatment $a$, we observe the potential outcome under that treatment regime.  By contrast, consistency would be violated if, for example, there were interference between subjects, such that one subject's treatment status affected another's outcome. Exchangeability says that treatment is effectively randomized within covariate strata, in the sense that treatment is independent of subjects’ potential outcomes after conditioning on covariates.  Positivity says that all subjects have a non-zero chance of receiving treatment or control, and positivity may be unrealistic in practice.  Although positivity is required to identify the ATE and the CATE, as we show next, only Assumptions \ref{asmp:cons} and \ref{asmp:exch} are be required to identify conditional incremental effects.

\subsection{Incremental Propensity Score Interventions}

The incremental intervention corresponds to multiplying each individual's odds of treatment by a user-specified parameter $\delta$.  We define the propensity score, the probability that an individual receives treatment, as $\pi (X) = \bbP(A = 1 \mid X)$, and then the shifted propensity score under an incremental intervention is defined as
\begin{equation} \label{eq:inc_int}
	q \{ \pi(X); \delta \} = \frac{\delta \pi(X)}{\delta \pi(X) + 1 - \pi(X)}.
\end{equation}
Then, the average incremental effect is 
\begin{equation} \label{eq:inc_int_gen}
	\bbE \left( Y^{Q_\delta} \right)
\end{equation}
where $Q_\delta \text{ is drawn from a Bernoulli distribution with parameter } q \{ \pi(X); \delta \}.$ Unlike ATE-style interventions, the incremental intervention is stochastic because it does not deterministically assign subjects to treatment or control - rather, it shifts their propensity score. The incremental intervention also corresponds to multiplying the odds of treatment by $\delta$ since $\delta = \frac{q \{ \pi(X); \delta \} / [ 1 - q \{ \pi(X); \delta \} ]}{ \pi(X) / \{ 1 - \pi(X) \}}$.  Incremental interventions were first proposed in \cite{kennedy2019incremental}, with double robust style estimators for average (possibly time-varying) incremental effects.  The analysis of average effects has been extended to censored data \cite{kim2021incremental} and used for estimating the effect of aspirin on the incidence of pregancy \cite{rudolph2022estimation}; a review is provided in \cite{bonvini2021incremental}.  

\medskip

While this intervention is not prescriptive - it is unlikely a hospital would seriously consider an intervention where patients are admitted to the ICU by draws from a Bernoulli distribution - it can be useful for describing interventions that could be implemented in practice.  For example, a researcher may want to know what would happen if a hospital changed its admission criteria to make it slightly more likely that emergency room entrants were admitted to the ICU. This cannot be described by the CATE, whereas an incremental intervention with $\delta > 1$ could appropriately describe this counterfactual question. Further, a spectrum of $\delta$ could appropriately describe the range of admission criteria changes that a hospital may implement in practice.

\medskip

The incremental intervention is also dynamic in the sense that the intervention changes with $X$ if $\pi(X)$ changes with $X$. This occurs because the intervention is constant on the odds ratio scale rather than the unit scale.  For example, if $\delta = 2$ and the propensity scores for two covariate values are $\{ \pi(X = x_1), \pi(X = x_2) \} = \left\{ 0.25, 0.5 \right\}$, then the intervention propensity scores are $\{ q \{ \pi(X = x_1) = 0.25; \delta = 2 \}, q\{ \pi(X = x_2) = 0.5; \delta = 2  \} \} = \{ 0.4, 0.66 \}$.  Therefore, the propensity score increases by $0.15$ when $\pi(x_1) = 0.25$ and increases by $0.16$ when $\pi(x_2) = 0.5$.  However, although the interventions are dynamic in the sense just outlined, they are not dynamic in the sense that the user-specified parameter $\delta$ changes with $X$.  We leave this as an avenue for future exploration. If $\delta$ were allowed to vary with $X$, a natural question then might be: what is the ``optimal'' choice of $\delta$ at a particular value $X = x$?  As in the deterministic intervention literature, finding an optimal intervention could fruitfully build on the conditional effect estimators proposed in this paper \cite{murphy2003optimal, chakraborty2014dynamic}.

\medskip

Other stochastic interventions have also been considered in the literature, such as modified treatment policies, which shift a continuous treatment by a specified amount \cite{haneuse2013estimation, diaz2012population}; dynamic interventions that depend on some time-varying information about subjects \cite{young2014identification, taubman2009intervening}; and exponential tilts, which shift a discrete but possibly multi-valued treatment distribution \cite{diaz2020causal}.  The incremental effect can also be interpreted as an exponential tilt.  \cite{wen2021intervention} recently proposed a similar intervention to the incremental intervention, but their intervention is parameterized as a shift of the risk ratio $\frac{q\{ \pi(X); \delta \}}{\pi(X)}$, rather than the odds ratio.

\subsection{Conditional Incremental Effects}

Now we'll consider conditional incremental effects.  We denote $V \subseteq X$ as either one or a set of covariates, and define the conditional incremental effect (CIE) as the counterfactual mean under an incremental intervention conditional on covariates $V$,
\begin{equation} \label{eq:cie}
	\text{CIE: } \tau_{cie} (v; \delta) = \bbE \left( Y^{Q_\delta} \mid V = v \right).
\end{equation}
The following proposition establishes that the CIE is identifiable as a function of the observed data distribution.
\begin{restatable}{proposition}{proponetime}\label{prop:id_t1}
	Let $Q_\delta$ denote the incremental intervention defined in eq. \eqref{eq:inc_int_gen}.   Under Assumptions \ref{asmp:cons}-\ref{asmp:exch}, then the mean counterfactual outcome given covariates $V = v$ is identified by
	\begin{equation} \label{eq:cie_id}
		\bbE \left( Y^{Q_\delta} \mid V = v \right) = \bbE \left[ \frac{ \delta \pi(X) \mu(1, X) + \{ 1 - \pi(X) \} \mu(0, X)}{\delta \pi(X) + 1 - \pi(X)} \ \bigg| \ V = v \right]
	\end{equation}
	where $\mu(a, x) = \bbE (Y \mid A = a, X = x)$.
\end{restatable}
We leave all proofs to the appendix.  Proposition \ref{prop:id_t1} is a straightforward corollary of Corollary 1 in \cite{kennedy2019incremental}, and shows that the CIE is identified by a linear combination of the regression functions $\mu(A, X)$ where the weights depend on the probabilities of receiving treatment and control under the incremental intervention.

\medskip

The CIE does not consider a contrast between two interventions, and so it does not immediately describe treatment effect heterogeneity.  In this sense, it is similar to the conditional counterfactual mean under treatment, $\bbE (Y^1 \mid V = v)$.  As a first approach to understanding treatment effect heterogeneity, we define a second estimand, the conditional incremental constrast effect (CICE), which considers the difference between two incremental effects,
\begin{equation} \label{eq:cice}
	\text{CICE: } \tau_{cice} (v; \delta_u, \delta_l) \equiv \bbE \left( Y^{Q_{\delta_u}} - Y^{Q_{\delta_l}} \mid V = v \right).
\end{equation}
The CICE tells us the difference (conditional at $V = v$) between the average outcomes if we multiply the odds of treatment by $\delta_u$ and if we multiply the odds of treatment by $\delta_l$.  We can understand treatment effect heterogeneity by looking at how the CICE changes with $V$.  The CICE is readily comparable to the CATE since both consider contrasts between two interventions.  In fact, if positivity is satisfied in Assumption \ref{asmp:positivity}, then the CICE approaches the CATE as $\delta_u \to \infty$ and $\delta_l \to 0$ since $\lim_{\delta_u \to \infty, \delta_l \to 0} \tau_{cice} (v; \delta_u, \delta_l) = \bbE \left(Y^1 - Y^0 \mid V = v \right).$  Identification for the CICE follows by Proposition \ref{prop:id_t1} and linearity of expectation since $\tau_{cice} (v; \delta_u, \delta_l) = \tau_{cie} (v; \delta_u) - \tau_{cie} (v; \delta_l)$.   

\subsection{Derivative Effects}

A limitation of the CICE is that it requires specifying two parameters, $\delta_u$ and $\delta_l$, and it may not immediately be clear which parameter values would be of most interest in a particular application.  Instead, we can consider a derivative effect, which describes the change in counterfactual outcomes with an infinitesimally small change in the treatment distribution.  To ease exposition, we re-parametrize the average incremental effect with $t$ instead of $\delta$, and define the average derivative effect with respect to $t$, and evaluated at $\delta$, as
$$
\frac{\partial}{\partial t} \bbE \left( Y^{Q_t} \right) \Big|_{t = \delta}
$$ 
and the associated conditional incremental derivative effect as
\begin{equation} \label{eq:cide}
	\text{CIDE: } \tau_{cide} (v; \delta) = \frac{\partial}{\partial t} \tau_{cie} (v; t) \Big|_{t = \delta}.
\end{equation}
The CIDE demonstrates treatment effect heterogeneity if it varies across $v$.  Thus, it can illustrate effect heterogeneity across a continuum of policies if it is evaluated at several values for $\delta$.  Under suitable regularity conditions such that the Leibniz integral rule to exchange differentiation and integration applies, the CIDE is identified according to the following result.
\begin{restatable}{proposition}{propderiv}\label{prop:deriv}
	Let $Q_\delta$ denote the incremental intervention defined in eq. \eqref{eq:inc_int_gen}.   Under Assumptions \ref{asmp:cons} and \ref{asmp:exch}, the CIDE is identified by
	\begin{equation} \label{eq:cide_id}
		\tau_{cide} (v; \delta) = \bbE \left( \left[ \frac{\pi(X) \{ 1 - \pi(X) \}}{ \{ \delta \pi(X) + 1 - \pi(X) \}^2} \right] \cdot \{ \mu(1, X) - \mu(0, X) \} \ \bigg| \ V = v \right)
	\end{equation}
	where $\mu(a, x) = \bbE (Y \mid A = a, X = x)$.
\end{restatable}
Proposition \ref{prop:deriv} shows that the CIDE is a weighted average of the difference in mean outcomes under treatment and control, where the weights depend on the propensity scores and the incremental propensity scores.  

\begin{remark} \label{rem:vary}
	When $V = X$, the CIDE (and, by extension, the CIE) must be monotonic across $\delta$.  This is clear because $\frac{\pi(x) \{ 1 - \pi(x) \}}{\{ \delta \pi(x) + 1 - \pi(x) \}^2}$ is always non-negative, while $\mu(1, x) - \mu(0, x)$ does not change with $\delta$.  However, when $V \subset X$, the CIDE and the CIE need not be monotonic across $\delta$.  If they are not monotonic, this indicates that $\mu(1, X) - \mu(0, X)$ changes sign across $X \setminus V$.
\end{remark}

We also propose a one-dimensional functional to assess treatment effect heterogeneity.  We consider the variance of the conditional incremental derivative effect (V-CIDE), defined as 
\begin{equation} \label{eq:vcide}
	\text{V-CIDE: } \bbV \left\{  \tau_{cide} (V; \delta) \right\}.
\end{equation}
When this variance equals zero, it implies that the CIDE is constant over $V$, and thus there is no treatment effect heterogeneity.  As  before, the V-CIDE depends on $\delta$, so it can be estimated over a grid of $\delta$ to evaluate treatment effect heterogeneity over a continuum of policies.  By Proposition \ref{prop:deriv}, the V-CIDE is identified by
$$
\bbV \left\{ \tau_{cide}(V; \delta) \right\} = \bbV \left\{ \bbE \left( \left[ \frac{\pi(X) \{ 1 - \pi(X) \}}{ \{ \delta \pi(X) + 1 - \pi(X) \}^2} \right] \cdot \{\mu(1, X) - \mu(0, X) \ \bigg| \ V = v \right) \right\}
$$
and when $V = X$, this simplifies to
$$
\bbV \left\{ \tau_{cide}(X; \delta) \right\} = \bbV \left(  \left[ \frac{\pi(X) \{ 1 - \pi(X) \}}{ \{ \delta \pi(X) + 1 - \pi(X) \}^2} \right] \cdot \{\mu(1, X) - \mu(0, X) \} \right).
$$
In Section \ref{sec:heterogeneity}, we will derive an efficient estimator for the V-CIDE and a propose a test for whether there is any effect heterogeneity at all.  In the next section, we will derive efficient estimators for the CIE, the CICE, and the CIDE.

\section{Estimating conditional incremental effects} \label{sec:conditional_effects}

The identification results in Propositions \ref{prop:id_t1} and \ref{prop:deriv} suggest straightforward ``plug-in'' estimators for the conditional effects.  Given estimates for $\widehat \pi(X), \widehat \mu(a, X)$ and $\widehat \bbP(X \mid V = v)$, an estimator can be constructed by plugging these estimates into the identification formulae in equations \eqref{eq:cie_id} and \eqref{eq:cide_id}.  If models for the nuisance functions are parametric and correctly specified, this approach can be optimal as the plug-in estimator will converge to a normal distribution at a $n^{-1/2}$-rate.  However, if the parametric models are misspecified, then the plug-in estimator will be biased \citep{vansteelandt2012model}.   Given this, it is tempting to use flexible nonparametric models to estimate the nuisance functions, in order to alleviate issues of model misspecification.  However, in this case, typically the plug-in estimator will inherit the slow rate of convergence for the nonparametric models.

\medskip

This motivates estimators based on semiparametric efficiency theory \citep{bickel1993efficient, van2000asymptotic, tsiatis2006semiparametric, van2002semiparametric, van2003unified}.  The first-order bias of the nonparametric plug-in can be characterized by the efficient influence function of the estimand, which can be thought of as the first derivative in a Von Mises expansion of the estimand \citep{von1947asymptotic}.  Thus, a natural approach is to estimate the efficient influence function and subtract this estimate from the nonparametric plug-in estimate in order to ``de-bias'' the plug-in.  A benefit of estimators based on the efficient influence function is that their bias is a second-order product of errors of the nuisance function estimators, such that the estimator can achieve $n^{-1/2}$ efficiency even when the nuisance functions are estimated at slower nonparametric rates \citep{van2003unified, chernozhukov2018double, kennedy2022semiparametric}.  We consider two estimators that utilize efficient influence functions; as a result, they both exhibit double robust style error guarantees.

\medskip

Our first estimator, the Projection-Learner, targets the projection of the true conditional effect onto a finite dimensional working model. Projection estimators have a long history in statistics \citep{huber1967the, white1980using, buja2019modelsi, buja2019modelsii} and causal inference \citep{ neugebauer2007nonparametric, NBERw24678, semenova2020debiased, kennedy2021semiparametric, cuellar2020non}.  This added structure  us to re-frame the estimator as the solution to a moment condition, and derive an efficient influence
function.  We show that the Projection-Learner exhibits a version of double robustness, and attains parametric efficiency under weak model-agnostic $n^{-1/4}$ conditions on the nuisance function estimators, which are achievable for nonparametric estimators under suitable smoothness or sparsity. 

\medskip

Our second estimator, the I-DR-Learner (inspired by the ``DR-Learner'' in \cite{kennedy2020towards}), instead targets the true conditional effect.   The I-DR-Learner is an estimation procedure that, like many recent CATE estimation approaches, tries to estimate the true conditional effect as flexibly as possible \citep{athey2016recursive, foster2019orthogonal, hahn2020bayesian, kunzel2019metalearners, nie2017quasi, shalit2017estimating, zimmert2019nonparametric, kennedy2020towards}.  Without any further assumptions, no efficient influence function exists for the true conditional effect because it is not pathwise differentiable \citep{hines2021demystifying}.  So, it is not possible to construct an estimator directly from an efficient influence function for the conditional effect.  Instead, the I-DR-Learner is a two stage meta-learner, which estimates the efficient influence function values for the relevant average effect (e.g., the average incremental effect for the CIE) in the first stage, and then regresses these values against the conditioning covariates in the second stage.  We show that if the second stage regression satisfies a generalization of the classic stochastic equicontinuity-type condition, the I-DR-Learner exhibits a form of double robustness and achieves oracle efficiency under weak model-agnostic conditions ($n^{-1/4}$ or slower convergence rates) on the nuisance function estimators.

\subsection{The Projection-Learner} \label{sec:proj-learner}

In this subsection, we illustrate the Projection-Learner.  We first define the finite dimensional working model 
$$
g(v; \delta, \beta) \equiv g(v; \beta), \beta \in \bbR^p,
$$
for incremental intervention parameter $\delta$ and model parameter $\beta$.  This model could be for the CIE, the CICE, or the CIDE, in which case we would use $\delta$ for the CIDE and the CIE, and $\delta_u$ and $\delta_l$ for the CICE.  For ease of exposition, we suppress the dependence of $g(v; \delta, \beta)$ on $\delta$ (or $\delta_u$ and $\delta_l$).  A simple example might be $g(v; \beta) = \beta_1 v$, where the covariate effect modification depends linearly on the value of the covariate.  But, the working model can be complex if needed, and should be informed by subject-specific knowledge if possible.   In what follows, we present results in terms of the CIDE, but the results also apply to the CIE and the CICE.

\medskip

We define the projection of the CIDE onto $g(v; \beta)$ as the $g(v; \beta)$ closest to $\tau_{cide} (v; \delta)$ over weighted $\ell_2$ distance.  Specifically, we define $\beta^\ast$ as the coefficients corresponding to the least-squares projection, and $g(v; \beta^\ast)$ as the projection. Mathematically, $\beta^\ast$ is 
\begin{equation} \label{eq:wm_min}
    \beta^\ast = \arg \min_{\beta} \int_{\mathcal{V}} \big\{ \tau_{cide} (v; \delta) - g(v; \beta) \big\}^2 d\bbP(v).
\end{equation}
One could also incorporate a weight function and use a different distance metric \citep{kennedy2021semiparametric}.  We set the weights to $1$ and focus on $\ell_2$ distance for ease of exposition, but all our results follow with other weights, and could be extended to other distance metrics.

\medskip

As long as $g(v; \beta)$ is differentiable with respect to $\beta$, $\beta^\ast$ is the solution of a moment condition.  The moment condition corresponds to the first derivative with respect to $\beta$,
\begin{equation} \label{eq:moment_condition}
    m(\beta) \equiv 2 \int_{\mathcal{V}} \frac{\partial g(v; \beta)}{\partial \beta} \{ \tau_{cide} (v; \delta) - g(v; \beta) \} d\bbP(v).
\end{equation}
Then, the solution $\beta^\ast$ in \eqref{eq:wm_min} satisfies $m(\beta^\ast) = 0$ in \eqref{eq:moment_condition}, and the projection of the CIDE onto the working model is $g(v; \beta^\ast)$.  

\begin{remark}
This setup is different from the proper semiparametric approach, since the definition of $\beta^\ast$ in eq. \eqref{eq:wm_min} does not assume anything about the true conditional effect curve.  By contrast, a proper semiparametric approach assumes a finite dimensional model is correctly specified for the conditional effect curve \citep{robins1992estimating, robinson1988root, vansteelandt2014structural, robins1994correcting}.
\end{remark}

It it is possible to derive an efficient influence function and thus a semiparametrically efficient estimator for the moment condition $m(\beta)$, and by extension for $\beta^\ast$ and $g(v; \beta^\ast)$.  We use this efficient influence function to construct the Projection-Learner.  The primary building block for the efficient influence function of the moment condition is the un-centered efficient influence function for the relevant average effect.  The efficient influence functions for the average incremental effect and the average incremental contrast effect were derived in \cite{kennedy2019incremental} Corollary 2, and are stated in equations \eqref{eq:ie_eif} and \eqref{eq:ice_eif} in the appendix.  Meanwhile, Lemma \ref{lem:cide_eif} establishes the efficient influence function for the average incremental derivative.

\begin{restatable}{lemma}{lemcideeif} \label{lem:cide_eif}
Under Assumptions \ref{asmp:cons} and \ref{asmp:exch}, the un-centered efficient influence function for the average incremental derivative effect, $\bbE \big\{ \tau_{cide} (V; \delta) \big\}$, is
\begin{align}
    \xi(Z; \delta) &= \left[ \frac{\pi(X) \{ 1 - \pi(X) \}}{ \{ \delta \pi(X) + 1 - \pi(X) \}^2} \right] \cdot \left[ \frac{A}{\pi(X)} \Big\{ Y - \mu(1, X) \Big\} - \frac{1-A}{1-\pi(X)} \Big\{ Y - \mu(0, X) \Big\} \right]  \nonumber \\
    &\hspace{0.2in} + \left[ \frac{1}{\{ \delta \pi(X) + 1 - \pi(X) \}^2} - \frac{2\delta\pi (X)}{\{ \delta \pi(X) + 1 - \pi(X) \}^3} \right] \cdot \Big\{ A - \pi(X) \Big\} \cdot \{ \mu(1, X) - \mu(0, X) \} \nonumber \\
    &\hspace{0.2in} + \left[ \frac{\pi(X) \{ 1 - \pi(X) \}}{ \{ \delta \pi(X) + 1 -  \pi(X) \}^2} \right] \cdot \{ \mu(1, X) - \mu(0, X) \}. \label{eq:xi}
\end{align}
\end{restatable}

The un-centered efficient influence function, $\xi(Z; \delta)$, depends only on the nuisance functions $\mu(A, X)$ and $\pi(X)$, and consists of three terms.  The first term is a product of the weight term, $\frac{\pi(X) \{ 1 - \pi(X) \}}{ \{ \delta \pi(X) + 1 - \pi(X) \}^2}$, and an inverse weighted residual for the outcome model.  The second term is a product of the difference in means, $\mu(1, X) - \mu(0, X)$, and an inverse weighted residual for the propensity score.  And, the third term is the ``plug-in'' for the CIDE.  

\begin{remark}
Throughout, we invoke Assumptions \ref{asmp:cons} and \ref{asmp:exch} so that the target of estimation is some counterfactual quantity (e.g., the CIDE).  If these assumptions do not hold, the results still apply if the targets of estimation are the observed data functionals on the right hand side of the identification results in Propositions \ref{prop:id_t1} and \ref{prop:deriv}.
\end{remark}

From Lemma \ref{lem:cide_eif} above, and Corollary 2 in \cite{kennedy2019incremental}, we can derive the efficient influence function for the moment condition $m(\beta)$ for estimating the projection of the CIDE, the CIE, or the CICE.

\begin{restatable}{corollary}{cormomenteif} \label{cor:m_eif}
Let $\xi(Z; \delta)$ denote the true influence function values of the relevant average effect, where $\xi(Z; \delta)$ is defined in \eqref{eq:xi} if the estimand is a projection of $\tau_{cide} (v; \delta)$, and is defined analogously, as shown in equations \eqref{eq:ie_eif} and \eqref{eq:ice_eif} in the appendix, if the estimand is a projection of $\tau_{cie}(v; \delta)$ or $\tau_{cice} (v; \delta_u, \delta_l)$.  Under Assumptions \ref{asmp:cons} and \ref{asmp:exch}, the un-centered efficient influence function for $m(\beta)$ under a nonparametric model with unknown propensity scores and a uniform weight function constructed over $\ell_2$ distance is
$$
\phi(Z; \delta, \beta) = \frac{\partial g(V; \beta)}{\partial \beta} \{ \xi (Z; \delta) - g(V; \beta) \},
$$
where $g(v; \beta)$ is the working model.
\end{restatable}

Corollary \ref{cor:m_eif} motivates estimators for $\beta^\ast$ and $g(v; \beta^\ast)$. The first step estimates the un-centered efficient influence function values for the relevant average effect; for example, when estimating the projection of the CIDE, we have
\begin{align} 
    \widehat \xi(Z; \delta) &= \left[ \frac{\widehat \pi(X) \{ 1 - \widehat \pi(X) \}}{ \{ \delta \widehat \pi(X) + 1 - \widehat \pi(X) \}^2} \right] \cdot \left[ \frac{A}{\widehat \pi(X)} \Big\{ Y - \widehat \mu(1, X) \Big\} - \frac{1-A}{1-\widehat \pi(X)} \Big\{ Y - \widehat \mu(0, X) \Big\} \right]  \nonumber \\
    &\hspace{0.2in} + \left[ \frac{1}{\{ \delta \widehat \pi(X) + 1 - \widehat \pi(X) \}^2} - \frac{2\delta \widehat \pi (X)}{\{ \delta \widehat \pi(X) + 1 - \widehat \pi(X) \}^3} \right] \cdot \Big\{ A - \widehat \pi(X) \Big\} \cdot \{ \widehat \mu(1, X) - \widehat \mu(0, X) \} \nonumber \\
    &\hspace{0.2in} + \left[ \frac{\widehat \pi(X) \{ 1 - \widehat \pi(X) \}}{ \{ \delta \widehat \pi(X) + 1 - \widehat \pi(X) \}^2} \right] \cdot \{ \widehat \mu(1, X) - \widehat \mu(0, X) \} \label{eq:xihat}
\end{align}
where $\widehat \mu(0, X), \widehat \mu(1, X)$, and $\widehat \pi(X)$ are (possibly nonparametric) estimates of the nuisance functions.  The second step estimates the population moment condition by solving the empirical moment condition using the estimated un-centered efficient influence function values for $m(\beta)$,
$$
\bbP_n \left[ \frac{\partial g(V; \widehat \beta)}{\partial \beta} \left\{ \widehat \xi(Z; \delta) - g (V; \widehat \beta) \right\} \right] = 0.
$$
We state the Projection-Learner formally in the following algorithm.
\begin{algorithm} \label{alg:working_mod} \emph{(Projection-Learner)} Assume as inputs $(D_{1}, D_2)$, which denote two independent samples of $n$ observations of $Z_i = (X_i, A_i, Y_i)$. Then:
\begin{enumerate}
    \item On the training data $D_1$, estimate the nuisance functions $\widehat \mu (0, X)$, $\widehat \mu(1, X)$ and $\widehat \pi(X)$.
    \item  On the estimation data $D_2$, estimate the un-centered influence function values $\widehat \xi (Z; \delta)$ using the models $\widehat \mu(0, X)$, $\widehat \mu(1, X)$ and $\widehat \pi(X)$ from step 1, where $\widehat \xi (Z; \delta)$ is defined in \eqref{eq:xihat} if the conditional effect of interest is $\tau_{cide}$, and analogously for $\tau_{cie}$ and $\tau_{cice}$ in equations \eqref{eq:ie_eif} and \eqref{eq:ice_eif} in the appendix.
    \item On the estimation data $D_2$, estimate $\widehat \beta$ by solving the empirical moment condition
    $$
    \bbP_n \left[ \frac{\partial g(V; \beta)}{\partial \beta} \left\{ \widehat \xi (Z; \delta) - g(V; \beta)  \right\} \right] = 0
    $$
    \end{enumerate}
\end{algorithm}

Algorithm \ref{alg:working_mod} is relatively straightforward. For example, if the working model is $g(v; \beta) = \beta_1 + \beta_2 \cdot v + \beta_3 \cdot v^2$, then Algorithm \ref{alg:working_mod} solves the empirical moment condition
$$
\bbP_n \left[ \begin{pmatrix} 1 \\ V \\ V^2 \end{pmatrix}  \left\{ \widehat \xi (Z; \delta) - (\beta_1 + \beta_2 V + \beta_3 V^2)  \right\} \right] = 0
$$
which can be achieved in \texttt{R} by running the regression 
\begin{lstlisting}[language=R]
    model <- lm(formula = xihat ~ V + I(V^2))
\end{lstlisting}
where \texttt{xihat} is calculated from estimated nuisance functions $\widehat \mu(A, X)$ and $\widehat \pi(X)$.

\begin{remark} \label{rem:twostage}
The structure of Algorithm \ref{alg:working_mod} and the example code also illustrate that the Projection-Learner uses estimated un-centered efficient influence functions values for $\bbE \{ \tau_{cide} (V; \delta) \}$ as pseudo-outcomes in a parametric second stage regression.  In Section \ref{sec:i-dr-learner}, we show that the I-DR-Learner follows the same form, but with a nonparametric second stage regression.
\end{remark}

To guarantee the convergence rates demonstrated in Theorem \ref{thm:fixed_mod} below, we could assume Donsker-type or low-entropy conditions for the nuisance functions $\mu(A, X)$ and $\pi(X)$, which restricts what types of flexible estimators we can use \citep{van1996weak, van2000asymptotic}.  Instead, we use sample splitting in step 1 of Algorithm \ref{alg:working_mod} to estimate the nuisance functions; i.e., we split our sample in two, and estimate the nuisance functions on the training data, $D_1$, and calculate $\widehat \xi(Z; \delta)$ and solve the empirical moment condition on the estimation data, $D_2$.  Sample splitting allows us to condition on the training sample and treat the estimated nuisance functions as fixed functions, which expands the class of estimators possible for estimating the nuisance functions. A concern one might then have with Algorithm \ref{alg:working_mod} is that it only estimates $\widehat \beta$ on half the sample.  To utilize the whole sample for inference, we can improve on Algorithm \ref{alg:working_mod} with cross-fitting by estimating the nuisance functions on both folds ($D_1$ and $D_2$), constructing $\widehat \xi(Z; \delta)$ values on the opposite fold (i.e., by estimating $\widehat \xi(Z; \delta)$ in $D_1$ using nuisance functions constructed on $D_2$, and vice versa), and solving the empirical moment condition on the whole dataset ($D_1$ and $D_2$ together) \citep{chernozhukov2018double, zheng2010asymptotic, robins2008higher}.  This cross-fitting approach is also compatible with more folds (``k-fold cross-fitting''), which can be more stable than two-fold cross-fitting.

\medskip

The following theorem shows that the estimator $\widehat \beta$ for $\beta^\ast$ outlined in Algorithm \ref{alg:working_mod} converges to an asymptotically linear expansion around $\beta^\ast$ where the bias is expressed as a product of errors from estimating the nuisance functions $\widehat \mu(0, X), \widehat \mu(1, X)$ and $\widehat \pi(X)$.  For this result, and the rest of Section \ref{sec:proj-learner}, we let $\mu = \{ \mu(0, X), \mu(1, X) \}$ and $\pi = \pi(X)$ denote generic nuisance functions, $\widehat \mu$ and $\widehat \pi$ denote the nuisance function estimators, and define $\mu^\ast$ and $\pi^\ast$ as the true nuisance functions (consistent with the projection notation $\beta^\ast$).

\begin{restatable}{theorem}{thmfixedmod} \label{thm:fixed_mod} Let $\varphi(Z; \beta, \mu, \pi) \equiv \phi(Z; \delta, \beta) - m(\beta)$ denote the centered efficient influence function from Corollary \ref{cor:m_eif}.  With Assumptions \ref{asmp:cons} and \ref{asmp:exch}, also assume
\begin{enumerate}[label=(\alph*)]
    \item $\bbP \Big( |\widehat \mu (1, X) - \widehat \mu(0, X)| \leq C \Big) = 1$ and $\bbP \Big( | \mu^\ast (1, X) - \mu^\ast (0, X) | \leq C \Big) = 1$ for some $C < \infty$. \label{asm:boundedness}
    \item $\bbP \left\{ \left| \frac{\partial g(\beta; v)}{\partial \beta} \right| \leq C \right\} = 1$ for all $v$ \label{asm:bdd_model}
    \item The function class $\varphi(Z; \beta, \mu, \pi)$ is Donsker in $\beta$ for any fixed $\mu, \pi$. \label{asm:beta_donsker}
    \item The estimators are consistent in the sense that $\widehat{\beta} - \beta^\ast = o_{\bbP} (1)$ and $\lVert \varphi(Z; \widehat \beta; \widehat \mu, \widehat \pi) - \varphi(Z; \beta^\ast; \mu^\ast, \pi^\ast) \rVert = o_\bbP (1)$.\label{asm:beta_varphi_cons} 
    \item The map $\beta \mapsto \bbP \{ \varphi(z; \beta, \mu, \pi) \}$ is differentiable at $\beta^\ast$ uniformly in $(\mu, \pi)$, with nonsingular derivative matrix $\frac{\partial}{\partial \beta} \bbP \{ \varphi(Z; \beta, \mu, \pi) \} |_{\beta = \beta^\ast} = M(\beta^\ast, \mu, \pi)$, where $M(\beta^\ast, \widehat \mu, \widehat \pi) \overset{p}{\to} M(\beta^\ast, \mu^\ast, \pi^\ast)$. \label{asm:diff_map}
\end{enumerate}
Then
$$
\widehat{\beta} - \beta^\ast = - M(\beta^\ast, \mu^\ast, \pi^\ast)^{-1} (\bbP_n - \bbP) \left\{ \varphi(Z; \beta^\ast, \mu^\ast, \pi^\ast)  \right\} + O_{\bbP} \left( R_n + o_{\bbP} \left( \frac{1}{\sqrt{n}} \right) \right)
$$
where 
\begin{align*}
    R_n = \Big( \lVert \widehat \mu - \mu^\ast \rVert + \lVert \widehat \pi - \pi^\ast \rVert \Big) \lVert \widehat \pi - \pi^\ast \rVert
\end{align*}
\end{restatable}
Theorem \ref{thm:fixed_mod} provides a convergence statement for the coefficient estimate $\widehat{\beta}$ to the true projection parameter $\beta^\ast$ under relatively weak conditions.  Assumption \ref{asm:boundedness} says the CATE and the estimate of the CATE are bounded.  Assumption \ref{asm:bdd_model} says that the derivative of the model $g(v; \beta)$ with respect to $\beta$ is bounded, which is quite weak and can be enforced through choice of an appropriate model.  Assumption \ref{asm:beta_donsker} ensures the influence function $\varphi$ is not too complex as a function of $\beta$, but allows for arbitrary complexity in the nuisance functions; again, this can be enforced with appropriate choice of $g(v; \beta)$, and most reasonable choices will suffice. Assumption \ref{asm:beta_varphi_cons} requires that $\{\beta^\ast, \varphi(Z; \beta^\ast, \mu^\ast, \pi^\ast) \}$ is consistently estimated by $\{ \widehat \beta, \varphi(Z; \widehat \beta, \widehat \mu, \widehat \pi) \}$ at any rate.  Finally, Assumption \ref{asm:diff_map} requires some smoothness of $\bbP \{ \varphi (Z; \beta, \mu, \pi) \}$ in $\beta$, to allow for use of the delta method.  Assumptions \ref{asm:beta_donsker}-\ref{asm:diff_map} are standard in the literature (see, for example, \cite{van2000asymptotic} Theorem 5.31).

\medskip

The convergence statement shows that $\widehat \beta$ obtains a faster rate of convergence to $\beta^\ast$ than the nuisance function estimators $\widehat \mu$ and $\widehat \pi$ obtain to $\mu^\ast$ and $\pi^\ast$ respectively.  The first term, $M(\beta^\ast, \mu^\ast, \pi^\ast)^{-1} (\bbP_n - \bbP) \{ \varphi(Z; \beta^\ast, \mu^\ast, \pi^\ast) \}$, is a sample average scaled by a constant, and so by the central limit theorem it is asymptotically Gaussian.  Therefore, if $R_n = o_{\bbP} (n^{-1/2})$ then the remainder terms $O_{\bbP} \left( R_n + o_{\bbP} (n^{-1/2}) \right)$ will be asymptotically negligible and so $\widehat \beta - \beta^\ast$ will converge in distribution to a mean-zero Gaussian distribution with variance equal to the variance of $M(\beta^\ast, \mu^\ast, \pi^\ast)^{-1} (\bbP_n - \bbP) \{ \varphi(Z; \beta^\ast, \mu^\ast, \pi^\ast) \}$, as shown in the following result.

\begin{corollary} \label{cor:lim_dist}
Under the same assumptions as Theorem \ref{thm:fixed_mod}, if
$$
\Big( \lVert \widehat \mu - \mu^\ast \rVert + \lVert \widehat \pi - \pi^\ast \rVert \Big)  \lVert \widehat \pi - \pi^\ast \rVert  = o_{\bbP} \left( \frac{1}{\sqrt{n}} \right),
$$
then
$$
\sqrt{n} (\widehat{\beta} - \beta^\ast) \indist N \Big( 0, M^{-1} \bbE(\varphi \varphi^T) (M^{-1})^T \Big),
$$
and for any fixed $v$ we have
$$
\sqrt{n} ( g(v; \widehat \beta) - g(v; \beta^\ast) \indist N\left(0, \left( \frac{\partial g(v; \beta^\ast)}{\partial \beta} \right)^{T} M^{-1} \bbE(\varphi \varphi^T) (M^{-1})^T  \left( \frac{\partial g(v; \beta^\ast)}{\partial \beta} \right) \right),
$$
where
\begin{align*}
    M &= M(\beta^\ast, \mu^\ast, \pi^\ast), \text{and} \\
    \bbE(\varphi \varphi^T) &= \bbE \Big\{ \varphi(Z; \beta^\ast, \mu^\ast,  \pi^\ast) \times \varphi(Z; \beta^\ast, \mu^\ast, \pi^\ast)^{T} \Big\}.
\end{align*} 
\end{corollary}
\noindent Corollary \ref{cor:lim_dist} provides a way to construct an asymptotically valid Wald-style 1-$\alpha$ confidence interval around $g(\widehat \beta, v)$ with
$$
g(v; \widehat \beta) \pm \Phi^{-1}(1-\alpha/2) \left( \frac{\widehat{\sigma}(v)}{\sqrt{n}} \right),
$$
where $\Phi(\cdot)$ is the cumulative distribution function for the standard normal, 
$$
\widehat{\sigma}^2(v) = \left( \frac{\partial g(v; \widehat \beta)}{\partial \beta} \right)^{T} \widehat{M}^{-1} \bbE(\widehat{\varphi} \widehat{\varphi}^T) (\widehat{M}^{-1})^T  \left( \frac{\partial g(v; \widehat \beta)}{\partial \beta} \right),
$$
and $\widehat{M} = \bbP_n (\partial \widehat{\varphi} / \partial \beta)$ is an estimate of the derivative matrix. Furthermore, this corollary demonstrates that $\widehat \beta$ and $g(v; \widehat \beta)$ converge at $n^{-1/2}$ rates to Gaussian distributions, centered at $\beta^\ast$ and $g(v; \beta^\ast)$ respectively, with less stringent model-agnostic convergence conditions on the nuisance function estimators $\widehat \mu(0, X), \widehat \mu(1, X)$, and $\widehat \pi(X)$.  Thus, $\widehat \beta$ and $g(v; \widehat \beta)$ still attain $n^{-1/2}$ convergence rates if both nuisance functions are estimated at $n^{-1/4}$ rates, which are attainable with nonparametric estimators under relatively realistic assumptions such as smoothness or sparsity \citep{tsybakov2009introduction, birge1995estimation, tibshirani1996regression, farrell2015robust}.

\medskip

\begin{remark}
These results are doubly-robust in spirit since the remainder bias is expressed as a product of nuisance function errors.  However, there is no ``double robustness'' in the traditional sense, which would only require $\lVert \widehat \mu - \mu^\ast \rVert \lVert \widehat \pi - \pi^\ast \lVert = o_{\bbP}(n^{-1/2})$. Instead, Corollary \ref{cor:lim_dist} requires that the propensity score is estimated well enough that $\lVert \widehat \pi - \pi^\ast \rVert^2 = o_{\bbP} (n^{-1/2})$.  Intuitively, this occurs because incremental interventions shift the observed propensity scores, and thus require a good estimate of the propensity score. By contrast, the intervention corresponding to the CATE does not depend on the propensity score, so the convergence rate for the propensity score estimator is less critical, depending on that of the outcome regression.  
\end{remark}

\medskip

As demonstrated in Theorem \ref{thm:fixed_mod} and Corollary \ref{cor:lim_dist}, the Projection-Learner can attain $n^{-1/2}$ convergence rates to the projection of the true CIDE (or CIE, or CICE) onto the chosen working model $g(v; \beta)$.  If, instead, we wish to target the true conditional effect curve, and that curve does not coincide with the projection, then we need to use a different estimator, as we describe in the next section.


\subsection{The I-DR-Learner} \label{sec:i-dr-learner}

In this section, we outline the I-DR-Learner and illustrate its convergence properties.  The I-DR-Learner targets the true conditional effects.  Since the conditional effects are not pathwise differentiable, no efficient influence function exists for them.  Instead, the I-DR-Learner makes use of the efficient influence function values for the relevant average effect by regressing them against the covariates of interest to estimate the conditional effect.  In this way, the I-DR-Learner is a two stage meta-learner, where the first stage estimates the efficient influence function values for the relevant average effect, and the second stage uses these values as pseudo-outcomes in a second stage regression against the conditioning covariates.  The I-DR-Learner is stated formally in the following algorithm:

\begin{algorithm} \emph{(I-DR-Learner).} \label{alg:i-dr-learner}
Assume as inputs $(D_{1}, D_2)$, which denote two independent samples of $n$ observations of $Z_i = (X_i, A_i, Y_i)$.

    \begin{enumerate}
        \item On the training data $D_1$, estimate the nuisance functions $\widehat \mu (0, X)$, $\widehat \mu(1, X)$ and $\widehat \pi(X)$.
        \item On the estimation data $D_2$, estimate the un-centered influence function values $\widehat \xi (Z; \delta)$ using the models $\widehat \mu(0, X)$, $\widehat \mu(1, X)$ and $\widehat \pi(X)$ from step 1, where $\widehat \xi (Z; \delta)$ is defined in \eqref{eq:xihat} if the conditional effect of interest is $\tau_{cide}$, and analogously for $\tau_{cie}$ and $\tau_{cice}$ in equations \eqref{eq:ie_eif} and \eqref{eq:ice_eif} in the appendix.
        \item In the estimation sample $D_2$, regress $\widehat \xi (Z; \delta)$ on the conditioning covariates $V$ to obtain the estimate
        $$
        \widehat{\tau}_{i-dr} (v; \delta) = \widehat \bbE_n \left\{ \widehat \xi (Z; \delta) \mid V = v \right\}.
        $$
    \end{enumerate}
    
\end{algorithm}

Like the Projection-Learner, the I-DR-Learner also uses sample splitting and estimates the nuisance functions on a separate sample to avoid imposing Donsker-type conditions on the nuisance function estimators.  The I-DR-Learner is also compatible with cross-fitting.

\medskip

The I-DR-Learner can estimate all three conditional effects - the CIE, CICE, and CIDE.  Furthermore, the error of the estimator is asymptotically equal to that of an oracle estimator under certain conditions.  Specifically, the second stage regression must satisfy the stability condition in Definition 1 of \cite{kennedy2020towards}.  This is a generalization of the classic stochastic equicontinuity condition to nonparametric regression (Lemma 19.24 \cite{van2000asymptotic}), and says that the second stage regression is \emph{stable} with respect to a distance metric $d$ if the difference between second stage regressions with estimated outcomes and true outcomes shrinks appropriately fast.  We discuss Definition 1 of \cite{kennedy2020towards} in more detail in the appendix.  The stability condition is satisfied by the class of linear smoothers (see \cite{kennedy2020towards} Theorem 1), which includes nonparametric estimators like kernel smoothers, series regression, and random forests.  It is possible that other classes of estimators also satisfy the stability condition, although examining that question is beyond the scope of this work.

\medskip

Under this stability condition, the error of the I-DR-Learner can be tied to the error of an oracle estimator, which would have access to the un-centered efficient influence function values for the relevant average effect and would estimate the conditional effect merely by running a regression of $\xi(Z; \delta)$ against $V$.  This approach was considered in \cite{kennedy2020towards} for estimating the CATE, and their Theorem 2 showed that, under certain assumptions, the error of their DR-Learner will only exceed the error of an oracle estimator by an amount that depends on the product of errors in estimating the nuisance functions.  The same logic holds for the I-DR-Learner, and we formally state the convergence result in the following theorem.  We slightly amend notation from Section \ref{sec:proj-learner}, and allow $\xi, \mu,$ and $\pi$ to denote the true efficient influence function values and nuisance functions.

\begin{restatable}{theorem}{thmdrilearner} \label{thm:i-dr-learner}
Let $\tau_{i-dr}$ stand in for $\tau_{cide}, \tau_{cie}$, or $\tau_{cice}$, and let $\xi(Z; \delta)$ denote the true influence function values of the relevant average effect.  Furthermore, let $\tilde \tau_{i-dr} (v; \delta) = \widehat \bbE_n \{ \xi(Z; \delta) \mid V = v \}$ denote an oracle estimator that regresses $\xi(Z; \delta)$ on $V$, and let $\widehat \tau_{i-dr} (v; \delta)$ denote the I-DR-Learner from Algorithm \ref{alg:i-dr-learner}.  With Assumptions \ref{asmp:cons} and \ref{asmp:exch}, and Assumption \ref{asm:boundedness} from Theorem \ref{thm:fixed_mod}, also assume that the second stage regression is stable according to Definition 1 of \cite{kennedy2020towards}. Then,
$$
\widehat \tau_{i-dr} (v; \delta) - \tau_{i-dr} (v; \delta) = \tilde \tau_{i-dr} (v; \delta) - \tau_{i-dr} (v; \delta) + \widehat \bbE_n \{ \widehat b(X) \mid V = v \} + o_{\bbP} \left( R^\ast (v; \delta) \right)
$$
for 
\begin{align*}
    \widehat b (x) &\lesssim  \Big( \big| \widehat \mu (0, x) - \mu (0, x) \big| + \big| \widehat \mu(1, x) - \mu(1, x) \big| + \big| \widehat \pi(x) - \pi(x) \big| \Big) \cdot \big| \widehat \pi(x) - \pi(x) \big|
\end{align*}
and 
$$
R^\ast (v; \delta) = \sqrt{\bbE \left[ \Big\{ \tilde \tau_{i-dr} (v; \delta) - \tau_{i-dr} (v; \delta) \Big\}^2 \right]}.
$$
\end{restatable}

Theorem \ref{thm:i-dr-learner} shows that error for the I-DR-Learner differs from the error for the oracle estimator by at most $\widehat \bbE_n \{ \widehat b(X) \mid X = x \}$ plus other terms, captured by $o_\bbP (R^\ast(v; \delta))$, that are asymptotically negligible compared to the error of the oracle estimator.  Thus, whether the I-DR-Learner achieves oracle efficiency is driven by the asymptotic behavior of the smoothed bias term $\widehat \bbE_n \{ \widehat b(X) \mid X = x \}$.  This bias term is asymptotically less than the product of errors for estimating $\mu(x)$ and $\pi(x)$, $\big| \widehat \mu (a, x) - \mu (a, x) \big|  \cdot \big| \widehat \pi(x) - \pi(x) |$, and the squared error for estimating $\pi(x)$, $\big\{ \widehat \pi(x) - \pi(x) \big\}^2$.  Therefore, the convergence rate of the I-DR-Learner is faster than the convergence rate of the nuisance function estimators.  For example, if the nuisance functions are estimated at $n^{-1/4}$ rates, then the bias term $\widehat b(X)$ will converge to zero at a $n^{-1/2}$ rate.  Importantly, Theorem \ref{thm:i-dr-learner} does not require any assumptions about how the estimators $\widehat \mu$ and $\widehat \pi$ are constructed, beyond the boundedness conditions from Assumption \ref{asm:boundedness} from Theorem \ref{thm:fixed_mod}. 

\medskip

However, the performance of the I-DR-Learner is also constrained by the oracle convergence rate for the second stage regression. For example, if $\bbE \{ \xi(Z; \delta) \mid V = v \}$  is H\"{o}lder-smooth with smoothness $s$, then the minimax rate in root mean squared error is $n^{- 1 / (2 + \frac{d}{s})}$, which is slower than $n^{-1/2}$ \citep{birge1995estimation}.  This is not surprising -- since the conditional effect is a regression function, if we are only willing to assume it lies in a large nonparametric class, then the minimax rate of convergence will be slower than $n^{-1/2}$.  One can also think of the slower oracle convergence as a positive aspect to the I-DR-Learner, since it reduces how well the nuisance functions must be estimated to achieve oracle efficiency. For example, if the oracle convergence rate is ``only'' $n^{-1/4}$, then the I-DR-Learner can estimate each nuisance function at $n^{-1/8}$ convergence rates and still attain oracle efficiency.  When the nuisance functions are estimated well enough and the I-DR-Learner is oracle efficient, confidence bands can be constructed following well-known processes for nonparametric regression \citep{wasserman2006all}.  

\medskip

Both the Projection-Learner and the I-DR-Learner can be used to estimate conditional effect curves across $\delta$ and $V$, thereby quantifying how causal effects vary across $V$.  A natural question is whether there is any treatment effect heterogeneity across $V$.  In the following section, we outline how to quantify and test for treatment effect heterogeneity.

\section{Understanding effect heterogeneity with the V-CIDE} \label{sec:heterogeneity}

There is a large literature for understanding treatment effect heterogeneity by summarizing the CATE (e.g., \cite{crump2008nonparametric, ding2016randomization, ding2019decomposing, luedtke2019omnibus}).  In this section, we demonstrate how the variance of the CIDE, the V-CIDE, defined in eq. \eqref{eq:vcide}, can be used to understand effect heterogeneity.  To ease exposition, we focus on the case where $V = X$, and examine effect heterogeneity across all covariates.  The case where $V$ is a strict subset of $X$ (i.e., $V \subset X$) is outlined in the appendix.  By Proposition \ref{prop:deriv}, the V-CIDE is identified by 
$$
\bbV \left\{ \tau_{cide}(X; \delta) \right\} = \bbV  \left( \left[ \frac{\pi(X) \{ 1 - \pi(X) \}}{ \{ \delta \pi(X) + 1 - \pi(X) \}^2} \right] \cdot \{\mu(1, X) - \mu(0, X) \} \right).
$$
When the V-CIDE is zero, the derivative is constant across $V$, and so shifting the treatment distribution has the same effect on all subjects.  If the V-CIDE is greater than zero, then there is treatment effect heterogeneity in the incremental effect. 

\medskip

We construct an estimator in two pieces by first noting that the V-CIDE is the difference between two effects since $\bbV \{ \tau_{cide} (X; \delta) \} = \bbE \{ \tau_{cide} (X; \delta)^2\} - \bbE \{ \tau_{cide} (X; \delta) \}^2$. The first effect, $\bbE \{ \tau_{cide} (X; \delta)^2\}$, admits an efficient influence function by the following lemma:

\begin{restatable}{lemma}{lemvareif}\label{lem:var_eif}
Under Assumptions \ref{asmp:cons} and \ref{asmp:exch}, the un-centered efficient influence function for $\bbE \left\{ \tau_{cide}(X; \delta)^2 \right\}$ is
$$
2 \omega (X; \delta) \Big\{ \mu(1, X) - \mu(0, X) \Big\} \left[ \omega(X; \delta) \varphi(Z) + \phi(Z; \delta) \Big\{ \mu(1, X) - \mu(0, X) \Big\} \right] + \left[ \omega(X; \delta) \Big\{ \mu(1, X) - \mu(0, X) \Big\} \right]^2 
$$
where
\begin{align}
    \omega(X; \delta) &= \frac{\pi(X) \{ 1 - \pi(X) \}}{ \{ \delta \pi(X) + 1 - \pi(X) \}^2} \label{eq:omega} \\
    \varphi(Z) &= \frac{A}{\pi(X)} \Big\{ Y - \mu(1, X) \Big\} - \frac{1-A}{1-\pi(X)} \Big\{ Y - \mu(0, X) \Big\}  \label{eq:varphi} \\
    \phi(Z; \delta) &= \left[ \frac{1}{\{ \delta \pi(X) + 1 - \pi(X) \}^2 } - \frac{2 \delta \pi(X) }{\{ \delta \pi(X) + 1 - \pi(X) \}^3} \right] \cdot \left\{ A - \pi(X) \right\} \label{eq:phi} 
\end{align}
\end{restatable}
Lemma \ref{lem:var_eif} shows that the un-centered efficient influence function for $\bbE \left\{ \tau_{cide}(X; \delta)^2 \right\}$ can be written as a weighted residual plus a plug-in. The second effect, $\bbE \{ \tau_{cide} (X; \delta)\}^2$, is also pathwise differentiable and admits an efficient influence function.   However, since it is a smooth transformation of an already pathwise differentiable function, we estimate it by squaring the estimator based on the efficient influence function for $\bbE \{ \tau_{cide} (X; \delta)\}$ provided in Lemma \ref{lem:cide_eif}.  Therefore, informed by Lemmas \ref{lem:cide_eif} and \ref{lem:var_eif}, we propose the estimator
\begin{align}
    \widehat {\bbV} \{ \tau_{cide}( X; \delta) \} &= \underbrace{\bbP_n \left[ 2 \widehat \omega \big( \widehat \mu_1 - \widehat \mu_0 \big) \left\{ \widehat \omega \widehat \varphi + \widehat \phi \big( \widehat \mu_1 - \widehat \mu_0 \big) \right\} + \left\{ \widehat \omega \big( \widehat \mu_1 - \widehat \mu_0 \big) \right\}^2 \right]}_{\text{Estimator for } \bbE \{ \tau_{cide} (X; \delta)^2 \}}  \label{eq:exp_sq_est} \\
    &- \underbrace{\left[ \bbP_n \Big\{ \widehat \omega \widehat \varphi + \widehat \phi \big( \widehat \mu_1 - \widehat \mu_0 \big) + \widehat \omega \big( \widehat \mu_1 - \widehat \mu_0 \big) \Big\} \right]^2}_{\text{Estimator for } \bbE \{ \tau_{cide} (X; \delta) \}^2}, \label{eq:sq_exp_est}
\end{align}
where we omit $\delta$, $X$, and $Z$ arguments and let $\mu_a = \mu(a, X)$ for brevity, and where $\widehat \omega, \widehat \varphi, \widehat \phi$ indicate the relevant formulae from \eqref{eq:omega}-\eqref{eq:phi}, but with the estimated nuisance functions (e.g., $\widehat \omega = \frac{\widehat \pi(X) \{ 1 -  \widehat \pi(X) \}}{ \{ \delta \widehat \pi(X) + 1 - \widehat \pi(X) \}^2}$).  Eq. \eqref{eq:exp_sq_est} is the estimator for $\bbE \{ \tau_{cide} (X; \delta)^2 \}$ motivated by Lemma \ref{lem:var_eif} - it takes the sample average of the estimated un-centered efficient influence function values for $\bbE \{ \tau_{cide} (X; \delta)^2 \}$.  Eq. \eqref{eq:sq_exp_est} is the estimator for $\bbE \{ \tau_{cide} (X; \delta) \}^2$ motivated by Lemma \ref{lem:cide_eif} - it squares the estimator for $\bbE \{ \tau_{cide} (X; \delta) \}$, which itself is just the sample average of the estimated un-centered efficient influence function values for $\bbE \{ \tau_{cide} (X; \delta) \}$ in eq. \eqref{eq:xihat}.  Formally, we outline the estimator in the following algorithm:

\begin{algorithm} \label{alg:vcide} \emph{(V-CIDE Estimator)} Assume as inputs $(D_{1}, D_2)$, which denote two independent samples of $n$ observations of $Z_i = (X_i, A_i, Y_i)$, then:
\begin{enumerate}
    \item On the training data $D_1$, estimate the nuisance functions $\widehat \mu (0, X)$, $\widehat \mu(1, X)$ and $\widehat \pi(X)$.
    \item On the estimation data $D_2$, estimate $\bbV \{ \tau_{cide} (X; \delta) \}$ per equations \eqref{eq:exp_sq_est} and \eqref{eq:sq_exp_est}, plugging in the estimates for $\widehat \mu(0, X), \widehat \mu(1, X)$ and $\widehat \pi(X)$ using the models from step 1.
\end{enumerate}
\end{algorithm}

As before, Algorithm \ref{alg:vcide} uses sample splitting to estimate the nuisance functions, which allows for estimating the nuisance functions with flexible machine learning models.  Again, this estimator could use cross-fitting by repeating the algorithm but with $D_1$ and $D_2$ reversed and then averaging the two estimates.  We establish the error guarantees of the estimator in the following result.

\begin{restatable}{theorem}{thmvcideconv}\label{thm:vcide_conv}
Let $\widehat \psi_n$ denote the estimator from Algorithm \ref{alg:vcide}.  With Assumptions \ref{asmp:cons} and \ref{asmp:exch}, and Assumption \ref{asm:boundedness} from Theorem \ref{thm:fixed_mod}, also assume that 
\begin{enumerate}[label=(\alph*)]
    \item $\bbP \Big\{ \omega (\mu_1 - \mu_0) + \omega \varphi + \phi \tau \leq C \Big\}$ and $\bbP \Big\{ \widehat \omega (\widehat \mu_1 - \widehat \mu_0) + \widehat \omega \widehat \varphi + \widehat \phi \widehat \tau \leq C \Big\} = 1$ for some $C < \infty$. \label{asmp:bounded_eif}
\end{enumerate}
If
$$
\Big( \lVert \widehat \mu - \mu \lVert + \lVert \widehat \pi - \pi \rVert \Big)^2 = o_{\bbP} \left( \frac{1}{\sqrt{n}} \right),
$$
then
$$
\sqrt{n} \Big[ \widehat \psi_n - \bbV \{ \tau_{cide}(X; \delta) \} \Big] \indist N(0, \sigma^2),
$$
where 
\begin{align}
    \sigma^2 = \bbV \bigg[ &2 \omega \big( \mu_1 - \mu_0 \big) \left\{ \omega \varphi + \phi \big( \mu_1 - \mu_0 \big) \right\} + \left\{ \omega \big( \mu_1 - \mu_0 \big) \right\}^2 \nonumber \\
    &- \bbE \Big\{ \omega \varphi + \phi \big( \mu_1 - \mu_0 \big) + \omega \big( \mu_1 - \mu_0 \big) \Big\} \cdot \Big\{ \omega \varphi + \phi \big( \mu_1 - \mu_0 \big) + \omega \big( \mu_1 - \mu_0 \big) \Big\} \bigg],  \label{eq:pop_vcide_var}
\end{align}
$\mu_a = \mu(a, X)$, and $\omega = \omega(X; \delta), \varphi = \varphi(Z)$, and $\phi = \phi(Z; \delta)$ as defined in equations \eqref{eq:omega}, \eqref{eq:varphi}, and \eqref{eq:phi}. 
\end{restatable}

Theorem \ref{thm:vcide_conv} shows that the estimator for the V-CIDE satisfies a version of double robustness under relatively weak conditions.  Assumption \ref{asmp:bounded_eif} says that the efficient influence function for the average derivative and the estimate for the efficient influence function are bounded, which is a mild assumption.  Then, if both nuisance function estimators converge at $n^{-1/4}$ rates, the standardized difference between the estimator and the V-CIDE has a Gaussian limiting distribution.  This is a slightly stronger requirement than that of Corollary \ref{cor:lim_dist}, since both nuisance functions must be estimated at $n^{-1/4}$ rates, not just the propensity score.   This occurs due to the nonlinearity of $\bbE \{ \tau_{cide}(X; \delta)^2 \}$ in terms of $\mu$.  Nonetheless, this result is still model-agnostic about the nuisance function estimators, and the convergence requirement can be satisfied by nonparametric estimators under suitable smoothness or sparsity.  Theorem \ref{thm:vcide_conv} suggests constructing Wald-style $1-\alpha$ confidence intervals with 
\begin{equation} \label{eq:ci_alternative}
\widehat \psi_n \pm \Phi^{-1} (1 - \alpha/2) \sqrt{\frac{\widehat \sigma^2}{n}},
\end{equation}
where $\widehat \sigma^2$ is the sample variance estimator for $\sigma^2$ defined in eq. \eqref{eq:pop_vcide_var}; i.e., 
\begin{align}
    \widehat \sigma^2 = \bbV_n \bigg[ &2 \widehat \omega \big( \widehat \mu_1 - \widehat \mu_0 \big) \left\{ \widehat \omega \widehat \varphi + \widehat \phi \big( \widehat \mu_1 - \widehat \mu_0 \big) \right\} + \left\{ \widehat \omega \big( \widehat \mu_1 - \widehat \mu_0 \big) \right\}^2 \nonumber \\
    &- \bbP_n \Big\{ \widehat \omega \widehat \varphi + \widehat \phi \big( \widehat \mu_1 - \widehat \mu_0 \big) + \widehat \omega \big( \widehat \mu_1 - \widehat \mu_0 \big) \Big\} \cdot \Big\{ \widehat \omega \widehat \varphi + \widehat \phi \big( \widehat \mu_1 - \widehat \mu_0 \big) + \widehat \omega \big( \widehat \mu_1 - \widehat \mu_0 \big) \Big\} \bigg] \label{eq:vcide_var}
\end{align}
where $\bbV_n$ denotes the sample variance.  

\medskip

Unfortunately, the estimator in Algorithm \ref{alg:vcide} converges to a degenerate distribution when $\bbV \{ \tau_{cide} (X; \delta) \} = 0$ because the efficient influence function values are identically zero, and $\sigma^2 = 0$ in eq \eqref{eq:pop_vcide_var}.  So, the confidence interval in \eqref{eq:ci_alternative} would under-cover the true parameter.  Instead, we can construct a conservative estimate of the variance by noting that the efficient influence function values of $\bbE \{ \tau_{cide}(X; \delta)^2 \}$ and $\bbE \{ \tau_{cide} (X; \delta) \}^2$ have non-negative covariance when $\bbV \{ \tau_{cide} (X; \delta) \} = 0$ (which is stated formally in Proposition \ref{prop:cons_var} in the appendix).  This suggests a simple way to conservatively estimate the variance of $\widehat \psi_n$, and construct a valid confidence interval with
\begin{equation} \label{eq:ci_null}
    \widehat \psi_n \pm \Phi^{-1} ( 1 - \alpha / 2) \sqrt{ \frac{\widehat \sigma_1^2 + \widehat \sigma_2^2}{n} }
\end{equation}
where
\begin{align}
    \widehat \sigma_1^2 &= \widehat \bbV_n \bigg[ 2 \widehat \omega \big( \widehat \mu_1 - \widehat \mu_0 \big) \left\{ \widehat \omega \widehat \varphi + \widehat \phi \big( \widehat \mu_1 - \widehat \mu_0 \big) \right\} + \left\{ \widehat \omega \big( \widehat \mu_1 - \widehat \mu_0 \big) \right\}^2 \bigg] \text{, and } \label{eq:hatsigma1} \\
    \widehat \sigma_2^2 &= \bbV_n \bigg[ \bbP_n \Big\{ \widehat \omega \widehat \varphi + \widehat \phi \big( \widehat \mu_1 - \widehat \mu_0 \big) + \widehat \omega \big( \widehat \mu_1 - \widehat \mu_0 \big) \Big\} \cdot \Big\{ \widehat \omega \widehat \varphi + \widehat \phi \big( \widehat \mu_1 - \widehat \mu_0 \big) + \widehat \omega \big( \widehat \mu_1 - \widehat \mu_0 \big) \Big\} \bigg] \label{eq:hatsigma2}
\end{align} 
are, respectively, consistent estimators of the variance of the estimators in eq. \eqref{eq:exp_sq_est} and \eqref{eq:sq_exp_est} for $\bbE \{ \tau_{cide}(X; \delta)^2 \}$ and $\bbE \{ \tau_{cide}(X; \delta) \}^2$.  This confidence interval suggests the following one-sided test for treatment effect heterogeneity:
\begin{equation} \label{eq:test}
    \begin{cases}
    \text{Reject } H_0: \bbV \{ \tau_{cide} (X; \delta) \} = 0 &\text{ if } \widehat \psi_n - \Phi^{-1} ( 1 - \alpha) \sqrt{ \frac{\widehat \sigma_1^2 + \widehat \sigma_2^2}{n} } > 0, \\
    \text{ Fail to reject } H_0: \bbV \{ \tau_{cide} (X; \delta) \} = 0 &\text{otherwise.}
    \end{cases}
\end{equation}
This test controls Type I error at the appropriate level, as shown in the following result.
\begin{restatable}{proposition}{propcoverage}\label{prop:coverage}
Under Assumptions \ref{asmp:cons} and \ref{asmp:exch}, Assumption \ref{asm:boundedness} from Theorem \ref{thm:fixed_mod}, and Assumption \ref{asmp:bounded_eif} from Theorem \ref{thm:vcide_conv}, if
$$
\Big( \lVert \widehat \mu - \mu \rVert + \lVert \widehat \pi - \pi \rVert \Big)^2 = o_\bbP \left( \frac{1}{\sqrt{n}} \right),
$$
then the asymptotic Type I error rate of the test in \eqref{eq:test} is less than or equal to $\alpha$.
\end{restatable}

\begin{remark}
In the causal inference literature, at least two other solutions have been proposed for constructing confidence intervals when an estimator converges to a degenerate distribution.  Our approach is similar to that of \cite{williamson2021general}, where they focus on testing variable importance.  \cite{luedtke2019omnibus} propose a different approach - they derive the higher order influence function for their parameter, and construct an associated estimator that achieves $n^{-1}$ convergence under $n^{-1/4}$ conditions on the nuisance function estimators.
\end{remark}


\begin{remark}
When we do not have knowledge of the true parameter value, and we want to construct a valid confidence interval (rather than conduct a test), we can combine the confidence intervals in \eqref{eq:ci_alternative} and \eqref{eq:ci_null} and construct a valid confidence interval with
$$
\widehat \psi_n \pm \Phi^{-1} (1 - \alpha / 2) \sqrt{ \frac{\max(\widehat \sigma^2, \widehat \sigma_1^2 + \widehat \sigma_2^2  )}{n} }. 
$$
\end{remark}

In the Appendix, we illustrate several simulations that demonstrate the properties of the Projection-Learner and I-DR-Learner.  In short, the Projection-Learner achieves correct coverage for the projection parameter and the I-DR-Learner achieves oracle efficiency when the nuisance functions are estimated well enough.  In the next section, we apply these estimators to real ICU data, and demonstrate how they can uncover interesting phenomena that would be obscured by looking at effects with deterministic interventions, like the ATE and the CATE.

\section{Data Analysis of the Effect of Intensive Care Unit Admission on Mortality} \label{sec:data_analysis}

In this section we illustrate the I-DR-Learner and the estimator for the V-CIDE by analyzing data from the (SPOT)light prospective cohort study in which investigators collected data on intensive care unit (ICU) transfers and mortality.  This data is a cohort study collected between November 1st, 2010 and December 31st, 2011 of $13,011$ patients with deteriorating health who were assessed for critical care unit admission across 49 National Health Service hospitals in the UK \citep{harris2018impact, keele2019does}.  

\medskip

Previous literature has considered whether admission to the ICU reduces mortality \citep{gabler2013mortality, renaud2009association}, where the relevant exposure of interest is a binary indicator for whether someone was admitted to the ICU.  Recent analyses have estimated the ATE or used ICU bed availability as an instrumental variable to estimate the local average treatment effect (LATE) \citep{keele2019does}.  Flexible estimation of the ATE finds that the ICU is harmful, whereas estimates for the LATE find a null effect, albeit with wide confidence intervals.  However, arguably, this situation is ideal for conditional effect estimation with incremental interventions.  First, the relevant counterfactual interventions where everyone is sent to the ICU or no one is sent to the ICU may not be feasible (e.g., the ICU might not have capacity to admit everyone), but an intervention where it is made more or less likely that people are sent to the ICU could be feasible.  Second, one might expect a priori that the positivity assumption is violated, in the sense that some patients - depending on their condition - may be almost certain to be admitted or never be admited to the ICU.  Indeed, this is validated by the data, as shown in Figure 2; thus, an intervention that does not require positivity is desirable for this application.  Finally, understanding effect heterogeneity would be of great interest in this application, since it may be the case that the ICU is helpful for some patients while unhelpful or even harmful for others.

\subsection{Data}

The data contains 28-day mortality as an outcome variable and a binary indicator for whether someone was admitted to the ICU.  The data also contains detailed demographic, physiological, comorbidity, and mortality information for all patients.  In terms of demographic information, the data includes age, sex, septic diagnosis (0/1), and peri-arrest (0/1).  In terms of physiology data, there are three risk scores: the ICNARC physiology score \citep{harrison2007new}, the NHS National Early Warning score \citep{williams2012national}, and the Sepsis-related Organ Failure Assessment score \citep{vincent1996sofa}.  Finally, the data also records the patient's existing level of care at assessment and recommended level of care after assessment, which were defined using the UK Critical Care Minimum Dataset levels of care.  We used all these covariates in our analysis, and also included ICU bed availability, which is a binary measure of whether $< 4$ ICU beds were available at the time of assessment.

\subsection{Method}

We consider the counterfactual 28-day mortality rate if we increased or decreased the odds of ICU admission according to an incremental intervention.  We use the I-DR-Learner to nonparametrically estimate the CIE and the CIDE over the ICNARC physiology score.  We focus on the ICNARC score because it is a measure of the health risk of the patient, and a natural question is whether the ICU affects healthier and sicker patients differently.  Then, we estimate the V-CIDE to test for treatment effect heterogeneity across a continuum of policies.  The nuisance functions $\widehat \pi$ and $\widehat \mu$ were estimated with random forests via the \texttt{ranger} package in \texttt{R} \citep{wright2017fast}.  The I-DR-Learner second stage regression was estimated with a smoothing spline via the \texttt{mgcv::gam} function in \texttt{R} \citep{wood2012mgcv}.  \texttt{R} code demonstrating how our analyses were implemented is provided in Section \ref{sec:r-code} of the Appendix.

\subsection{Results}

Figure \ref{fig:prop_count} shows estimated propensity scores by ICNARC score, which confirms prior intuition that positivity might be violated with this data, since for most ICNARC scores there are estimated propensity scores very near 0 and 1.  Figure \ref{fig:3d_results} shows that the CIE varies across $\delta$ for all ICNARC scores.  Estimated counterfactual 28-day mortality is lowest under the observed treatment process  (when $\delta = 1$), and increases when the odds of ICU admission increase ($\delta > 1$) or decrease ($\delta < 1$).  This suggests that the NHS ICU admission protocol during the study was optimal or close to optimal over the class of incremental interventions, and interventions that make it significantly more or less likely for people to be admitted to the ICU could lead to higher mortality rates.  We also see strong evidence that the CIE varies across ICNARC score, and mortality increases as the ICNARC score increases.  This agrees with what one might expect, since the ICNARC is a risk measure where a higher ICNARC score denotes a patient with a higher risk of death.  However, this does not necessarily suggest treatment effect heterogeneity, since one would need to consider a contrast between two levels of the CIE to understand effect heterogeneity.  

\medskip

Figure \ref{fig:smiles} shows the CIE across $\delta$ for four ICNARC scores (0, 15, 30, and 40).  Examining only four curves shows more clearly that the shape of the CIE at each ICNARC score is very similar, suggesting that perhaps there is little heterogeneity.  Figure \ref{fig:smiles} also illustrates a further nuance.  Previous work has estimated the ATE and found that mortality rates would be higher if everyone were admitted to the ICU versus if no one were admitted \citep{keele2019does}. Taken at face value, this suggests that hospitals ought to send fewer people to the ICU; however, due to positivity issues in the data, ATE estimates are likely invalid.  The difference between the  endpoints of the curves in Figure \ref{fig:smiles} (i.e., $\tau_{cice}(v; \delta_u = 5, \delta_l = 0.2)$) suggests a similar conclusion to that implied by ATE estimates, since the mortality rate at $\delta = 5$ is higher than at $\delta = 0.2$.  But, by examining the curve across the spectrum of interventions, one would instead conclude that sending fewer people to the ICU would increase mortality rates as compared to the status quo ($\delta = 1$).  Therefore, our analysis validates previous research - in the sense that it estimates mortality to be lower when no one is admitted to the ICU, compared to everyone is admitted - but it also suggests a different practical implication, since one would conclude from our analysis that sending no one to the ICU is worse than maintaining the status quo.  This highlights how examining a spectrum of interventions can be more informative than examining a contrast like the ATE.

\medskip

Meanwhile, Figure \ref{fig:cide} shows the CIDE across ICNARC score for five $\delta$ values, and shows that the CIDE is generally very near to zero, and is only significantly different from zero at a few points across $\delta$ and ICNARC score.  Figure \ref{fig:vcide} shows there is significant treatment effect heterogeneity across ICNARC score with $95$\% confidence intervals, but that the magnitude of the effect is very small, since the estimate for the V-CIDE is very close to zero for all $\delta$ values.

\begin{figure}
    \centering
    \includegraphics[width = 0.9\textwidth]{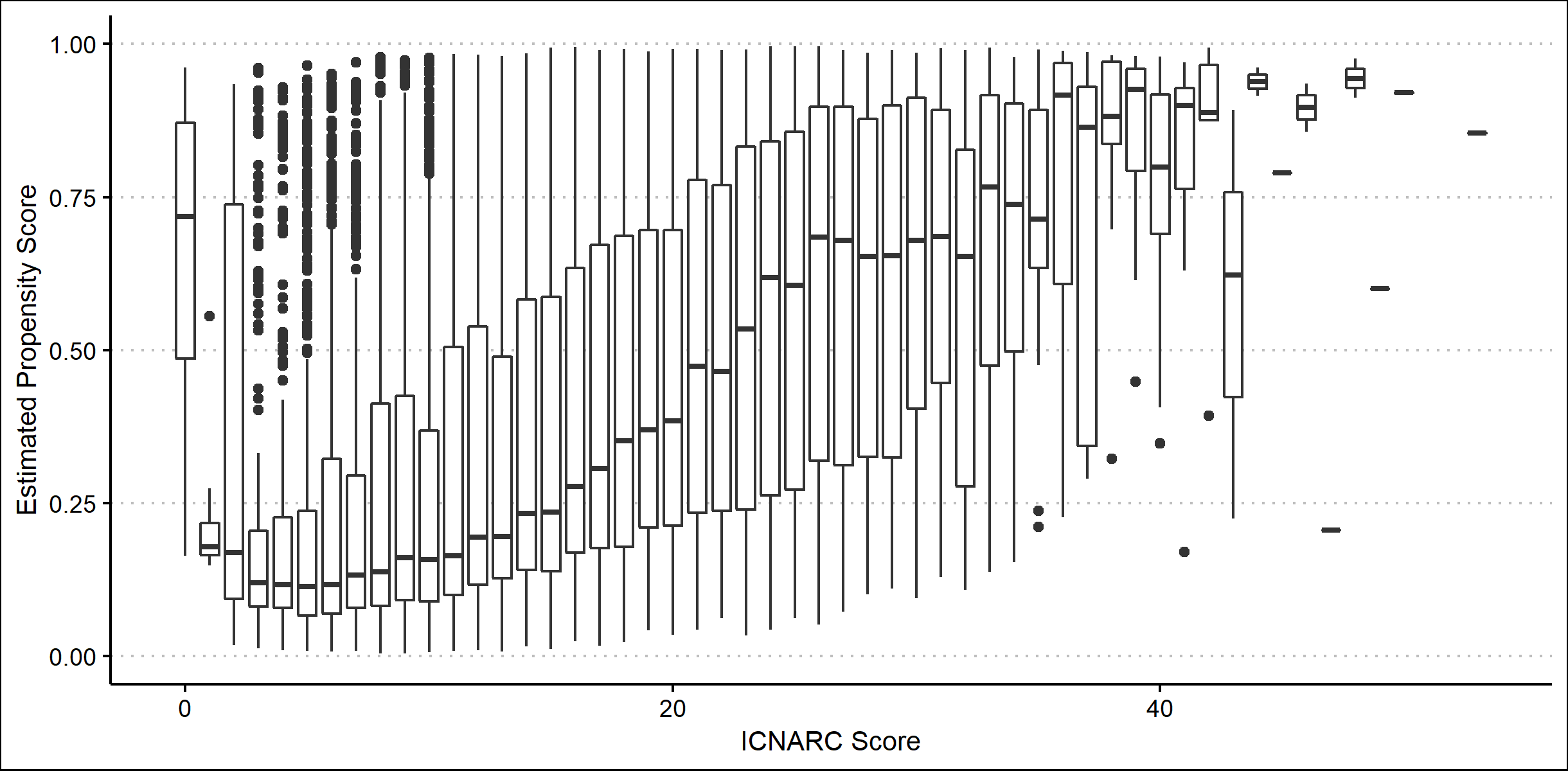}
    \caption{Propensity Scores by ICNARC Score}
    \label{fig:prop_count}
\end{figure}

\begin{figure}[ht]
    \centering
    
    \begin{subfigure}{.5\textwidth}
        \centering
        \includegraphics[width=.9\linewidth]{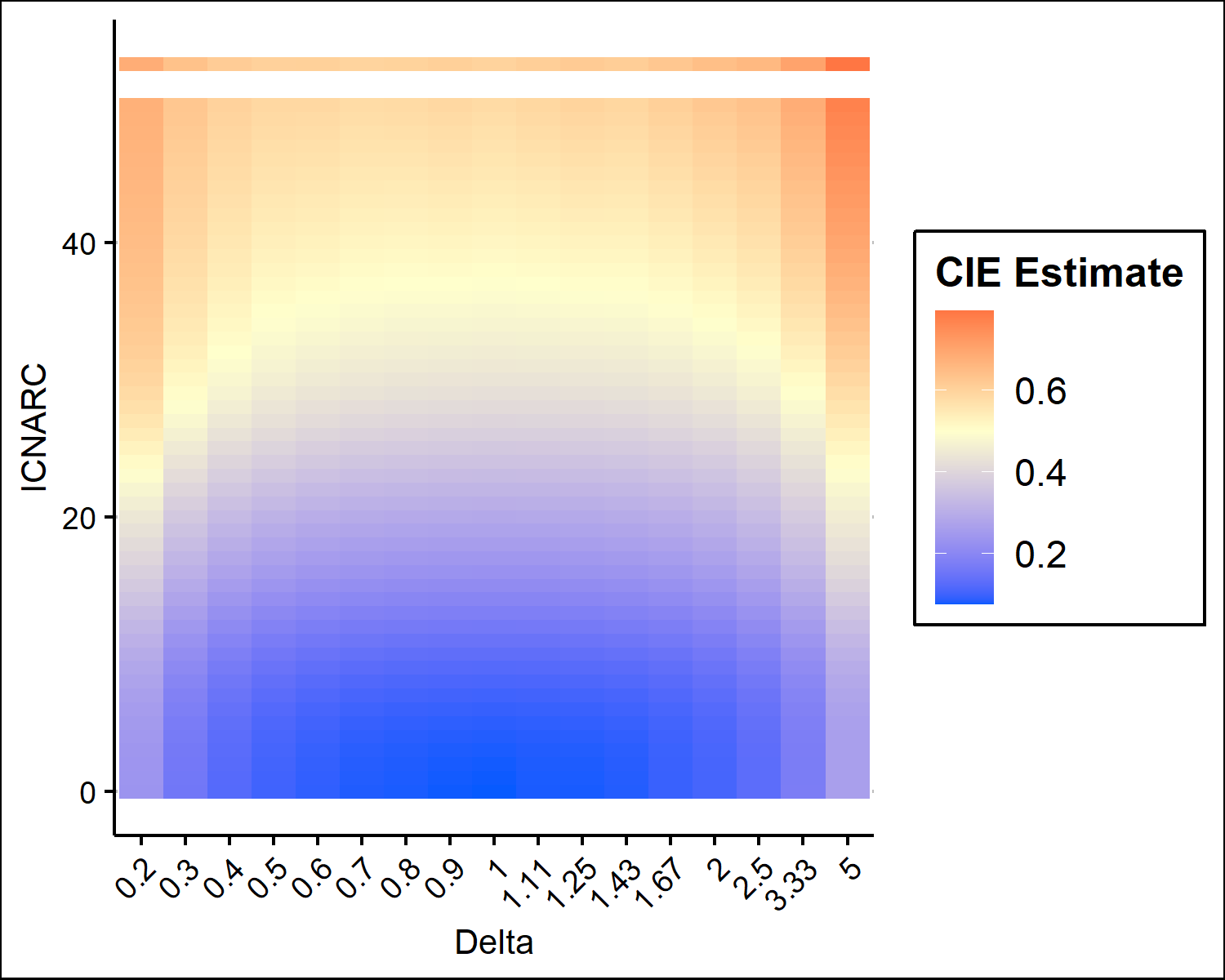}
        \caption{Predicted CIE for all ICNARC scores}
        \label{fig:3d_results}
    \end{subfigure}%
    \begin{subfigure}{.5\textwidth}
        \centering
        \includegraphics[width=.9\linewidth]{smiles_plot.png}
        \caption{Predicted CIE for select ICNARC scores}
        \label{fig:smiles}
    \end{subfigure}
    
    \caption{Predicted Conditional Incremental Effect by $\delta$ and ICNARC Score}
    \label{fig:data_analysis}
\end{figure}

\begin{figure}[h]
    \centering
    \includegraphics[width=.9\linewidth]{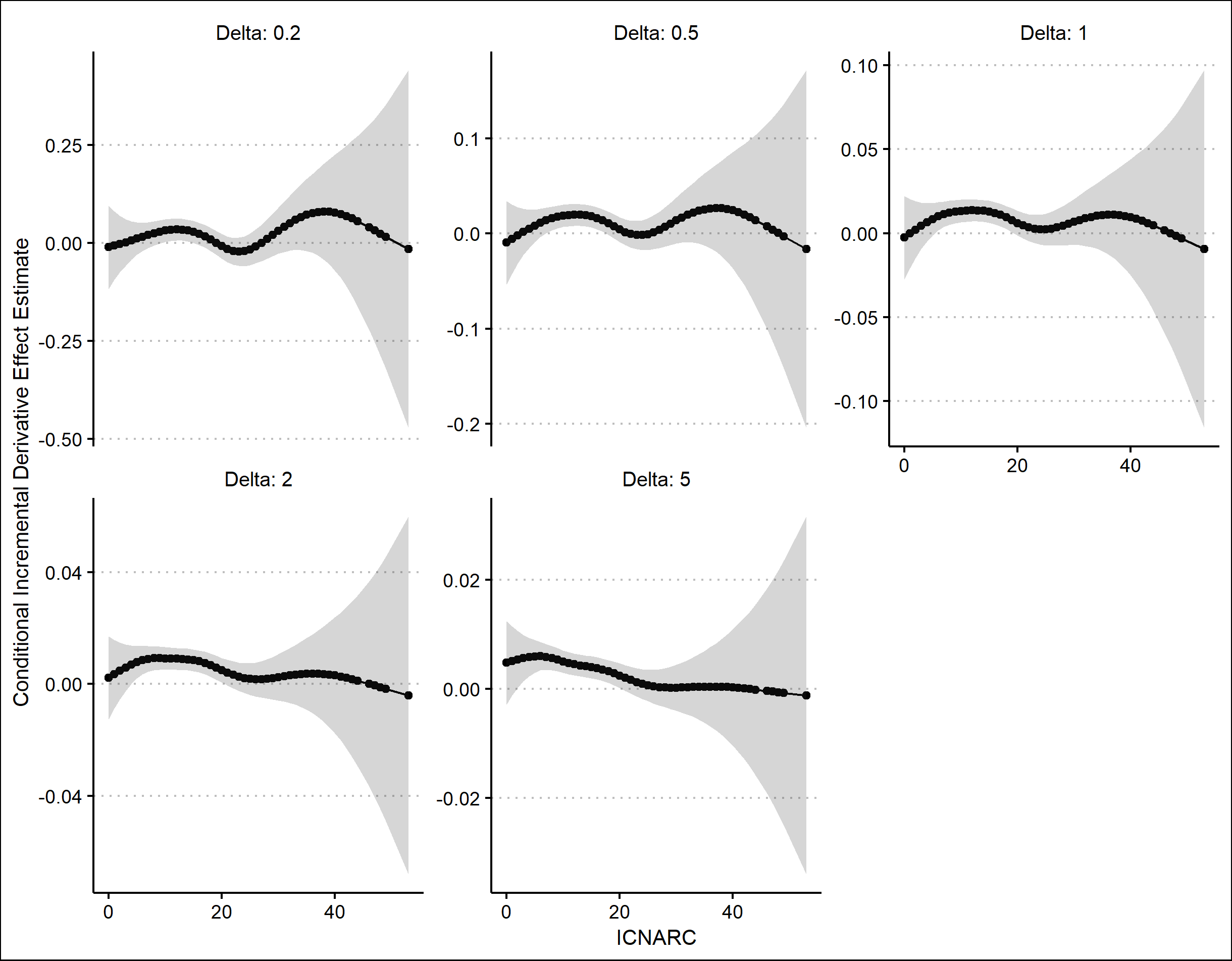}
    \caption{Predicted CIDE versus ICNARC score over $\delta$}
    \label{fig:cide}
\end{figure}

\begin{figure}[h]
    \centering
    \includegraphics[width = .9\linewidth]{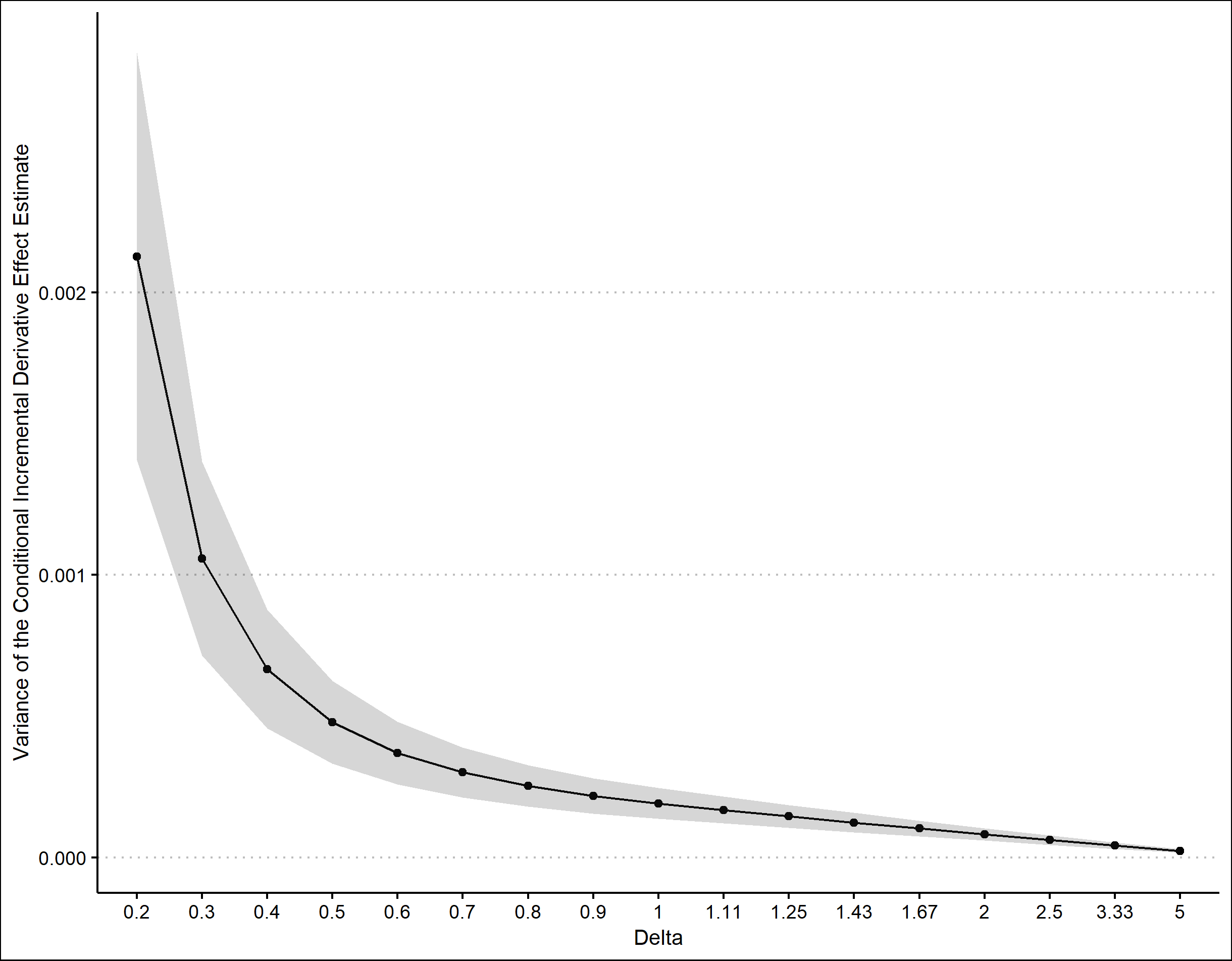}
    \caption{V-CIDE versus $\delta$}
    \label{fig:vcide}
\end{figure}

\newpage
\section{Discussion} \label{sec:discussion}

In this paper, we introduced three conditional effects based on incremental propensity score interventions - the conditional incremental effect (CIE), the conditional incremental contrast effect (CICE) and the conditional incremental derivative effect (CIDE). We proposed two estimators, the Projection-Learner and the I-DR-Learner, which can be used to estimate any of the three conditional effects.  We showed that the Projection-Learner, a projection estimator, achieves parametric efficiency under weak $n^{-1/4}$ conditions on the nuisance function estimators and that the I-DR-Learner, a nonparametric estimator, achieves oracle efficiency under similarly weak conditions.  We also proposed a fourth effect, the variance of the CIDE (V-CIDE), which is a one-dimensional summary of effect heterogeneity.  For the V-CIDE, we proposed a new estimator also with double robust style properties, and outlined methods for inference and testing for treatment effect heterogeneity.  

\medskip

Finally, we illustrated our methods with a real data analysis of the effect of ICU admission on mortality conditional on a patient's risk score.  This analysis demonstrated that estimating counterfactual mean outcomes across a spectrum of incremental interventions can be more informative than just estimating the average treatment effect.  We found evidence that the average treatment effect is positive, suggesting that sending no one to the ICU is better that sending everyone to the ICU in terms of average mortality rates.  However, by examining the spectrum of incremental interventions, we found that average mortality is lowest under the observed treatment process, and mortality would increase if patients were either more or less likely to be admitted to the ICU, suggesting that maintaining the status quo is optimal.    Further, we found that there is indeed statistically significant treatment effect heterogeneity across patient risk scores, but the magnitude of heterogeneity is small. 

\medskip

Here, we proposed conditional incremental effect estimators with the simplest data generating setup - one time point and binary treatment.  There are several natural extensions of this work to more complex frameworks, such as (i) time-varying data, (ii) incremental parameters that can depend on covariate data or past data, and (iii) multi-valued or continuous treatments with different stochastic interventions.  Since positivity violations are almost guaranteed with time-varying data or multi-valued or continuous treatment, it would also be important to understand how nonparametric estimators behave and how projection estimators might be utilized to approximate ATE-style effects when positivity is violated.

\section*{Acknowledgements}

The authors thank Kathryn Haderlein-McClean, Nick Kissel, Iv\'{a}n D\'{i}az, Eli Ben-Michael, Larry Wasserman, Matteo Bonvini, and the Causal Inference Reading Group at Carnegie Mellon University for helpful discussion and comments, and Luke Keele for guidance on the (SPOT)light study data.

\newpage

\section*{References}
\vspace{-0.5in}
\bibliographystyle{plainnat}
\bibliography{references}

\newpage
\appendix

\section{Stability Condition for Theorem \ref{thm:i-dr-learner}}

In this section, we state the stability condition invoked in Section \ref{sec:i-dr-learner} and Theorem \ref{thm:i-dr-learner}.  This stability condition is described in detail in Section 3 of \cite{kennedy2020towards}, and can be viewed as a form of stochastic equicontinuity for nonparametric regression.

\begin{definition} \label{def:stability} \emph{(Stability)} Suppose $D_1 = \{ Z_i \}_{i=1}^{n}$ and $D_2 = \{ Z_i \}_{i=n+1}^{2n}$ are independent training and estimation samples of $n$ observations where $X \subset Z$ are covariates (e.g., $Z_i = (X_i, A_i, Y_i)$.  Let 
\begin{enumerate}
    \item $\widehat f(z) = \widehat f (z; D_1)$ be an estimate of some function of the data, $f(z)$, using the training data $D_1$,
    \item $\widehat b(x) = \widehat b(x; D_1) \equiv \bbE \{ \widehat f (X) - f(X) \mid D_1, X = x \}$, the conditional bias of the estimator $\widehat f$,
    \item $\widehat{\bbE}_n (Y \mid X = x)$ denote a generic regression estimator that regresses outcomes $(Y_{n+1}, ..., Y_{2n})$ on predictors $(X_{n+1}, ..., X_{2n})$ in the estimation sample $D_2$.
\end{enumerate}
Then, the regression estimator $\widehat{\bbE}_n$ is defined as \emph{stable} (with respect to a distance metric $d$) if 
$$
\frac{\widehat{\bbE}_n \{ \widehat f(Z) \mid X = x \} - \widehat{\bbE}_n \{ f(Z) \mid X = x\} - \widehat{\bbE}_n \{ \widehat b (X) \mid X = x\} }{\sqrt{\bbE \left( \left[ \widehat{\bbE}_n \{ f(Z) \mid X = x \} - \bbE \{ f(Z) \mid X = x\} \right]^2 \right)}} \inprob 0
$$
wherever $d(\widehat f, f) \inprob 0$
\end{definition}

Definition \ref{def:stability} says that the difference between the regression estimate with estimated outcomes ($\widehat{\bbE}_n \{ \widehat f(Z) \mid X = x\}$) and the oracle regression ($\widehat{\bbE}_n \{ f(Z) \mid X = x \}$) converges to zero appropriately fast.  This definition can be viewed as a generalization of the classic stochastic equicontinuity condition
$$
\frac{(\bbP_n - \bbP)(\widehat f - f)}{1 / \sqrt{n}} \inprob 0.
$$
where $\bbP_n (\widehat f - f)$ is replaced by $\widehat{\bbE}_n \{ \widehat f(Z) \mid X = x \} - \widehat{\bbE}_n \{ f(Z) \mid X = x\}$, $\bbP (\widehat f - f)$ is replaced by the conditional bias term, $\widehat{\bbE}_n \{ \widehat b (X) \mid X = x\}$, and the denominator $1 / \sqrt{n}$ is replaced by the pointwise RMSE of the oracle estimator, $\sqrt{\bbE \left( \left[ \widehat{\bbE}_n \{ f(Z) \mid X = x \} - \bbE \{ f(Z) \mid X = x\} \right]^2 \right)}$.   This stability condition is satisfied by linear smoothers, as is demonstrated in \cite{kennedy2020towards} Theorem 1, and may be satisfied by more classes of estimators.

\section{Estimator for the variance of the conditional incremental deriative effect when the conditioning covariate is a strict subset of all covariates} 

In this section, we briefly outline an estimator for the variance of the conditional incremental derivative effect (V-CIDE) when the conditioning covariates $V$ are a strict subset of all covariates $X$ (i.e., $V \subset X$) and discuss the convergence properties of the associated estimator.  As a reminder, the variance of the CIDE is identified by
$$
\bbV \{ \tau_{cide} (V; \delta) \} = \bbV \left( \bbE \left[ \frac{\pi(X) \{ 1 - \pi(X) \}}{\{ \delta \pi(X) + 1 - \pi(X) \}^2} \{ \mu(1, X) - \mu(0, X) \} \ \Big| \ V = v \right] \right).
$$
By the definition of the variance and iterated expectation,
$$
\mathbb{V} \{ \tau_{cide}(V; \delta) \} = \bbE \{ \tau_{cide}(V; \delta)^2 \} - \bbE \{ \tau_{cide}(V; \delta) \}^2 = \bbE \{ \tau_{cide}(V; \delta)^2 \} - \bbE \{ \tau_{cide}(X; \delta) \}^2.
$$
The squared expectation term, $\bbE \{ \tau_{cide}(X; \delta) \}^2$, is the same as what appears when $V = X$, and therefore can be estimated as in eq. \eqref{eq:sq_exp_est} in Section \ref{sec:heterogeneity}.  The expected square term is new, and we derive the efficient influence function in the following result:
\begin{lemma}
    Under Assumptions \ref{asmp:cons} and \ref{asmp:exch}, the un-centered efficient influence function for $\bbE \{ \tau_{cide}(V; \delta)^2 \}$ is
    $$
    \tau_{cide}(V; \delta)^2 + 2 \tau_{cide}(V; \delta) \big\{ \omega \varphi + \phi \big( \mu_1 - \mu_0 \big) + \omega \big( \mu_1 - \mu_0 \big) - \tau_{cide}(V; \delta)^2 \big\},
    $$
    where $\mu_a = \mu(a, X)$, and $\omega = \omega(X; \delta), \varphi = \varphi(Z)$, and $\phi = \phi(Z; \delta)$ as defined in equations \eqref{eq:omega}, \eqref{eq:varphi}, and \eqref{eq:phi}.
\end{lemma}
\noindent This result suggests the following estimator:
\begin{equation} \label{eq:exp_sq_v}
    \bbP_n \left[ \widehat \tau_{cide}(V; \delta)^2 + 2 \widehat \tau_{cide}(V; \delta) \left\{ \widehat \omega \widehat \varphi + \widehat \phi \big( \widehat \mu_1 - \widehat \mu_0 \big) + \widehat \omega \big( \widehat \mu_1 - \widehat \mu_0 \big) - \widehat \tau_{cide}(V; \delta)^2 \right\} \right].
\end{equation}
And, combined with the results in Section \ref{sec:i-dr-learner} and Section \ref{sec:heterogeneity}, suggests the following estimator for the V-CIDE:
\begin{algorithm} \label{alg:vcide_v} \emph{(V-CIDE Estimator when $V \subset X$)} Assume as inputs $(D_{1}, D_2)$, which denote two independent samples of $n$ observations of $Z_i = (X_i, A_i, Y_i)$, then:
\begin{enumerate}
    \item On the training data $D_1$, estimate the nuisance functions $\widehat \mu (0, X)$, $\widehat \mu(1, X)$ and $\widehat \pi(X)$.
    \item On the estimation data $D_2$, estimate the un-centered influence function values $\widehat \xi (Z; \delta)$ using the models $\widehat \mu$ and $\widehat \pi$ from step 1, where $\widehat \xi (Z; \delta)$ is defined in \eqref{eq:xihat} if the conditional effect of interest is $\tau_{cide}$, and analogously for $\tau_{cie}$ and $\tau_{cice}$ in equations \eqref{eq:ie_eif} and \eqref{eq:ice_eif}.
    \item In the estimation sample $D_2$, regress $\widehat \xi (Z; \delta)$ on the conditioning covariates $V$ to obtain the estimate
    $$
    \widehat{\tau}_{cide} (v; \delta) = \widehat \bbE_n \left\{ \widehat \xi (Z; \delta) \mid V = v \right\}.
    $$    
    \item On the estimation data, estimate $\bbV \{ \tau_{cide}(V; \delta \}$ per equations \eqref{eq:exp_sq_v} and \eqref{eq:sq_exp_est}, plugging in the estimates for $\widehat \mu$, $\widehat \pi$, $\widehat \tau_{cide}$ from above.
\end{enumerate}
\end{algorithm}

This estimator satisfies a similar double robustness condition to the estimator outlined in Section \ref{sec:heterogeneity}, but includes a dependence on $\widehat \tau_{cide}(V; \delta)$. 

\begin{restatable}{theorem}{thmvcideconv-v}\label{thm:vcide_conv_v}
Let $\widehat \psi_n$ denote the estimator from Algorithm \ref{alg:vcide_v}.  Unser Assumptions \ref{asmp:cons}, and \ref{asmp:exch}, Assumption \ref{asm:boundedness} from Theorem \ref{thm:fixed_mod}, and Assumption \ref{asmp:bounded_eif} from Theorem \ref{thm:vcide_conv}, if
$$
\lVert \widehat \pi - \pi \rVert \Big( \lVert \widehat \mu - \mu \lVert + \lVert \widehat \pi - \pi \rVert \Big) + \lVert \widehat \tau_{cide} - \tau_{cide} \rVert^2,
$$
then
$$
\sqrt{n} \Big[ \widehat \psi_n - \bbV \{ \tau_{cide}(V; \delta) \} \Big] \indist N(0, \sigma^2)
$$
where 
\begin{align}
    \sigma^2 = \bbV \bigg[ &\tau_{cide}^2 + 2  \tau_{cide} \left\{  \omega  \varphi +  \phi \big(  \mu_1 -  \mu_0 \big) +  \omega \big(  \mu_1 -  \mu_0 \big) -  \tau_{cide}^2 \right\} \nonumber \\
    &- \bbE \Big\{ \omega \varphi + \phi \big( \mu_1 - \mu_0 \big) + \omega \big( \mu_1 - \mu_0 \big) \Big\} \cdot \Big\{ \omega \varphi + \phi \big( \mu_1 - \mu_0 \big) + \omega \big( \mu_1 - \mu_0 \big) \Big\} \bigg]  \label{eq:pop_vcide_var_v}
\end{align}
where $\tau_{cide} \equiv \tau_{cide} (V; \delta)$, and $\mu_a = \mu(a, X)$, $\omega = \omega(X; \delta), \varphi = \varphi(Z)$, and $\phi = \phi(Z; \delta)$ as defined in equations \eqref{eq:omega}, \eqref{eq:varphi}, and \eqref{eq:phi}. 
\end{restatable}
The theorem shows that the estimator for the V-CIDE satisfies a version of double robustness under relatively weak conditions.  The result shows that our estimator attains $n^{-1/2}$ convergence to the V-CIDE under model-agnostic $n^{-1/4}$ convergence rates for the nuisance function estimators and the I-DR-Learner. This is a different result from Theorem \ref{thm:i-dr-learner}, since it is required that the CIDE is estimated at $n^{-1/4}$ rates, but it is no longer required that $\widehat \mu$ is estimated at $n^{-1/4}$ rates.  As discussed in the body of the paper, $n^{-1/4}$ rates are achievable with nonparametric estimators under suitable smoothness or sparsity. 

\medskip

Like the estimator from Algorithm \ref{alg:vcide}, the estimator in Algorithm \ref{alg:vcide_v} converges to a degenerate distribution when the V-CIDE equals zero.  As discussed at the end of Section \ref{sec:heterogeneity}, we can construct a valid test for any treatment effect heterogeneity by overestimating the variance of the estimator $\widehat \psi_n$.  When $V \subset X$, we can constrcut the following asymptotically valid $1 - \alpha$ test
\begin{equation}
    \begin{cases}
    \text{Reject } H_0: \bbV \{ \tau_{cide} (V; \delta) \} = 0 &\text{ if } \widehat \psi_n - \Phi^{-1} ( 1 - \alpha) \sqrt{ \frac{\widehat \sigma_1^2 + \widehat \sigma_2^2}{n} } > 0, \\
    \text{ Fail to reject } H_0: \bbV \{ \tau_{cide} (V; \delta) \} = 0 &\text{otherwise.}
    \end{cases}
\end{equation}
where
\begin{align*}
    \widehat \sigma_1^2 &=  \widehat{\bbV}_n \bigg[ \widehat \tau_{cide}^2 + 2  \widehat \tau_{cide} \left\{  \widehat \omega \widehat \varphi + \widehat \phi \big( \widehat \mu_1 - \widehat \mu_0 \big) + \widehat \omega \big(  \widehat \mu_1 - \widehat \mu_0 \big) - \widehat \tau_{cide}^2 \right\} \bigg] \text{, and }  \\
    \widehat \sigma_2^2 &= \bbV_n \bigg[ \bbP_n \Big\{ \widehat \omega \widehat \varphi + \widehat \phi \big( \widehat \mu_1 - \widehat \mu_0 \big) + \widehat \omega \big( \widehat \mu_1 - \widehat \mu_0 \big) \Big\} \cdot \Big\{ \widehat \omega \widehat \varphi + \widehat \phi \big( \widehat \mu_1 - \widehat \mu_0 \big) + \widehat \omega \big( \widehat \mu_1 - \widehat \mu_0 \big) \Big\} \bigg].
\end{align*}

\newpage
\section{Code for ICU Data Analysis} \label{sec:r-code}

The code below was used for the ICU data analysis in Section \ref{sec:data_analysis}. For brevity, we omitted the package loading and figure generation, so this code is not perfectly reproducible.  However, the code does demonstrate how to estimate second stage regressions for the I-DR-Learner and how to calculate the V-CIDE.  The first section of the code estimates the propensity scores and the CIE to create Figures \ref{fig:prop_count} and \ref{fig:data_analysis}, while the second section of the code estimates the CIDE and the V-CIDE to create Figures \ref{fig:cide} and \ref{fig:vcide}.

\medskip

To estimate the propensity scores and the CIE, we adapted the \texttt{npcausal::ipsi()} function to output the propensity scores and the estimated influence function values \citep{kennedy2021npcausal}.

\begin{lstlisting}[language=R]
###############################################################################
### Author: Alec McClean
### Purpose: Data analysis with ICU data
###############################################################################

### Package loading (omitted)

################################
### Calculate CIE
################################

source("ipsi_updated.R") # An amended version of npcausal::ipsi() that outputs
# estimated propensity scores and influence function values.  Also omitted

# -------------------------------
### Load and clean data

icu <- read.csv("../Data/icuData.csv") 

# Use row number as ID
icu %<>% rename(id = X)

# Outcome variable: dead28
icu %<>% select(-dead7, -dead90) 

# Change variables to factors
icu %<>% mutate_at(vars(site, male, sepsis_dx:winter, v_cc1:v_cc_r5), as.factor)

# Deltas
DELTAS <- seq(0.2, 0.9, 0.1)
DELTAS <- c(DELTAS, 1, rev(1 / DELTAS))


# ----------------------------------------------------------
### Calculate influence function values for pseudo outcomes

results <- ipsi_update(y = icu$dead28,
                       a = icu$icu_bed,
                       id = icu$id,
                       x.trt = icu %>% select(-id, -dead28, -icu_bed),
                       x.out = icu %>% select(-id, -dead28, -icu_bed),
                       time = rep(1, nrow(icu)),
                       fit = "rf",
                       delta.seq = DELTAS,
                       nsplits = 2,
                       return_ifvals = TRUE)

ifvals <- as.data.frame(results$ifvals)
colnames(ifvals) <- DELTAS
ifvals$id <- icu$id 
ifvals %<>% left_join(icu %>% select(id, age, icnarc_score, news_score, sofa_score))
ifvals$split <- results$splits

ifvals %<>% 
  gather(delta, ifval, `0.2`:`5`) %>% 
  mutate(delta = round(as.numeric(delta), 2))

point_estimates <- results$res %>% 
  select(delta = increment, pt = est) %>%
  mutate(delta = round(delta, 2))

# Calculate un-centered influence function values as pseudo outcomes 
ifvals %<>% left_join(point_estimates) %>%
  mutate(pseudo = pt + ifval)

# -------------------------------------
# Second stage regressions for CIE

for (DELTA in unique(ifvals$delta)) {
  
  dat <- ifvals %>% filter(delta == DELTA)
  
  mod <- mgcv::gam(pseudo ~ s(icnarc_score), data = dat)
  dat$pred <- predict(mod)
  dat$upr <- dat$pred + 1.96 * predict(mod, se.fit = TRUE)$se.fit
  dat$lwr <- dat$pred - 1.96 * predict(mod, se.fit = TRUE)$se.fit
  ifvals %<>% filter(delta != DELTA) %>% bind_rows(dat)
    
}

# -------------------------------
# Generate Figures 2 and 3

# Code to generate Figures 2 and 3 omitted.
# Figures were generated with ggplot2

####################################
### Estimate CIDE and VCIDE 
####################################

# ----------------------------------------------------------
### Calculate influence function values for pseudo outcomes

FOLDS <- 2
icu$fold <- sample(1:FOLDS, size = nrow(icu), replace = T)
output <- data.frame()

for (FOLD in 1:FOLDS) {
  
  test <- icu %>% filter(fold == FOLD)
  train <- icu %>% filter(fold != FOLD)
  
  pimod <- ranger(icu_bed ~ ., dat = train %>% select(-id, -dead28, -fold))
  mumod <- ranger(dead28 ~ ., dat = train %>% select(-id, -fold))
  
  test$pihat <- 
    predict(pimod,
            data = test %>% select(-id, -dead28, -fold))$predictions
  
  test$mu1hat <- 
    predict(mumod,
            data = test %>% select(-id, -fold) %>% mutate(icu_bed = 1))$predictions
  
  test$mu0hat <- 
    predict(mumod,
            data = test %>% select(-id, -fold) %>% mutate(icu_bed = 0))$predictions
  
  output %<>% bind_rows(test)
}

ifvals <- data.frame()
for (DELTA in DELTAS) {
  
  temp <- output %>% mutate(delta = DELTA)
  
  temp %<>% mutate(
    omega = pihat * (1 - pihat) / ((DELTA * pihat + 1 - pihat)^2),
    tau = mu1hat - mu0hat,
    eif_omega = 
      (icu_bed - pihat) * 
      (1 / (DELTA * pihat + 1 - pihat)^3 - (2 * DELTA * pihat) /
        (DELTA * pihat + 1 - pihat)^2),
    eif_tau = (icu_bed / pihat) * 
      (dead28 - mu1hat) + ((1 - icu_bed) / (1 - pihat)) * (dead28 - mu0hat),
    plugin = omega * tau,
    eif_terms = omega * eif_tau + eif_omega * tau,
    ifval_cide = eif_terms + plugin,
  )
  
  temp$ifval_vcide_t2 <- 
    mean(temp$eif_terms + temp$plugin) * (temp$eif_terms + temp$plugin)
  
  ifvals %<>% bind_rows(temp)
}


# ------------------------------------------
### Second stage regressions for CIDE

temp <- data.frame()
for (DELTA in unique(ifvals$delta)) {
  
  for (FOLD in 1:max(ifvals$fold)) {

    dat <- ifvals %>% filter(delta == DELTA) %>% filter(fold == FOLD)
    
    mod <- mgcv::gam(ifval_cide ~ s(icnarc_score), data = dat)
    dat$pred <- predict(mod)
    dat$upr <- dat$pred + 1.96 * predict(mod, se.fit = TRUE)$se.fit
    dat$lwr <- dat$pred - 1.96 * predict(mod, se.fit = TRUE)$se.fit
    
    temp %<>% bind_rows(dat)
  }
  
}

ifvals <- temp
rm(temp)
gc()

# --------------------------------------------
### Calculate V-CIDE across delta 

### Calculate IF values 
ifvals %<>% mutate(
  ifval_vcide_t1 = (pred * pred) + 2 * pred * (eif_terms + plugin - (pred * pred)),
  ifval_vcide = ifval_vcide_t1 - ifval_vcide_t2
)

vcide <- ifvals %>% group_by(delta) %>%
  summarize(pt_est = mean(ifval_vcide),
            sd_est = sd(ifval_vcide) / sqrt(n()),
            conservative_est = sqrt((var(ifval_vcide_t1) + var(ifval_vcide_t2)) / n())) %>%
  ungroup() %>%
  mutate(lower = pt_est - 1.96 * pmax(sd_est, conservative_est),
         upper = pt_est + 1.96 * pmax(sd_est, conservative_est))


# --------------------------------------
# Generate Figures 4 and 5

# Code to generate Figures 4 and 5 omitted
# Figures were generated with ggplot2
\end{lstlisting}

\section{Simulations for the Projection-Learner and I-DR-Learner}

In this section we study the performance of the Projection-Learner and I-DR-Learner for estimating the conditional incremental contrast effect (CICE) with $\delta_u = 5$ and $\delta_l = 0.2$.  As a reminder, the CICE is defined as
$$
\tau_{cice} (v; \delta_u = 5, \delta_l = 0.2) \equiv \bbE \left( Y^{Q_5} - Y^{Q_{0.2}} \mid V = v \right),
$$
which corresponds to the difference between the counterfactual mean outcomes when the odds of treatment are multiplied by $5$ minus the counterfactual mean outcome when the odds of treatment are divided by $5$.  For all the analyses, we simulate $1,000$ times a dataset of size $n = 1,000$.  In each dataset, we have $\{ (X_i, A_i, Y_i) \}$ from $i = 1$ to $n = 1000$ where $X \in \bbR, A \in \{0, 1\}$ and $Y \in \bbR$.  In this case, we condition on the only covariate $X$, so $V = X$.  For each dataset, we specify a quadratic CICE,
$$
\tau_{cice} (\delta_u = 5, \delta_l = 0.2; X) = 1 + 0.5 X - 0.2X^2,
$$
and then the CATE is defined implicitly as
$$
\tau_{cate} (X) = \frac{\tau_{cice} (X; \delta_u = 5, \delta_l = 0.2)}{q \{ \pi(X); \delta_u = 5 \} - q \{ \pi(X); \delta_l = 0.2 \}} \\ 
$$
which follows from the identification result in Proposition \ref{prop:id_t1}.  Each simulated dataset $\{ (X_i, A_i, Y_i) \}_{i=1}^{n = 1000}$ is then constructed in the following manner:
\begin{align}
    X &\sim Unif(-4, 4) \\
    \pi(X) &= \expit \left( \frac{X}{2} \right) \\
    \mu (0, X) &= \one(X < -3) \cdot 2 + \one(X > -2) \cdot 2.55 + \one(X > 0) \cdot -2 \\
    &\hspace{0.5in} + \one(X > 2) \cdot 4 + \one(X > 3) \cdot -1 \\
    \mu (1, X) &= \mu (0, X) + \tau_{cate} (X) \\
    A &\sim Bernoulli(\pi) \\
    Y &\sim A \cdot \mu (1, X) + (1 - A) \cdot \mu (0, X) + N(0, 1)
\end{align}
The data generating process is also illustrated in Figure \ref{fig:data}. The covariate data $X$ is one dimensional and uniform over $[-4, 4]$.  The propensity score, shown in the top panel of Figure \ref{fig:data}, follows a logistic model, which remains within reasonable bounds on the support of $X$, since $\pi(-4) \approx 0.12$ and $\pi(4) \approx 0.88$.  The outcome regressions are complicated discontinuous functions, and are shown in the middle panel of Figure \ref{fig:data}.  The treatment $A$ and outcome $Y$ are defined implicitly from the propensity score and regression functions.  The second panel also shows the CATE, which is a smooth function.  This is because the CICE and the propensity scores are smooth functions.  The CICE is shown in the bottom panel of Figure \ref{fig:data}.

\medskip

We simulated estimates for the propensity scores and regression functions by adding noise, parameterized by $\alpha$, to the true nuisance functions:
\begin{align}
    \widehat \pi (X) &\sim \expit \left[  \logit \left\{ \pi(X) + N(n^{-\alpha_\pi}, n^{-2 \alpha_\pi}  ) \right\} \right]  \\
    \widehat{\mu} (a, X) &\sim \mu (a, X) + N \left[ \{ \max_x \mu(a, x) - \min_x \mu (a, x) \} \cdot n^{-\alpha_\mu}, \{ \max_x \mu(a, x) - \min_x \mu (a, x) \}^2 
    \cdot n^{-2 \alpha_\mu} \right] 
\end{align}
The $\alpha$ parameter allows us to control how well the nuisance functions are ``estimated''.  For example, when $\alpha = 0.1$, this corresponds to estimating a nuisance function with error converging at $n^{-1/10}$.  We scale the error for the regression functions $\mu$ by the range of the regression function values; this is purely a computing trick so that the error for neither nuisance function dominates the other, and this does not affect the convergence rates of the estimators.

\begin{figure}[H]
    \centering
    \includegraphics[height = 4in]{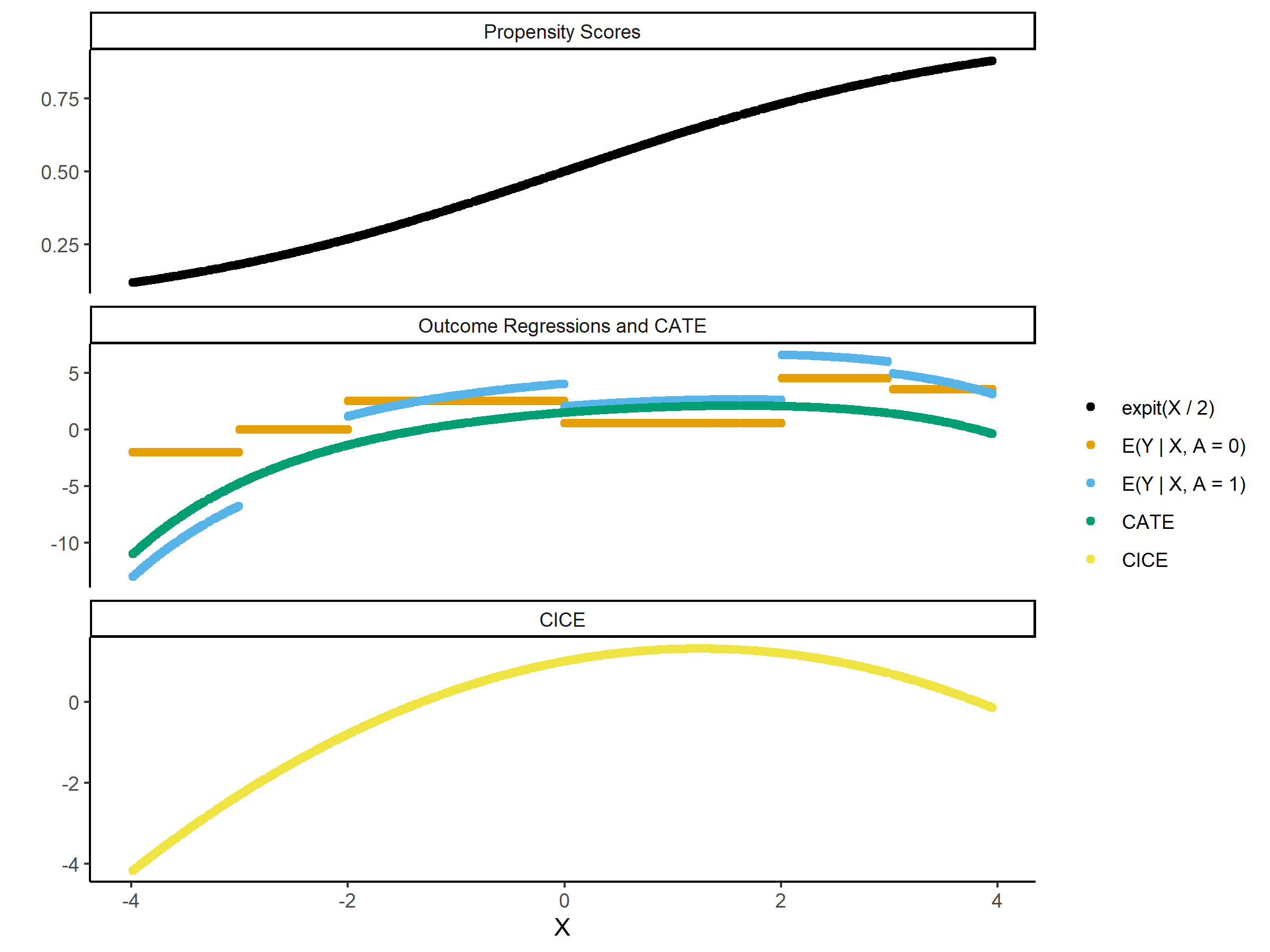}
    \caption{Data Generating Process}
    \label{fig:data}
\end{figure}

First, we compare the I-DR-Learner to the oracle estimator (``Oracle I-DR-Learner'') and a baseline learner (``Baseline CICE'') in terms of integrated mean squared error (MSE).  The oracle estimator constructs the true influence function value from $\mu (a, X)$ and $\pi(X)$ and regresses them against $X$.  Both the oracle estimator and the I-DR-Learner use the \texttt{smooth.spline} function in \texttt{R} for the second stage regression.  The baseline estimator is a plug-in estimator which calculates
$$
\widehat \tau_{cice} (X; \delta_u = 5, \delta_l = 0.2) = \{ \widehat \mu (1, X) - \widehat \mu (0, X)\} \cdot \Big[ q \{ \widehat \pi(X); \delta_u = 5 \} - q \{ \widehat \pi(X); \delta_l = 0.2 \} \Big].
$$
This is motivated by causal identification, such as the result in Proposition \ref{prop:id_t1}, and does not make use of the efficient influence function for the relevant average effect. The baseline estimator for the CATE was previously examined in the literature, and has been referred to as the `T-Learner' \citep{kunzel2019metalearners}.

\medskip

The results of these simulations are summarized in Figure \ref{fig:mse}.  Each different panel corresponds to a different convergence rate $\alpha_\mu$ for estimating $\widehat \mu$, and $\alpha_\mu$ increases from left to right.  The x-axis shows the convergence rate $\alpha_\pi$ for estimating $\widehat \pi$, and the y-axis shows the integrated MSE for each estimator.  Finally, each estimator is denoted by a different color, and the points and whiskers show the sample mean and 95\% confidence interval for the MSE over $1,000$ simulations.  The oracle estimator performs the best, which we would expect since it has access to the true nuisance functions.  The I-DR-Learner performs the next best, and its error approaches that of the oracle estimator as $\min(2 \alpha_\pi, \alpha_\pi + \alpha_\mu)$ increases.  Asymptotically, Theorem \ref{thm:i-dr-learner} dictates that the I-DR-Learner will attain the oracle convergence rate when $\min(2 \alpha_\pi, \alpha_\pi + \alpha_\mu) \geq 0.5$.  Here, we see that even when this is the case, there is still some discrepancy between the oracle estimator and the I-DR-Learner. This is because these errors were evaluated on finite samples of size $1,000$.  If we considered larger and larger sample sizes, the discrepancy between the DR learners and the oracle estimators would shrink and approach zero asymptotically.

\begin{figure}
    \centering
    \includegraphics[width = \linewidth]{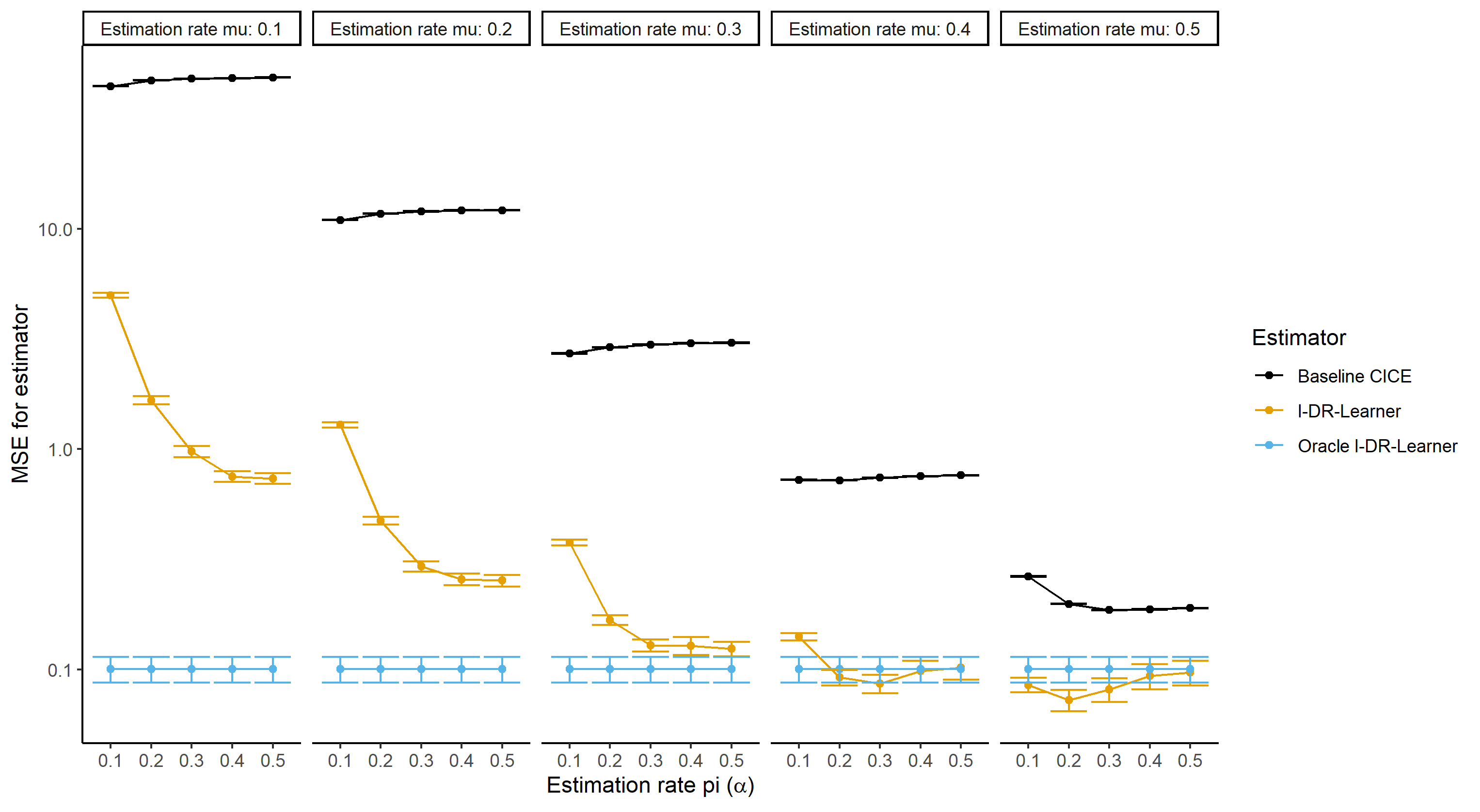}
    \caption{Comparing CICE estimators}
    \label{fig:mse}
\end{figure}

The baseline learners fare the worst.  The furthest right panel, where $\alpha_\mu = 0.5$, shows what would happen if we could estimate correct parametric models for $\mu_0$ and $\mu_1$.  In that case, we expect the baseline learner to perform as well as the oracle estimator and I-DR-Learner asymptotically, so the gap between the baseline learner and the other learners might be slightly disconcerting.  However, this is again a story of asymptotics  - if we increased $n$, the gap between the three estimators in the far right panel would decrease towards zero. For the baseline learner, we expect there to be a trend in $\pi$ as well as $\mu$, since the error of the CICE baseline learner is additive in the errors of the nuisance functions
$$
\lVert \widehat \tau_{base} - \tau_{cice} \rVert \lesssim \lVert \widehat \mu - \mu \rVert + \lVert \widehat \pi - \pi \rVert^2.
$$
However, whenever one nuisance function has much higher error, this sum effectively becomes a maximum.  In the left-most panels, the error $\lVert \widehat \mu - \mu \rVert $ is orders of magnitude larger than the error $\lVert \widehat \pi - \pi \rVert$, so any effect from changing $\alpha_\pi$ is swamped by the error from $\mu$. In the far right panel, once $\alpha_\mu = 0.5$, the error in  $\pi$ matters, and the MSE of the estimator decreases as $\alpha_\pi$ increases.

\subsection{Coverage of the Projection-Learner}

In this subsection, we outline results for the Projection-Learner.  Specifically, we show that the Projection-Learner achieves approximately correct coverage for the true coefficients in the model.  The true model is
$$
\tau_{cice} (X; \delta_u = 5, \delta_l = 0.2) = 1 + 0.5 X - 0.2X^2,
$$
and the working model is
$$
g(\beta; X) = \beta^\ast + \beta_1 X + \beta_2 X^2
$$
Since the working model is well-specified, the Projection-Learner estimates the true coefficients.  Figure \ref{fig:coverage} shows the coverage of 95\% confidence intervals constructed for each coefficient using the sandwich variance as in Corollary \ref{cor:lim_dist}.  When the nuisance function estimators have large error, such that $\alpha_\pi + \alpha_\mu < 0.5$ or $2 \alpha_\pi < 0.5$ then the confidence intervals have poor coverage.  As the nuisance function estimators improve, the coverage of the Projection-Learner becomes very close to 95\%.  Generally, the coverage is still slightly less than 95\%, and this is because the datasets are only of size $n = 1000$, so there is still some discrepancy between the true distribution of the estimator and a Gaussian distribution.

\begin{figure}
    \centering
    \includegraphics[width = \linewidth]{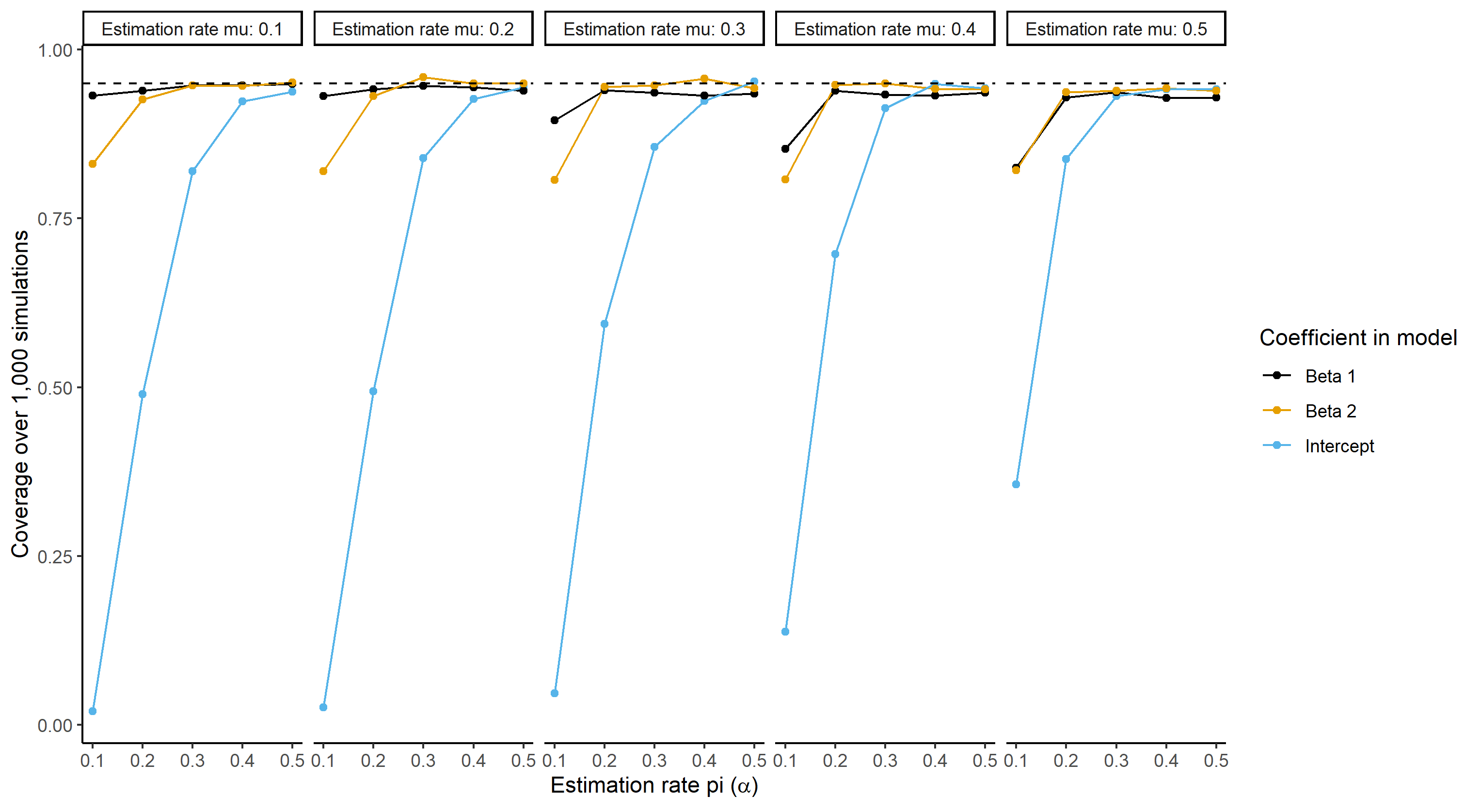}
    \caption{Coverage of Projection-Learner confidence intervals for coefficient estimates}
    \label{fig:coverage}
\end{figure}




\newpage
\section{Proofs for results in Section \ref{sec:setup}}

\propderiv*

\begin{proof}
We have
\begin{align*}
    \tau_{cide} (v; \delta) &= \frac{\partial}{\partial t} \bbE \left\{  Y^{Q_t} \mid V = v \right\} \bigg|_{t = \delta} \\
    &= \frac{\partial}{\partial t} \bbE \left[ \left\{ \frac{t \pi(X)}{t \pi(X) + 1 - \pi(X)} \right\} \mu(1, X) + \left\{ \frac{1 - \pi(X)}{t \pi(X) + 1 - \pi(X)} \right\} \mu(0, X) \mid V = v \right] \bigg|_{t = \delta} \\
    &= \bbE \left( \left[ \frac{\pi(X) \{ 1 - \pi(X) \}}{\{ \delta \pi(X) + 1 - \pi(X) \}^2} \right] \cdot \{ \mu(1, X) - \mu(0, X) \} \mid V = v \right)
\end{align*}
where the first line follows by definition, the second by Assumptions \ref{asmp:cons} and \ref{asmp:exch}, and the third and final line by exchanging expectation and derivative, taking the derivative with respect to $t$, rearranging, and setting $t = \delta$.
\end{proof}

\section{Proofs for results in Section \ref{sec:conditional_effects}}

\lemcideeif*

\begin{proof}
We prove the result by showing that $\bbE \{ \tau_{cide} (v; \delta) \}$ satisfies a Von Mises expansion where $\xi (Z; \delta)$ is the influence function and there is a second order remainder term.  At the end of the proof, we will relate this to smooth parametric submodels and submodel scores to prove that $\xi(Z; \delta)$ is the efficient influence function.

\medskip

Let $\xi(\bbP) \equiv \xi(Z; \delta)$ and $\psi(\bbP) = \bbE \{ \tau_{cide} (v; \delta) \}$.  Then, by a Von Mises expansion,
\begin{equation} \label{eq:vm_deriv}
    \psi(\overline \bbP) = \psi (\bbP) + \int_{\mathcal{Z}} \xi(\overline \bbP) d(\overline \bbP - \bbP) + R_2 (\overline \bbP, \bbP)
\end{equation}
where $\bbP$ and $\overline \bbP$ are two different distributions at which the functional $\psi$ is evaluated.  Rearranging, we can see that
\begin{equation}
    R_2(\overline \bbP, \bbP) = \psi(\overline \bbP) - \psi (\bbP) -  \int_{\mathcal{Z}} \xi(\overline \bbP) d(\overline \bbP - \bbP) = \int_{\mathcal{Z}} \xi(\overline \bbP) - \xi (\bbP) d \bbP \equiv \bbE_{\bbP} \{ \xi(\overline \bbP) - \xi (\bbP) \}
\end{equation}
where $\bbE_{\bbP}$ denotes expectation under the distribution $\bbP$.  By the definition of $\xi$,
\begin{align*}
    R_2(\overline \bbP, \bbP) &= \bbE_\bbP \left( \left[ \frac{\overline \pi(X) \{ 1 - \overline \pi(X) \}}{ \{ \delta \overline\pi(X) + 1 - \overline\pi(X) \}^2} \right] \cdot \left[ \frac{A}{\overline\pi(X)} \Big\{ Y - \overline \mu(1, X) \Big\} - \frac{1-A}{1-\overline \pi(X)} \Big\{ Y - \overline \mu(0, X) \Big\} \right] \right) \\
    &\hspace{0.2in} + \bbE_\bbP \left( \left[ \frac{1}{\{ \delta \overline \pi(X) + 1 - \overline \pi(X) \}^2} - \frac{2\delta \overline \pi (X)}{\{ \delta \overline \pi(X) + 1 - \overline \pi(X) \}^3} \right] \cdot \Big\{ A - \overline \pi(X) \Big\} \cdot \{ \overline \mu(1, X) - \overline \mu(0, X) \} \right) \\
    &\hspace{0.2in} + \bbE_\bbP \left( \left[ \frac{\overline \pi(X) \{ 1 - \overline \pi(X) \}}{ \{ \delta \overline \pi(X) + 1 - \overline \pi(X) \}^2} \right] \cdot \{ \overline \mu(1, X) - \overline \mu(0, X) \} \right) \\
    &- \bbE_\bbP \left( \left[ \frac{\pi(X) \{ 1 - \pi(X) \}}{ \{ \delta \pi(X) + 1 - \pi(X) \}^2} \right] \cdot \left[ \frac{A}{\pi(X)} \Big\{ Y - \mu(1, X) \Big\} - \frac{1-A}{1-\pi(X)} \Big\{ Y - \mu(0, X) \Big\} \right] \right) \\
    &\hspace{0.2in} - \bbE_\bbP \left( \left[ \frac{1}{\{ \delta \pi(X) + 1 - \pi(X) \}^2} - \frac{2\delta\pi (X)}{\{ \delta \pi(X) + 1 - \pi(X) \}^3} \right] \cdot \Big\{ A - \pi(X) \Big\} \cdot \{ \mu(1, X) - \mu(0, X) \} \right) \\
    &\hspace{0.2in} - \bbE_\bbP \left( \left[ \frac{\pi(X) \{ 1 - \pi(X) \}}{ \{ \delta \pi(X) + 1 -  \pi(X) \}^2} \right] \cdot \{ \mu(1, X) - \mu(0, X) \} \right) 
\end{align*}
By iterated expectations and rearranging,
\begin{align}
    R_2 (\overline \bbP, \bbP) &= \bbE_\bbP \left ( \overline \omega(X; \delta) \left[ \left\{ \frac{\pi(X) - \overline \pi(X)}{\overline \pi(X)} \right\} \{ \mu(1, X) - \overline \mu(1, X) \} + \left\{ \frac{\pi(X) - \overline \pi(X)}{1 - \overline \pi(X)}  \right\} \{\mu(0, X) - \overline \mu(0, X) \} \right] \right) \nonumber \\
    &\hspace{0.1in} + \bbE_{\bbP} \left( \{ \overline \mu(1, X) - \overline \mu(0, X) \} \Big[ \bbE (\overline \phi(Z; \delta) \mid X) + \overline \omega(X; \delta) - \omega(X; \delta) \Big]  \right) \nonumber \\
    &\hspace{0.1in} + \bbE_{\bbP} \left( \left\{ \overline \omega(X; \delta) - \omega(X; \delta) \right\} \Big[ \mu(1, X) - \overline \mu(1, X) - \{ \mu(0, X) - \overline \mu(0, X) \} \Big]  \right), \label{eq:r2_deriv}
\end{align}
where
\begin{align*}
    \omega(X; \delta) &= \frac{\pi(X) \{ 1 - \pi(X)\}}{\{ \delta \pi(X) + 1 - \pi(X) \}^2}, \text{ and} \\
    \phi(Z; \delta) &= \left[ \frac{1}{\{ \delta \pi(X) + 1 - \pi(X) \}^2 } - \frac{2 \delta \pi(X) }{\{ \delta \pi(X) + 1 - \pi(X) \}^3} \right] \cdot \left\{ A - \pi(X) \right\}.
\end{align*}
It can be shown that $R_2(\overline \bbP, \bbP)$ is second order, since
\begin{align*}
R_2(\overline \bbP, \bbP) = 
\bbE_\bbP &\left[ g_0(X) \Big\{ \pi(X) - \widehat \pi(X) \Big\} \left\{ \mu(0, X) - \widehat \mu(0, X) \right\} + g_1 (X) \Big\{ \pi(X) - \widehat \pi(X) \Big\} \Big\{ \mu(1, X) + \widehat \mu(1, X) \Big\}
\right. \\
& \left. + h(X) \left\{ \pi(X) - \widehat \pi(X) \right\}^2 \right]
\end{align*}
We show this in the postscript to this proof.  Here, we provide some intuition.   It is clear that the first line of \eqref{eq:r2_deriv} is already a second order.  The second line of \eqref{eq:r2_deriv} is second order because $\phi$ is the un-centered efficient influence function of $\omega$, and so the second multiplicand on the second line is the error term for estimating $\omega$, which we would expect to be second order.  The third line of \eqref{eq:r2_deriv} is more intuitively second order because it is the product of the errors of two plug-ins, and that product can expressed as a product of errors.

\medskip

Now, we relate $\xi(Z; \delta)$ back to scores of smooth parametric submodels.  Recall from semiparametric efficiency theory that the nonparametric efficiency bound for a functional is given by the supremum of Cramer-Rao lower bounds for that functional across smooth parametric submodels \citep{bickel1993efficient, van2002semiparametric}.  The efficient influence function is the unique mean-zero function $\xi$ that is a valid submodel score satisfying pathwise differentiability; i.e.,
\begin{equation} \label{eq:pathdiff}
    \frac{d}{d \epsilon} \psi (\bbP_{\epsilon}) \bigg|_{\epsilon = 0} = \int_{\mathcal{Z}} \xi(\bbP) \left( \frac{d}{d\epsilon} \log d\bbP_{\epsilon} \right) \bigg|_{\epsilon = 0} d\bbP
\end{equation}
for $\bbP_\epsilon$ any smooth parametric submodel.  To see that $\xi (Z; \delta)$ is the efficient influence function for $\bbE \{ \tau_{cide} (v; \delta) \}$, observe that the von Mises expansion in \eqref{eq:vm_deriv} implies
\begin{align*}
    \frac{d}{d \epsilon} \psi(\bbP_{\epsilon}) &= \frac{d}{d\epsilon} \left( \psi(\bbP) - \int_{\mathcal{Z}} \xi(\bbP) d(\bbP - \bbP_\epsilon) -  R_2 (\bbP, \bbP_\epsilon) \right) \\
    &= \frac{d}{d\epsilon} \int_{\mathcal{Z}}  \xi(\bbP) d(\bbP_\epsilon - \bbP) -  \frac{d}{d\epsilon}  R_2 (\bbP, \bbP_\epsilon)  \\
    &= \int_{\mathcal{Z}}  \xi(\bbP)  \frac{d}{d\epsilon} d\bbP_\epsilon -  \frac{d}{d\epsilon}  R_2 (\bbP, \bbP_\epsilon) \\
    &= \int_{\mathcal{Z}}  \xi(\bbP)  \left(\frac{d}{d\epsilon} \log d\bbP_\epsilon \right) d\bbP_\epsilon -  \frac{d}{d\epsilon}  R_2 (\bbP, \bbP_\epsilon)
\end{align*}
with $R_2$ defined in \eqref{eq:r2_deriv}, and where the second line follows because $\psi(\bbP)$ does not depend on $\epsilon$, the third because $\int \xi(\bbP) d\bbP = 0$, and the fourth and final line since $\frac{d}{d\epsilon} \log d\bbP_\epsilon  = \frac{1}{d\bbP_\epsilon} \frac{d}{d\epsilon} d\bbP_\epsilon$.  Evaluating this expression at $\epsilon = 0$, we have
$$
\int_{\mathcal{Z}}  \xi(\bbP)  \left(\frac{d}{d\epsilon} \log d\bbP_\epsilon \right) d\bbP_\epsilon -  \frac{d}{d\epsilon}  R_2 (\bbP, \bbP_\epsilon) \bigg|_{\epsilon = 0} = \int_{\mathcal{Z}}  \xi(\bbP)  \left(\frac{d}{d\epsilon} \log d\bbP_\epsilon \right) \bigg|_{\epsilon = 0} d\bbP -  0 
$$
since
$$
\frac{d}{d\epsilon} R_2 (\bbP, \bbP_\epsilon) \Big|_{\epsilon = 0} = 0
$$
which shows that $\xi$ satisfies the property in \eqref{eq:pathdiff}.  The last equation involving $R_2$ follows because $R_2$ consists of only second-order products of errors between $\bbP_\epsilon$ and $\bbP$.  Therefore, the derivative is composed of a sum of terms, each of which is a product of either a derivative term that may not equal zero and an error term involving the differences of components of $\bbP$ and $\bbP_\epsilon$, which will be zero when $\epsilon = 0$ since $\bbP = \bbP_\epsilon$. 

\medskip

Since the model is nonparametric, the tangent space is the entire Hilbert space of mean-zero finite-variance functions, and so there is only one influence function satisfying \eqref{eq:pathdiff} and it is the efficient one \citep{tsiatis2006semiparametric}.   Therefore, $\xi(Z; \delta)$ is the efficient influence function for $\bbE \{ \tau_{cide}(V; \delta) \}$.
\end{proof}

Below, we show the algebra for why $R_2 (\overline \bbP, \bbP)$ is second order as stated above. Starting with the third line in \eqref{eq:r2_deriv}, and omitting arguments, we have
\begin{align*}
    \overline \omega - \omega  &= \frac{\overline \pi ( 1- \overline \pi)}{(\delta \overline \pi + 1 - \overline \pi)^2 } - \frac{\pi(1 - \pi)}{(\delta \pi + 1 - \pi)^2} \\
    &= \frac{\overline \pi (1 - \overline \pi) (\delta \pi + 1 - \pi)^2 - \pi(1 - \pi)(\delta \overline \pi + 1 - \overline \pi)^2}{(\delta \pi + 1 - \pi)^2 (\delta \overline \pi + 1 - \overline \pi)^2} \\
    &= (\pi - \overline \pi) \left\{ \frac{(\delta + 1)(\delta - 1) \pi \overline \pi + \pi + \overline \pi - 1}{(\delta \pi + 1 - \pi)^2 (\delta \overline \pi + 1 - \overline \pi)^2} \right\}.
\end{align*}
Therefore, letting $\mu_a = \mu(a, X)$, we have
\begin{align*}
    \{ \overline \omega (X; \delta) - \omega (X; \delta) \} &\cdot \big[ \mu(1, X) - \overline \mu(1, X) - \{ \mu(0, X) - \overline \mu(0, x) \} \big] = \\
    &(\pi - \overline \pi) \left\{ \frac{(\delta + 1)(\delta - 1) \pi \overline \pi + \pi + \overline \pi - 1}{(\delta \pi + 1 - \pi)^2 (\delta \overline \pi + 1 - \overline \pi)^2} \right\} \cdot \big\{ \mu_1 - \overline \mu_1 - ( \mu_0 - \overline \mu_0 ) \big\} \\
    &= \left\{ \frac{(\delta + 1)(\delta - 1) \pi \overline \pi + \pi + \overline \pi - 1}{(\delta \pi + 1 - \pi)^2 (\delta \overline \pi + 1 - \overline \pi)^2} \right\} \left\{ (\pi - \overline \pi) (\mu_1 - \overline \mu_1) \right\} \\
    &+ \left\{ \frac{(\delta + 1)(1 - \delta) \pi \overline \pi + \pi + \overline \pi - 1}{(\delta \pi + 1 - \pi)^2 (\delta \overline \pi + 1 - \overline \pi)^2} \right\} \left\{ (\pi - \overline \pi) (\mu_0 - \overline \mu_0) \right\}.
\end{align*}
For the second line in \eqref{eq:r2_deriv}, we have
\begin{align*}
    \bbE ( \overline \phi \mid X) + \overline \omega - \omega &= \left\{ \frac{1}{( \delta \overline \pi + 1 - \overline \pi )^2 } - \frac{2 \delta \overline \pi }{( \delta \overline \pi + 1 - \overline \pi )^3} \right\} \cdot ( \pi - \overline \pi ) \\
    &+ (\pi - \overline \pi) \left\{ \frac{(\delta + 1)(\delta - 1) \pi \overline \pi + \pi + \overline \pi - 1}{(\delta \pi + 1 - \pi)^2 (\delta \overline \pi + 1 - \overline \pi)^2} \right\} \\
    &= ( \pi - \overline \pi ) \left\{ \frac{1}{( \delta \overline \pi + 1 - \overline \pi )^2 } - \frac{2 \delta \overline \pi }{( \delta \overline \pi + 1 - \overline \pi )^3} + \frac{(\delta + 1)(\delta - 1) \pi \overline \pi + \pi + \overline \pi - 1}{(\delta \pi + 1 - \pi)^2 (\delta \overline \pi + 1 - \overline \pi)^2} \right\} \\
    &= ( \pi - \overline \pi )^2 \left[ \frac{(1 - \delta - \delta^2 - \delta^3) \pi \overline \pi + 1 - 2\delta + (1-\delta)\{ \overline \pi + (1-\delta)\pi \}}{(\delta \overline \pi + 1 - \overline \pi)^3 (\delta \pi + 1 - \pi)^2} \right]. 
\end{align*}
Finally, revisiting the first line of \eqref{eq:r2_deriv},
\begin{align*}
    \overline \omega(X; \delta) \Bigg[ \left\{ \frac{\pi(X) - \overline \pi(X)}{\overline \pi(X)} \right\} \{ \mu(1, X) - \overline \mu(1, X) \} &+ \left\{ \frac{\pi(X) - \overline \pi(X)}{1 - \overline \pi(X)}  \right\} \{\mu(0, X) - \overline \mu(0, X) \} \Bigg] = \\
    &(\pi - \overline \pi) (\mu_1 - \overline \mu) \left( \frac{\overline \omega}{\overline \pi} \right) + (\pi - \overline \pi) (\mu_0 - \overline \mu_0) \left( \frac{\overline \omega}{1 - \overline \pi} \right),
\end{align*}
\subsection*{Efficient Influence Functions for the average incremental effect and the average incremental contrast effect}

By Corollary 2 in \cite{kennedy2019incremental}, the efficient influence function for the average incremental effect is
\begin{align} 
    \xi_{ie}(Z; \delta) &= \frac{\delta \pi(X) \mu(1, X) + \{ 1 - \pi(X) \} \mu(0, X) }{\delta \pi(X) + 1 - \pi(X)} \nonumber \\
    &\hspace{0.2in} + \left\{ \frac{\delta A + 1 - A}{\delta \pi(X) + 1 - \pi(X)} \right\} \left\{ Y - \mu(A, X) \right\} \nonumber \\
    &\hspace{0.2in} + \big\{ \mu (1, X) - \mu(0, X) \big\} \frac{\delta \{ A - \pi(X) \} }{\left\{ \delta \pi(X) + 1 - \pi(X) \right\}^2}, \label{eq:ie_eif}
\end{align}
\noindent and the efficient influence function for the average incremental contrast effect is
\begin{equation} \label{eq:ice_eif}
    \xi_{ice} (Z; \delta_u, \delta_l) = \xi_{ie} (Z; \delta_u) - \xi_{ie} (Z; \delta_l).
\end{equation}

\cormomenteif*

\begin{proof}
Let $\varphi(\bbP) \equiv \varphi(Z; \delta; \beta)$.  The functional $m(\beta)$ satisfies the following von Mises expansion
\begin{equation} \label{eq:von_mises}
    m(\beta, \overline{\bbP}) - m(\beta, \bbP) = \int_{\mathcal{Z}} \varphi(\overline{\bbP}) d(\overline{\bbP} - \bbP) + R_2 (\overline \bbP, \bbP)
\end{equation}
where 
\begin{align*}
    R_2 (\overline{\bbP}, \bbP) &= m(\beta, \overline{\bbP}) - m(\beta, \bbP) - \int_{\mathcal{Z}} \varphi(\overline{\bbP}) d(\overline{\bbP} - \bbP) \\
    &= \int \varphi(\overline{\bbP}) - \varphi (\bbP) d\bbP.
\end{align*}
Following essentially the same logic as in Lemma \ref{lem:cide_eif}, by iterated expectations and rearranging,
\begin{align*}
    R_2 (\overline{\bbP}, \bbP) = \bbE_{\bbP} \Bigg\{ \frac{\partial g(V; \beta)}{ \partial \beta} &\Bigg( \overline \omega(X; \delta) \left[ \left\{ \frac{\pi(X) - \overline \pi(X)}{\pi(X)} \right\} \{ \mu(1, X) - \mu(1, X) \} + \left\{ \frac{\pi(X) - \overline \pi(X)}{1 - \pi(X)}  \right\} \{\mu(0, X) - \overline \mu(0, X) \} \right] \\
    &\hspace{0.1in} + \{ \overline \mu(1, x) - \overline \mu(0, X) \} \Big[ \bbE (\overline \phi(Z; \delta) \mid X) + \overline \omega(X; \delta) - \omega(X; \delta) \Big] \\
    &\hspace{0.1in} + \left\{ \overline \omega(X;\delta) - \omega(X; \delta) \right\} \left[ \mu(1, X) - \overline \mu(1, X) - \{ \mu(0, X) - \overline \mu(0, X) \} \right] \Bigg) \Bigg\}.
\end{align*}
This second order term can be expressed as a product of errors, as is shown in post script to the proof for Lemma \ref{lem:cide_eif}.  Therefore, since our model is nonparametric, $\varphi(Z; \delta, \beta)$ is the efficient influence function for $m(\beta)$.
\end{proof}

\thmfixedmod*

\begin{proof}
This proof follows closely both Lemma 3 from \cite{kennedy2021semiparametric} and Theorem 5.31 of \cite{van2000asymptotic}.  Since $(\widehat{\beta}, \widehat \mu, \widehat \pi)$ is an approximate solution to the empirical moment condition, $\bbP_n \{ \varphi(Z; \widehat{\beta}, \widehat \mu, \widehat \pi) \} = o_{\bbP} (1 / \sqrt{n})$.  Since $(\beta^\ast, \mu^\ast, \pi^\ast)$ is an exact solution to the population moment condition, $\bbP \{ \varphi(Z; \beta^\ast, \mu^\ast, \pi^\ast) \} = 0$.  Combining these two facts,
$$
o_{\bbP} (1 / \sqrt{n}) = \bbP_n \{ \varphi(Z; \widehat{\beta}, \widehat \mu, \widehat \pi) \} - \bbP \{ \varphi(Z; \beta^\ast, \mu^\ast, \pi^\ast) \}.
$$
By adding and subtracting on the right hand side of the equation above, omitting $Z$, and letting $\eta^\ast = (\mu^\ast, \pi^\ast)$ and $\widehat \eta = (\widehat \mu, \widehat \pi)$:
\begin{align*}
    o_{\bbP} (1 / \sqrt{n}) = &\ (\bbP_n - \bbP) \{ \varphi(\beta^\ast, \eta^\ast) \}  \\
    &+ (\bbP_n - \bbP) \{ \varphi(\widehat{\beta}, \widehat \eta) - \varphi(\beta^\ast, \eta^\ast) \} \\
    &+ (\bbP_n - \bbP) \{ \varphi(\beta^\ast, \widehat \eta) - \varphi(\beta^\ast, \eta^\ast) \} \\
    &+ \bbP \{ \varphi(\widehat{\beta}, \widehat \eta) - \varphi(\beta^\ast, \widehat \eta) \} \\
    &+ \bbP \{ \varphi(\beta^\ast, \widehat \eta) - \varphi(\beta^\ast, \eta^\ast) \}.
\end{align*}
The first term appears directly in the statement in the theorem, so we will not manipulate it.  It is a sample average of a fixed function, and so by the central limit theorem it will be asymptotically Gaussian.  The second and third terms are empirical process terms.  The fourth term can be linearized in $(\widehat{\beta} - \beta)$, and we will use this to rearrange and solve for the statement in the theorem. The fifth term captures the nuisance estimation error, and appears implicitly in the statement in the theorem if we define $R_n = \bbP\{ \varphi(Z; \beta^\ast, \widehat \mu, \widehat \pi) - \varphi(Z; \beta^\ast, \mu^\ast, \pi^\ast) \}$.

\medskip

First, we will tackle the second and third terms.  Under the Donsker and consistency conditions in Assumptions \ref{asm:beta_donsker} and \ref{asm:beta_varphi_cons}, the second term is $o_{\bbP} (1 / \sqrt{n})$ by Lemma 19.24 of \cite{van2000asymptotic}.  Further, under the consistency of $\varphi(\widehat \beta, \widehat \eta)$ in Assumption \ref{asm:beta_varphi_cons} and by sample splitting, the third term is $o_{\bbP} (1 / \sqrt{n})$ by Lemma 2 of \cite{kennedy2020sharp}. 

\medskip

The fourth term, by the differentiability of the map $\beta \mapsto \bbP \{ \varphi(\beta, \eta) \}$ in Assumption \ref{asm:diff_map}, can be expressed as
\begin{align*}
    \bbP \{ \varphi(\widehat{\beta}, \widehat{\eta}) - \varphi(\beta^\ast, \widehat{\eta}) \} &= M(\beta^\ast, \widehat{\eta}) (\widehat{\beta} - \beta^\ast) + o_{\bbP} ( \lVert \widehat{\beta} - \beta^\ast \rVert) \\
    &= M(\beta^\ast, \eta^\ast) (\widehat{\beta} - \beta^\ast) + o_{\bbP} ( \lVert \widehat{\beta} - \beta^\ast \rVert)
\end{align*}
where the first line is a first-order Taylor expansion about $\beta^\ast$ and the second line follows by the consistency of $M(\beta^\ast, \widehat{\eta})$ in Assumption \ref{asm:diff_map}.  

\medskip

Bringing everything together,
\begin{align*}
    o_{\bbP} (1 / \sqrt{n}) = &\ (\bbP_n - \bbP) \{ \varphi(\beta^\ast, \eta^\ast) \}  \\
    &+ (\bbP_n - \bbP) \{ \varphi(\widehat{\beta}, \widehat \eta) - \varphi(\beta^\ast, \eta^\ast) \} \\
    &+ (\bbP_n - \bbP) \{ \varphi(\beta^\ast, \widehat \eta) - \varphi(\beta^\ast, \eta^\ast) \} \\
    &+ \bbP \{ \varphi(\widehat{\beta}, \widehat \eta) - \varphi(\beta^\ast, \widehat \eta) \} \\
    &+ \bbP \{ \varphi(\beta^\ast, \widehat \eta) - \varphi(\beta^\ast, \eta^\ast) \} \\
    = &\ (\bbP_n - \bbP) \{ \varphi(\beta^\ast, \eta^\ast) \}  \\
    &+ o_{\bbP} (1 / \sqrt{n}) \\
    &+ o_{\bbP} (1 / \sqrt{n}) \\
    &+ M(\beta^\ast, \eta^\ast) (\widehat{\beta} - \beta^\ast) + o_{\bbP} ( \lVert \widehat{\beta} - \beta^\ast \rVert ) \\
    &+ R_n.
\end{align*}
Re-arranging, we have that:
$$
o_{\bbP}(1 / \sqrt{n}) = (\bbP_n - \bbP) \{ \varphi(Z; \beta^\ast, \eta^\ast) \} + M(\beta^\ast, \eta^\ast) (\widehat{\beta} - \beta^\ast) + o_{\bbP} ( \lVert \widehat{\beta} - \beta^\ast \rVert ) + R_n.
$$
We can re-arrange and by the non-singularity of the derivative matrix $M$ in Assumption \ref{asm:diff_map} we can pre-multiply both sides by $M(\beta^\ast, \eta^\ast)^{-1}$ and see that:
\begin{equation}
    \widehat{\beta} - \beta^\ast = - M(\beta^\ast, \eta^\ast)^{-1} (\bbP_n - \bbP) \{ \varphi(\beta^\ast, \eta^\ast) \} + O_{\bbP} (R_n) + o_{\bbP} ( \lVert \widehat{\beta} - \beta^\ast \rVert ) + o_{\bbP} (1 / \sqrt{n}). \label{eq:penultimate}
\end{equation}
To address the $o_{\bbP} (\lVert \widehat{\beta} - \beta^\ast \rVert)$ term, notice the following: by Assumption \ref{asm:beta_varphi_cons}, $\widehat{\beta} - \beta^\ast = o_{\bbP} (1)$, and by the Central Limit Theorem, $\varphi(Z; \beta^\ast, \eta^\ast) = O_{\bbP} (1 / \sqrt{n})$.  Therefore, rearranging equation \eqref{eq:penultimate} above to put the $(\widehat{\beta} - \beta^\ast)$ terms on the same side, we have:
$$
\widehat \beta - \beta^\ast + o_{\bbP} ( \lVert \widehat \beta - \beta^\ast \rVert )  = O_{\bbP}(1 / \sqrt{n}) + O_{\bbP} (R_n) + o_{\bbP} (1 / \sqrt{n})
$$
so that 
$$
\lVert \widehat{\beta} - \beta^\ast \rVert \{ 1 + o_{\bbP} (1) \} = O_{\bbP}(1 / \sqrt{n} + R_n)
$$
and therefore
$$
\lVert \widehat{\beta} - \beta^\ast \rVert  = O_{\bbP}(1 / \sqrt{n} + R_n)
$$
and
$$
o_{\bbP} ( \lVert \widehat{\beta} - \beta^\ast \rVert ) = o_{\bbP}( O_{\bbP}(1 / \sqrt{n} + R_n)) = o_{\bbP} (1 / \sqrt{n}) + o_{\bbP} (R_n).
$$
Finally, plugging these results back into equation \eqref{eq:penultimate}, we have:
\begin{align*}
    \widehat{\beta} - \beta^\ast &= - M(\beta^\ast, \eta^\ast)^{-1} (\bbP_n - \bbP) \{ \varphi(\beta^\ast, \eta^\ast) \} + O_{\bbP} (R_n) + o_{\bbP} ( \lVert \widehat{\beta} - \beta^\ast \rVert ) + o_{\bbP} (1 / \sqrt{n}) \\
    &= - M(\beta^\ast, \eta^\ast)^{-1} (\bbP_n - \bbP) \{ \varphi(\beta^\ast, \eta^\ast) \} + O_{\bbP} (R_n) + o_{\bbP} (1 / \sqrt{n}) + o_{\bbP} (R_n) + o_{\bbP} (1 / \sqrt{n}) \\
    &= - M(\beta^\ast, \eta^\ast)^{-1} (\bbP_n - \bbP) \{ \varphi(\beta^\ast, \eta^\ast) \} + O_{\bbP} (R_n) + o_{\bbP} (1 / \sqrt{n}).
\end{align*}
And, adding back in all arguments, we can conclude
$$
\widehat \beta - \beta^\ast =  - M(\beta^\ast, \mu^\ast, \pi^\ast)^{-1} (\bbP_n - \bbP) \{ \varphi(\beta^\ast, \mu^\ast, \pi^\ast; Z) \} + O_{\bbP} \left( R_n + o_{\bbP} (1 / \sqrt{n}) \right)
$$
which provides the first statment of the theorem.

\medskip

For the second statement, recall that $R_n = \bbP\{ \varphi(\beta^\ast, \widehat \mu, \widehat \pi; Z) - \varphi(\beta^\ast, \mu^\ast, \pi^\ast; Z) \}$.
\begin{align*}
    \bbP\{ \varphi(Z; \beta^\ast, \widehat \mu, \widehat \pi) - \varphi(Z; \beta^\ast, \mu^\ast, \pi^\ast) \} &= \int_{\mathcal{Z}} 2 \frac{\partial g(v; \beta^\ast)}{\partial \beta} \{ g(v; \beta^\ast) - \widehat \xi (z; \delta) \} d \bbP(z) \\
    &- 2 \int_{\mathcal{Z}} \frac{\partial g(v; \beta^\ast)}{\partial \beta} \{ g(v; \beta^\ast) - \xi_0 (z; \delta) \}  d \bbP(z) \\
    &= \int_{\mathcal{Z}} 2 \frac{\partial g(v; \beta^\ast)}{\partial \beta} \{ \xi_0 (z; \delta) - \widehat{\xi}(z; \delta) \}  d\bbP(z).
\end{align*}
The final result follows by Cauchy-Schwartz and a boundedness condition, which depends on the target estimand.  If the estimand is a projection of the CIDE, then the result follows by equivalent logic to the proof of Lemma \ref{lem:cide_eif} and the first part of Assumption \ref{asm:boundedness}, which says that the estimated CATE is bounded.  If, instead, the estimand is a projection of the CIE or the CICE, the result follows by equivalent logic to the proofs of Lemmas 5 and 6 in the appendix of \cite{kennedy2019incremental} and the second part of Assumption \ref{asm:boundedness}, which says that the true CATE is bounded.
\end{proof}

\thmdrilearner*

\begin{proof}
This follows from Proposition 1 of \cite{kennedy2020towards}, the definition of $\widehat b (x)$, and by iterated expectation.  
\end{proof}

The bounded functions in $\widehat b(x)$ depend on the relevant estimand.  In general, 
\begin{align*}
    \widehat b (x) &= g_0(x) \Big\{ \pi(x) - \widehat \pi(x) \Big\} \left\{ \mu(0, x) - \widehat \mu(0, x) \right\} + g_1 (x) \Big\{ \pi(x) - \widehat \pi(x) \Big\} \Big\{ \mu(1, x) + \widehat \mu(1, x) \Big\} + h(x) \left\{ \pi(x) - \widehat \pi(x) \right\}^2
\end{align*}
where $g_0 (x), g_1(x), h(x)$ are all bounded functions.  Omitting arguments and letting $\mu_a = \mu(a, X)$, when $\tau_{i-dr} = \tau_{cie}$, then
\begin{align*}
    g_0 (x) \equiv g_0(x; \delta) &= \frac{1}{\delta \widehat \pi + 1 - \widehat \pi} \left( \frac{1 - \widehat \pi}{1 - \pi} - \frac{\delta}{\delta \pi + 1 - \pi} \right) \\
    g_1(x) \equiv g_1(x; \delta) &= \frac{1}{\delta \widehat \pi + 1 - \widehat \pi} \left( \frac{\delta}{\delta \pi + 1 - \pi} - \frac{\delta \widehat \pi}{\pi}\right) \\
    h(x) \equiv h(x; \delta) &= (\mu_1 - \mu_0) \left\{ \frac{\delta (1 - \delta)}{(\delta \pi + 1 - \pi)(\delta \widehat \pi+ 1 - \widehat \pi)^2} \right\}
\end{align*}
When $\tau_{i-dr} = \tau_{cice}$, then
\begin{align*}
    g_0 (x) \equiv g_0(x; \delta_u, \delta_l) &= g_0 (x; \delta_u) - g_0 (x; \delta_l) \\
    g_1(x) \equiv g_1(x; \delta_u, \delta_l) &= g_1(x; \delta_u) - g_1 (x; \delta_l) \\
    h(x) \equiv h(x; \delta_u, \delta_l) &= h(x; \delta_u) - h(x; \delta_l)
\end{align*}
where $g_0(x; \delta), g_1(x; \delta), h(x; \delta)$ are define above for the CIE.  And, when $\tau_{i-dr} = \tau_{cide}$, then
\begin{align*}
    g_0(x) &= \frac{\widehat \pi}{(\delta \widehat \pi + 1 - \widehat \pi)^2} + \left\{ \frac{(\delta + 1)(1 - \delta) \pi \overline \pi + \pi + \overline \pi - 1}{(\delta \pi + 1 - \pi)^2 (\delta \overline \pi + 1 - \overline \pi)^2} \right\} \\
    g_1(x) &= \frac{1 - \widehat \pi}{(\delta \widehat \pi + 1 - \widehat \pi)^2} + \left\{ \frac{(\delta + 1)(1 - \delta) \pi \overline \pi + \pi + \overline \pi - 1}{(\delta \pi + 1 - \pi)^2 (\delta \overline \pi + 1 - \overline \pi)^2} \right\} \\
    h(x) &= (\widehat \mu_1 - \widehat \mu_0) \cdot \left[ \frac{(1 - \delta - \delta^2 - \delta^3) \pi \widehat \pi + 1 - 2\delta + (1-\delta)\{ \widehat \pi + (1-\delta)\pi \}}{(\delta \widehat \pi + 1 - \widehat \pi)^3 (\delta \pi + 1 - \pi)^2} \right]
\end{align*}

\section{Proofs of Results in Section \ref{sec:heterogeneity}}

\lemvareif*
 
\begin{proof}
As in Lemma \ref{lem:cide_eif}, we prove this result by showing that the estimand admits a Von-Mises expansion where the second-order term is a product of errors.  

\medskip

Omitting arguments, let $\xi(\bbP) = 2 \omega \tau (\omega \varphi + \phi \tau) + (\omega \tau)^2.$  Then,
\begin{align*}
    R_2(\overline \bbP, \bbP) &= \int_{\mathcal{Z}} \xi(\overline \bbP) - \xi (\bbP) d \bbP \\
    &= \bbE \left\{ 2 \overline \omega \overline \tau (\overline \omega \overline \varphi + \overline \phi \overline \tau) + (\overline \omega \overline \tau)^2 \right\} - \bbE \left\{ 2 \omega \tau ( \omega \varphi + \phi \tau) + (\omega \tau)^2 \right\}  \\
    &= \bbE \left\{ 2 \overline \omega \overline \tau (\overline \omega \overline \varphi + \overline \phi \overline \tau) + (\overline \omega \overline \tau)^2 - (\omega \tau)^2 \right\} \\
    &= \bbE \left\{ 2 \overline \omega \overline \tau \left( \overline \omega \overline K + \overline \omega (\tau - \overline \tau) + \overline \phi \overline \tau \right) + (\overline \omega \overline \tau)^2 - (\omega \tau)^2 \right\} \\ 
    &= \bbE \left[ 2 \overline \omega \overline \tau \left\{ \overline \omega \overline K + (\overline \omega - \omega) (\tau - \overline \tau) + (\overline \phi - \omega) \overline \tau + \omega \tau \right\} + (\overline \omega \overline \tau)^2 - (\omega \tau)^2 \right] \\
    &= \bbE \left[ 2 \overline \omega \overline \tau \left\{ \overline \omega \overline K + (\overline \omega - \omega) (\tau - \overline \tau) + (\overline \phi + \overline \omega - \omega) \overline \tau - \overline \omega \overline \tau + \omega \tau \right\} + (\overline \omega \overline \tau)^2 - (\omega \tau)^2 \right] \\
    &= \bbE \left( 2 \overline \omega \overline \tau \left[ \overline \omega \overline K + (\overline \omega - \omega) (\tau - \overline \tau) + \overline \tau \Big\{  \bbE ( \overline \phi \mid X) + \overline \omega - \omega \Big\} \right] - 2 (\overline \omega \overline \tau)^2 + 2 \overline \omega \overline \tau \omega \tau + (\overline \omega \overline \tau)^2 - (\omega \tau)^2 \right] \\
    &= \bbE \left( 2 \overline \omega \overline \tau \left[ \overline \omega \overline K + (\overline \omega - \omega) (\tau - \overline \tau) + \overline \tau \Big\{  \bbE ( \overline \phi \mid X) + \overline \omega - \omega \Big\} \right] - \Big( \overline \omega \overline \tau - \omega \tau \Big)^2 \right] \\
    &= \bbE \left( 2 \overline \omega \overline \tau \left[ \overline \omega \overline K + (\overline \omega - \omega) (\tau - \overline \tau) + \overline \tau \Big\{  \bbE (\overline \phi \mid X) + \overline \omega - \omega \Big\} \right] - \Big\{ \overline \omega (\overline \tau - \tau) + \tau (\overline \omega - \omega) \Big\}^2 \right] \\
    &\lesssim \Big( \lVert \overline \mu - \mu \rVert + \lVert \overline \pi - \pi \rVert \Big)^2
\end{align*}
where
$$
\overline K = \sum_a (2a - 1) \left( \frac{\bbP_a - \overline \bbP_a}{\bbP_a} \right) (\mu_a - \overline \mu_a).
$$
and $\bbP_a = \bbP(A = a \mid X)$. The last line follows by similar logic to the post script to Lemma \ref{lem:cide_eif}. This shows that the second order term is a product of errors, and so $\xi(\bbP)$ is the un-centered efficient influence function for $\bbE \{ \tau_{cide}(X; \delta)^2 \}$.
\end{proof}

\thmvcideconv*

\begin{proof}
First, note that, by construction
\begin{align*}
    \widehat \psi_n &= \bbP_n \left\{ 2 \widehat \omega \widehat \tau (\widehat \omega \widehat \varphi + \widehat \phi \widehat \tau) + (\widehat \omega \widehat \tau)^2 - (\widehat \omega \widehat \tau + \widehat \omega \widehat \varphi + \widehat \phi \widehat \tau ) \bbP_n (\widehat \omega \widehat \tau + \widehat \omega \widehat \varphi + \widehat \phi \widehat \tau )  \right\} \\
    &\equiv \bbP_n (\widehat \xi_1) - \bbP_n(\widehat \xi_2)^2
\end{align*}
where the second line follows by defining
\begin{align*}
    \widehat \xi_1 &= 2 \widehat \omega \widehat \tau (\widehat \omega \widehat \varphi + \widehat \phi \widehat \tau) + (\widehat \omega \widehat \tau)^2,\text{ and} \\
    \widehat \xi_2 &= (\widehat \omega \widehat \tau + \widehat \omega \widehat \varphi + \widehat \phi \widehat \tau ).
\end{align*}
Similarly, by construction, $\bbV \{ \tau_{cide} (X; \delta) \} = \bbE (\xi_1) - \bbE(\xi_2)^2$ where
\begin{align*}
    \xi_1 &= 2 \omega \tau (\omega \varphi + \phi \tau) + (\omega \tau)^2,\text{ and} \\
    \xi_2 &= (\omega \tau + \omega \varphi + \phi \tau ).
\end{align*}
Therefore, if we define $\widehat \xi_3 = \widehat \xi_1 - \widehat \xi_2 \bbP_n (\widehat \xi_2)$ and $\xi_3 = \xi_1 - \xi_2 \bbE (\xi_2)$, we have
$$
\widehat \psi_n - \bbV \{ \tau_{cide}(X; \delta) \} = \bbP_n (\widehat \xi_1) - \bbP_n(\widehat \xi_2)^2 - \bbP( \xi_1) + \bbP(\xi_2)^2 = \bbP_n (\widehat \xi_3) - \bbP(\xi_3)
$$
Then, by adding and subtracting terms, we have the usual expansion
\begin{equation} \label{eq:decomp}
    \widehat \psi_n - \bbV \{ \tau_{cide}(X; \delta) \} = (\bbP_n - \bbP) \xi_3 + (\bbP_n - \bbP) (\widehat \xi_3 - \xi_3) + \bbP( \widehat \xi_3 - \xi_3)
\end{equation}
where the first term on the right hand side of the final equation will follow a central limit theorem, the second term is an empirical process term, and the third term is a bias term.

\medskip

Starting with the third term, we see
\begin{align*}
    \bbP (\widehat \xi_3 - \xi_3) &= \bbP (\widehat \xi_1 - \xi_1) + \bbP \big\{ \widehat \xi_2 \bbP_n (\widehat \xi_2) - \xi_2 \bbP( \xi_2) \big\} \\
    &\lesssim \Big( \lVert \widehat \pi - \pi \rVert + \lVert \widehat \mu - \mu \rVert \Big)^2 + \bbP \big\{ \widehat \xi_2 \bbP_n (\widehat \xi_2) - \xi_2 \bbP( \xi_2) \big\}
\end{align*}
where the second line follows by the proof of Lemma \ref{lem:var_eif}. For the second term on the final line above, we see that
\begin{align*}
    \bbP( \widehat \xi_2 \bbP_n (\widehat \xi_2) - \xi_2 \bbP( \xi_2)) &= \left[ \bbE\left\{ (\widehat \omega \widehat \tau + \widehat \omega \widehat \varphi + \widehat \phi \widehat \tau ) \bbP_n (\widehat \omega \widehat \tau + \widehat \omega \widehat \varphi + \widehat \phi \widehat \tau ) \right\} - \bbE \left\{ (\omega \tau + \omega \varphi + \phi \tau ) \bbE (\omega \tau + \omega \varphi + \phi \tau ) \right\} \right] \\
    &= \left[ \bbE\left\{ (\widehat \omega \widehat \tau + \widehat \omega \widehat \varphi + \widehat \phi \widehat \tau ) \bbP_n (\widehat \omega \widehat \tau + \widehat \omega \widehat \varphi + \widehat \phi \widehat \tau ) \right\} - \bbE ( \omega \tau )^2 \right] \\
    &= \frac1n \left[ \bbE \left\{ (\widehat \omega \widehat \tau + \widehat \omega \widehat \varphi + \widehat \phi \widehat \tau )^2 \right\} - \bbE \left( \widehat \omega \widehat \tau + \widehat \omega \widehat \varphi + \widehat \phi \widehat \tau \right)^2 \right] + 2 \bbE \left( \widehat \omega \widehat \tau + \widehat \omega \widehat \varphi + \widehat \phi \widehat \tau  \right)^2 - 2 \bbE (\omega \tau)^2 \\
    &= o_{\bbP}(1/n) +  \bbE \left\{ (\widehat \omega \widehat \tau + \widehat \omega \widehat \varphi + \widehat \phi \widehat \tau ) \right\}^2 - \bbE (\omega \tau)^2 \\
    &= o_{\bbP}(1/n) + \bbE \left(\widehat \omega \widehat \tau + \widehat \omega \widehat \varphi + \widehat \phi \widehat \tau + \omega \tau \right) \bbE \left( \widehat \omega \widehat \tau + \widehat \omega \widehat \varphi + \widehat \phi \widehat \tau - \omega \tau \right)  \\
    &\lesssim o_{\bbP}(1/n) + \Big( \lVert \widehat \pi - \pi \rVert + \lVert \widehat \mu - \mu \rVert \Big) \lVert \widehat \pi - \pi \rVert
\end{align*}
where the fourth line follows by Assumption \ref{asmp:bounded_eif}, the fifth line because $a^2 - b^2 = (a + b)(a-b)$, the sixth by Assumption \ref{asmp:bounded_eif}, which implies $\omega \tau$ is bounded, and by the proof of Lemma \ref{lem:cide_eif}.  Therefore, if $\Big( \lVert \widehat \pi - \pi \rVert + \lVert \widehat \mu - \mu \rVert \Big)^2 = o_{\bbP} (n^{-1/2})$ then $\sqrt{n} \bbP( \widehat \xi_3 - \xi_3)  = o_{\bbP}(1)$.

\medskip

The second term from eq. \eqref{eq:decomp}, the empirical process term, is simpler to bound.  By Lemma 2 of \cite{kennedy2020sharp} and by sample splitting, we have
$$
(\bbP_n - \bbP) (\widehat \xi_3 - \xi_3) = O_{\bbP} \left( \frac{\lVert \widehat \xi_3 - \xi_3 \rVert }{n} \right). 
$$
By the triangle inequality
\begin{align*}
    \lVert \widehat \xi_3 - \xi_3 \rVert &\leq \lVert \widehat \xi_1 - \xi_1 \rVert + \lVert \widehat \xi_2 \bbP_n (\widehat \xi_2) - \xi_2 \bbP(\xi_2) \rVert \\
    &= \lVert \widehat \xi_1 - \xi_1 \rVert + \lVert \bbP_n (\widehat \xi_2) (\widehat \xi_2  - \xi_2) + \xi_2 \big\{ \bbP_n (\widehat \xi_2) -  \bbP(\xi_2) \big\} \rVert \\
    &\lesssim \lVert \widehat \xi_1 - \xi_1 \rVert + \lVert \widehat \xi_2 - \xi_2 \rVert + \lVert \bbP_n (\widehat \xi_2) -  \bbP(\xi_2) \rVert
\end{align*}
where the last line follows by assumption \ref{asmp:bounded_eif}, which says that $\xi_2$ and $\widehat \xi_2$ are bounded.  All three terms are $o_{\bbP}(1)$ since $\Big( \lVert \widehat \pi - \pi \rVert + \lVert \widehat \mu - \mu \rVert \Big)^2 = o_{\bbP} (n^{-1/2})$.  Thus,
$$
\psi_n - \bbV \{ \tau_{cide} (X; \delta) \} = (\bbP_n - \bbP)(\xi_3) + o_{\bbP} (n^{-1/2})
$$
Therefore, by the central limit theorem,
$$
\sqrt{n} \Big[ \psi_n - \bbV \{ \tau_{cide} (X; \delta) \} \Big] \indist N( 0, \bbV (\xi_3) ) = N(0, \sigma^2)
$$
with $\sigma^2$ as defined in the theorem.
\end{proof}

\begin{restatable}{proposition}{proplinear} \label{prop:linear}
Let 
\begin{align*}
    \widehat \psi_{n,1} &= \bbP_n \left[ 2 \widehat \omega \big( \widehat \mu_1 - \widehat \mu_0 \big) \left\{ \widehat \omega \widehat \varphi + \widehat \phi \big( \widehat \mu_1 - \widehat \mu_0 \big) \right\} + \left\{ \widehat \omega \big( \widehat \mu_1 - \widehat \mu_0 \big) \right\}^2 \right], \text{ and } \\
    \widehat \psi_{n,2} &= \bbP_n \Big\{ \widehat \omega \widehat \varphi + \widehat \phi \big( \widehat \mu_1 - \widehat \mu_0 \big) + \widehat \omega \big( \widehat \mu_1 - \widehat \mu_0 \big) \Big\}^2.
\end{align*}
Under Assumptions \ref{asmp:cons} and \ref{asmp:exch}, Assumption \ref{asm:boundedness} from Theorem \ref{thm:fixed_mod}, and Assumption \ref{asmp:bounded_eif} from Theorem \ref{thm:vcide_conv}.  If
$$
\Big( \lVert \widehat \pi - \pi \lVert + \lVert \widehat \mu - \mu \rVert \Big)^2 = o_\bbP \left( \frac{1}{\sqrt n} \right),
$$
then
\begin{align}
    \sqrt{n} \Big[ \widehat \psi_{n,1} - \bbE \{ \tau_{cide}(X; \delta)^2 \} \Big] &\indist N \Big( 0, \bbV ( \zeta_1) \Big), \text{and} \\
    \sqrt{n} \Big[ \widehat \psi_{n,2} - \bbE \{ \tau_{cide}(X; \delta) \}^2 \Big] &\indist N \Big( 0, 4 \bbV ( \zeta_2) \Big)
\end{align}
where
\begin{align}
    \zeta_1 &= 2 \omega \big( \mu_1 - \mu_0 \big) \left\{ \omega \varphi + \phi \big( \mu_1 - \mu_0 \big) \right\} + \left\{ \omega \big( \mu_1 - \mu_0 \big) \right\}^2, \text{and} \label{eq:zeta_1} \\
    \zeta_2 &= \bbE \Big\{ \omega \varphi + \phi \big( \mu_1 - \mu_0 \big) + \omega \big( \mu_1 - \mu_0 \big) \Big\} \cdot \Big\{ \omega \varphi + \phi \big( \mu_1 - \mu_0 \big) + \omega \big( \mu_1 - \mu_0 \big) \Big\} \label{eq:zeta_2}.
\end{align}
\end{restatable}
\begin{proof}
This result follows by Lemmas \ref{lem:cide_eif} and \ref{lem:var_eif}, the conditions of the Proposition, and the Delta method (for the second convergence result).
\end{proof}

\begin{restatable}{proposition}{propconsvar}\label{prop:cons_var}
Let $\zeta_1, \zeta_2$ be as defined in equations \eqref{eq:zeta_1} and \eqref{eq:zeta_2}. Then, when $\bbV \{ \tau_{cide} (X; \delta) \} = 0$,
\begin{equation} \label{eq:pos_cov}
    \cov (\zeta_1, \zeta_2) \geq 0,
\end{equation}
so that
\begin{equation} \label{eq:var_small}
    \bbV (\zeta_1 - \zeta_2) = \bbV (\zeta_1) + \bbV (\zeta_2) - 2 \cov (\zeta_1, \zeta_2) \leq \bbV (\zeta_1) + \bbV (\zeta_2).
\end{equation}
\end{restatable}

\begin{proof}
By the assumption that $\bbV \{ \tau_{cide}(X; \delta) \} = 0$, it follows that $\tau_{cide} (X; \delta) \equiv \omega (\mu_1 - \mu_0) = C$ for some constant $C$.  As a reminder, 
\begin{align*}
    \zeta_1 &= 2 \omega \big( \mu_1 - \mu_0 \big) \left\{ \omega \varphi + \phi \big( \mu_1 - \mu_0 \big) \right\} + \left\{ \omega \big( \mu_1 - \mu_0 \big) \right\}^2, \text{ and} \label{eq:zeta_1} \\
    \zeta_2 &= \bbE \Big\{ \omega \varphi + \phi \big( \mu_1 - \mu_0 \big) + \omega \big( \mu_1 - \mu_0 \big) \Big\} \cdot \Big\{ \omega \varphi + \phi \big( \mu_1 - \mu_0 \big) + \omega \big( \mu_1 - \mu_0 \big) \Big\}.
\end{align*}
Therefore, when $\omega (\mu_1 - \mu_0) = C$, then
\begin{align*}
    \zeta_1 &= 2 C \left\{ \omega \varphi + \phi \big( \mu_1 - \mu_0 \big) \right\} + C^2, \text{ and } \\
    \zeta_2 &= \bbE \{ \omega \varphi + \phi (\mu_1 - \mu_0) + C \} \{ \omega \varphi + \phi (\mu_1 - \mu_0) + C \} \\
    &= \bbE \{ C \}  \{ \omega \varphi + \phi (\mu_1 - \mu_0) + C \}  \\
    &= C \{ \omega \varphi + \phi (\mu_1 - \mu_0) \} + C^2.
\end{align*}
And so,
$$
2 \cov (\zeta_1, \zeta_2) = 2C^2 \bbV \{ \omega \varphi + \phi (\mu_1 - \mu_0) \} \geq 0,
$$
which implies the result.
\end{proof}

\propcoverage*

\begin{proof}
By definition,
$$
\widehat \psi_n = \widehat \psi_{n,1} - \widehat \psi_{n,2},
$$
where $\widehat \psi_{n,1}$ and $\widehat \psi_{n,2}$ are defined in Proposition \ref{prop:linear}.  By Proposition \ref{prop:linear}, $\widehat \psi_{n,1}$ and $\widehat \psi_{n,2}$ converge to non-degenerate distributions,
\begin{align*}
    \sqrt{n} \Big[ \widehat \psi_{n,1} - \bbE \{ \tau_{cide}(X; \delta)^2 \} \Big] &\indist N \Big( 0, \bbV ( \zeta_1) \Big), \text{ and} \\
    \sqrt{n} \Big[ \widehat \psi_{n,2} - \bbE \{ \tau_{cide}(X; \delta) \}^2 \Big] &\indist N \Big( 0, 4 \bbV ( \zeta_2) \Big),
\end{align*}
where $\zeta_1$ and $\zeta_2$ are defined in equations \eqref{eq:zeta_1} and \eqref{eq:zeta_2}.  And so, 
$$
\sqrt{n} \Big[ \widehat \psi_n - \bbV \{ \tau_{cide}(X; \delta) \} \Big] \indist N \Big( 0, \bbV(\zeta_1) + 4\bbV(\zeta_2) - 4\cov( \zeta_1, \zeta_2) \Big).
$$
By Proposition \ref{prop:cons_var}, when $\bbV \{ \tau_{cide}(X; \delta) \} = 0$, 
$$
\bbV(\zeta_1) + 4\bbV(\zeta_2) - 4\cov( \zeta_1, \zeta_2) \leq \bbV(\zeta_1) + 4\bbV(\zeta_2).
$$
Therefore, since $\widehat \sigma_1^2$ and $\widehat \sigma_2^2$ in equations \eqref{eq:hatsigma1} and \eqref{eq:hatsigma2} are consistent estimators for $\bbV(\zeta_1)$ and $4\bbV(\zeta_2)$ respectively by the weak law of large numbers, $\widehat \sigma_1^2 + \widehat \sigma_2^2$ is a consistent estimator for $\bbV(\zeta_1) + 4\bbV(\zeta_2)$.  And so, when $\bbV \{ \tau_{cide}(X; \delta) \} = 0$, 
$$
\lim_{n \to \infty} \mathbb{P} \left( \sqrt{n} \widehat \psi_n \leq \Phi^{-1}(1-\alpha) \sqrt{\widehat \sigma_1^2 + \widehat \sigma_2^2} \right) \leq \alpha,
$$
which implies that the asymptotic Type I error of the test in \eqref{eq:test} is less that or equal to $\alpha$.
\end{proof}

\end{document}